\pdfoutput=1

\documentclass[12pt,english]{article}
\usepackage[T1]{fontenc}
\usepackage[latin9]{inputenc}
\usepackage{geometry}
\usepackage{babel}
\usepackage{array}
\usepackage{multirow}
\usepackage{amsthm}
\usepackage{amsmath}
\usepackage{amssymb}
\usepackage{setspace}
\usepackage{color}
\usepackage{soul}
\usepackage{framed}
\usepackage{rotating}
\usepackage{graphicx}
\usepackage[authoryear,round]{natbib}
\usepackage[section]{placeins}
\usepackage[normalem]{ulem}
\usepackage[unicode=true,pdfusetitle,
 bookmarks=true,bookmarksnumbered=false,bookmarksopen=false,
 breaklinks=false,pdfborder={0 0 1},backref=false,colorlinks=false]{hyperref}
\usepackage{array}
\usepackage{versions}
\usepackage{xcolor}
\usepackage{pdfpages}

\definecolor{bluez}{rgb}{0,0,0}
\PassOptionsToPackage{normalem}{ulem}
\makeatletter

\newcolumntype{L}[1]{>{\raggedright\let\newline\\\arraybackslash\hspace{0pt}}m{#1}}
\newcolumntype{C}[1]{>{\centering\let\newline\\\arraybackslash\hspace{0pt}}m{#1}}
\newcolumntype{R}[1]{>{\raggedleft\let\newline\\\arraybackslash\hspace{0pt}}m{#1}}
\theoremstyle{plain}
\newtheorem{thm}{\protect\theoremname}
\newtheorem{prop}{Proposition}

\newtheorem{cor}{Corollary}
\newtheorem{lem}{Lemma}
\makeatother
\providecommand{\theoremname}{Theorem}
\begin{document}

\title{\textbf{Identifying Network Ties from Panel Data: Theory and an
Application to Tax Competition}\thanks{{\scriptsize We gratefully
acknowledge financial support from the ESRC through the Centre for the
Microeconomic Analysis of Public Policy (ES/T014334/1), the Centre for
Microdata Methods and Practice (RES-589-28-0001) and the Large Research
Grant ES/P008909/1 and from the ERC (SG338187). We thank Edo Airoldi, Luis
Alvarez, Michele Aquaro, Oriana Bandiera, Larry Blume, Yann Bramoull\'{e},
Stephane Bonhomme, Vasco Carvalho, Gary Chamberlain, Andrew Chesher,
Christian Dustmann, S\'{e}rgio Firpo, Jean-Pierre Florens, Eric Gautier, Giacomo
de Giorgi, Matthew Gentzkow, Stefan Hoderlein, Bo Honor\'{e}, Matt Jackson, Dale
Jorgensen, Christian Julliard, Maximilian Kasy, Miles Kimball, Thibaut
Lamadon, Simon Sokbae Lee, Arthur Lewbel, Tong Li, Xiadong Liu, Elena
Manresa, Charles Manski, Marcelo Medeiros, Angelo Mele, Francesca Molinari,
Pepe Montiel, Andrea Moro, Whitney Newey, Ariel Pakes, Eleonora Pattachini,
Michele Pelizzari, Martin Pesendorfer, Christiern Rose, Adam Rosen, Bernard
Salanie, Olivier Scaillet, Sebastien Siegloch, Pasquale Schiraldi, Tymon
Sloczynski, Kevin Song, John Sutton, Adam Szeidl, Thiago Tachibana, Elie
Tamer, and seminar and conference participants for valuable comments. We
also thank Tim Besley and Anne Case for comments and sharing data. Daniel
Barbosa provided outstanding research assistance. A previous version of
this paper was circulated as \textquotedblleft Recovering Social Networks
from Panel Data: Identification, Simulations and an
Application.\textquotedblright\ All errors remain our own. Codes are available on Zenodo (\texttt{https://zenodo.org/XXX}) and github (\texttt{https://github.com/YYY}) repositories.
}}}
\author{\'{A}ureo de Paula \and Imran Rasul \and Pedro CL Souza\thanks{%
{\scriptsize de Paula: University College London, CeMMAP and IFS,
a.paula@ucl.ac.uk; Rasul: University College London and IFS,
i.rasul@ucl.ac.uk; Souza: Queen Mary University, p.souza@qmul.ac.uk.}}.}
\date{October 2023}

\maketitle

\begin{abstract}
\noindent Social interactions determine many economic behaviors, but
information on social ties does not exist in most publicly available and
widely used datasets. We present results on the identification of social
networks from observational panel data that contains no information on
social ties between agents. In the context of a canonical social
interactions model, we provide sufficient conditions under which the social
interactions matrix, endogenous and exogenous social effect parameters are
globally identified if networks are constant over time. We also provide an
extension of the method for time-varying networks. We then describe how
high-dimensional estimation techniques can be used to estimate the
interactions model based on the Adaptive Elastic Net Generalized Method of
Moments. We employ the method to study tax competition across US states. The identified social interactions matrix implies that tax competition
differs markedly from the common assumption of competition between
geographically neighboring states, providing further insights into the
long-standing debate on the relative roles of factor mobility and yardstick
competition in driving tax setting behavior across states. Most broadly, our
identification and application show that the analysis of social interactions
can be extended to economic realms where no network data exists.\newline
\textit{JEL Classification: C31, D85, H71.}$\smallskip \smallskip \smallskip
\smallskip \smallskip \smallskip $
\end{abstract}

\section{Introduction}

In many economic environments, behavior is shaped by social interactions
between agents. In individual decision problems, social interactions have
been key to understanding outcomes as diverse as educational test scores,
the demand for financial assets, and technology adoption (%
\citealp{Sacerdote2001}; \citealp{Bursztynetal2014}; \citealp{ConleyUdry2010}%
). In macroeconomics, the structure of firms' production and credit networks
propagate shocks, or help firms to learn (\citealp{Acemogluetal2012}; %
\citealp{Chaney2014}). In political economy and public economics, ties between jurisdictions are
key to understanding tax setting behavior (\citealp{Tiebout1956}; %
\citealp{Shleifer1985}; \citealp{BesleyCase1994}).

Underpinning all these bodies of research is some measurement of the
underlying social ties between agents. However, information on social ties
does not exist in most publicly available and widely used datasets. To
overcome this limitation, studies of social interaction either \emph{%
postulate} ties based on common observables or homophily, or \emph{elicit}
data on networks. However, it is increasingly recognized that postulated and
elicited networks remain imperfect solutions to the fundamental problem of
missing data on social ties, because of econometric concerns that arise with
either method, or simply because of the cost of collecting network data.%
\footnote{%
As detailed in \citet{DePaula2017}, elicited networks are often
self-reported and can introduce error to the outcome of interest. Network
data can be censored if only a limited number of links can feasibly be
reported. Incomplete survey coverage of nodes in a network may lead to
biased aggregate network statistics. \citet{ChandrasekharLewis2016} show
that even when nodes are randomly sampled from a network, partial sampling
leads to non-classical measurement error and biased estimation. Collecting
social network data is also a time- and resource-intensive process. In
response to these concerns, a nascent strand of literature explores
cost-effective alternatives to full elicitation to recover aggregate network
statistics (\citealp{Brezaetal2017}).}

Two consequences are that (i) the classes of problems in which social
interactions occur are understudied, because social networks data is missing
or too costly to collect; and (ii) there is no way to validate social
interactions analysis in contexts where ties are postulated. In this paper,
we tackle this challenge by deriving sufficient conditions under which
global identification of the \emph{entire structure} of social networks\ is
obtained, using only observational panel data that itself contains \emph{no}
information on network ties. Our identification results allow the study of
social interactions without data on social networks, and the validation of
structures of social interaction where social ties have hitherto been
postulated. The recovered networks are economically meaningful to explain
the effects under study, since they are entirely estimated from the data
itself, and not driven by \emph{ex ante} assumptions on how individuals
interact.

A researcher is assumed to have panel data on individuals $i=1,...,N$ for
instances $t=1,...,T$. An instance refers to a specific observation for $i$
and need not correspond to a time period (for example, if $i$ refers to a
firm, $t$ could refer to market $t$). The outcome of interest for individual 
$i$ in instance $t$ is $y_{it}$ and is generated according to a canonical
structural model of social interactions:\footnote{\citet{Blumeetal2015}
present micro-foundations for this estimating equation based on
non-cooperative games of incomplete information for individual choice
problems.}%
\begin{equation}
y_{it}=\rho _{0}\sum_{j=1}^{N}W_{0,ij}y_{jt}+\beta _{0}x_{it}+\gamma
_{0}\sum_{j=1}^{N}W_{0,ij}x_{jt}+\alpha _{i}+\alpha _{t}+\epsilon _{it}.
\label{model motivation}
\end{equation}%
Outcome $y_{it}$ depends on the outcomes of other individuals to whom $i$ is
socially tied, $y_{jt}$, and $x_{jt}$ includes characteristics of those
individuals.\footnote{%
In the case in which $t$ is considered to be a time period, $x_{it}$ may
also include lagged values of $y_{it}$.} $W_{0,ij}$ measures how the outcome
and characteristics of $j$ causally impact the outcome for $i$. The network
is initially assumed to be fixed over time, and we later provide an extension of
the method for time-varying networks. As outcomes for all individuals obey
equations analogous to (\ref{model motivation}), the system of equations can
be written in matrix notation, where the structure of interactions is
captured by the adjacency matrix, denoted by $W_{0}$. Our approach allows
for unobserved heterogeneity across individuals $\alpha _{i}$ and common
shocks to individuals $\alpha _{t}$. This framework encompasses\ a classic
linear-in-means specification as in \citet{Manski1993}. In his terminology, $%
\rho _{0}$ and $\gamma _{0}$ capture endogenous and exogenous social
effects, and $\alpha _{t}$ captures correlated effects. The distinction
between endogenous and exogenous peer effects is critical, as only the
former generates social multiplier effects. In line with the literature, we
maintain that the same $W_{0}$ governs the structure of both endogenous and
exogenous effects. We later discuss relaxing this assumption when more than
one regressor is used.

Manski's seminal contribution set out the reflection problem of separately
identifying endogenous, exogenous, and correlated effects in linear models.
However, it has been somewhat overlooked that he also set out another
challenge in the identification of the social network in the first place.%
\footnote{\citet{Manski1993} highlights difficulties (and potential
restrictions) in identifying $\rho _{0},\beta _{0}$ and $\gamma _{0}$ when
\emph{all} individuals interact with each other, and when this is observed by the
researcher. In (\ref{model motivation}), this corresponds to $%
W_{0,ij}=N^{-1} $, for $i,j=1,\dots ,N$. At the same time, he states (p.
536), \textquotedblleft I have presumed that researchers know how
individuals form reference groups and that individuals correctly perceive
the mean outcomes experienced by their supposed reference groups. There is
substantial reason to question these assumptions (...) If researchers do not
know how individuals form reference groups and perceive reference-group
outcomes, then it is reasonable to ask whether observed behavior can be used
to infer these unknowns (...) The conclusion to be drawn is that informed
specification of reference groups is a necessary prelude to analysis of
social effects.\textquotedblright} This is the problem we tackle, and thus, we
expand the scope of identification beyond $\rho _{0}$, $\beta_{0}$, and $\gamma_{0}$. Our point of departure from much of the literature is 
therefore to presume $W_{0}$ is \emph{entirely unknown} to the researcher. We
derive sufficient conditions under which all the entries in $W_{0}$, and the
endogenous and exogenous social effect parameters, $\rho _{0}$ and $\gamma
_{0},$ are globally identified from \textquotedblleft reduced
form\textquotedblright\ parameters. By identifying the social interactions
matrix $W_{0}$, our results allow the recovery of aggregate network
characteristics, such as the degree distribution and patterns of homophily,
as well as node-level statistics such as the strength of social interactions
between nodes,   and the centrality of nodes. Such aggregate and node-level
statistics often map back to underlying models of social interaction (%
\citealp{Ballesteretal2006}; \citealp{Jacksonetal2017}; \citealp{DePaula2017}%
).

Our identification strategy is new and fundamentally different from those
employed elsewhere in the literature and does not rely on requirements about
network sparsity. However, it delivers sufficient conditions that are mild
and relate to existing results on the identification of social effects
parameters when $W_{0}$ is known (\citealp{Bramoulleetal2009}; %
\citealp{DeGiorgietal2010}; \citealp{Blumeetal2015}). The intuition for our
identification result is simple: model (1) has $N^{2}$ reduced-form
parameters, and there are $N(N-1)+3$ structural unknowns (as no unit affects
itself, so $W_{0,ii}=0$). So there are more equations than unknowns if $%
N\geq 2$, and we demonstrate those can be solved for the parameters of
interest under the assumptions we invoke. Our identification result is also
useful in other estimation contexts, such as when a researcher has partial
knowledge of $W_{0}$,\footnote{%
One such example is the nascent literature of Aggregate Relational Data
(ARD) as in \citet{Brezaetal2017}. Another possibility is that individuals
are known to belong to subgroups, so $W_{0}$ is block diagonal. 
} or in navigating between priors on reduced-form and structural parameters
in a Bayesian framework (see, e.g., \citealp{gefanghalltavlas2023}), thus avoiding issues the raised by %
\citet{KlineTamer2016}.

Global identification is a necessary requirement for consistency of extremum
estimators such as those based on the GMM (Hansen 1982; Newey and McFadden
1994). Our identification analysis provides primitives for this condition.
To estimate the model, we employ the Adaptive Elastic Net GMM method (%
\citealp{CanerZhang2014}), as this allows us to deal with a potentially
high-dimensional parameter vector (in comparison to the time dimension in
the data) including all the entries of the social interactions matrix $W_{0}$%
, although other estimation protocols may also be entertained (e.g. using
Bayesian methods or \emph{a priori} information).\footnote{The Elastic Net was introduced by \citet{zou2005regularization} in part to circumvent difficulties faced by alternative estimation protocols (e.g., LASSO) when the number of parameters, $p$, exceeds the number of observations, $n$ (where $p$ and $n$ follow the notation in that paper).  Whereas the theoretical results on the large-sample properties of elastic net estimators usually have not exploited sparsity, several articles have demonstrated their performance in data scenarios where this occurs. 
 In Section 3, we provide an informal discussion on the performance in our context. \label{footnote:elasticnet}}

We showcase the method using Monte Carlo simulations based on stylized
random network structures as well as real-world networks. In each case, we
take a fixed network structure $W_{0}$ and simulate panel data as if the
data generating process were given by (\ref{model motivation}). We then
apply the method to the simulated panel data to recover estimates of all
elements in $W_{0}$, as well as the endogenous and exogenous social effect
parameters ($\rho _{0}$, $\gamma _{0}$). The networks considered vary in
size, complexity, and their aggregate and node-level features. In small
samples, we find that the majority of links are identified even for $T=5$,
and the proportion of true non-links (zeros in $W_{0}$) captured correctly
as zeros is over $85$\% even when $T=5$. Of course, there are important
limitations to the use of the method in small-$T$ cases. Biases are expected
and manifest themselves in two ways. First, weak links can be shrunk to
zero, and the strength of strong edges can be overestimated. Second, the estimates of $\rho$ and $\gamma$ can
suffer from small-sample bias, being analogous to well-known results for
autoregressive time series models. Both properties rapidly improve with $T$.
For instance, biases in the estimation of endogenous and exogenous effects
parameters ($\hat{\rho},$ $\hat{\gamma}$) fall quickly with $T$ and are
close to zero for large sample sizes. The endogenous and exogenous social
effects are also correctly captured as $T$ increases. \emph{A fortiori}, we
estimate aggregate and node-level statistics of each network, demonstrating
the accurate recovery of key players in networks, for example.

In the final part of our analysis, we apply the method to shed new light on
a classic real-world social interactions problem:\ tax competition between
US\ states. The literatures in political economy and public economics have
long recognized the behavior of state governors might be influenced by
decisions made in ``neighboring'' states. The typical empirical approach has
been to postulate the relevant neighbors as being geographically contiguous
states. Our approach allows us to infer the set of ``economic'' neighbors
determining social interactions in tax setting behavior from panel data on
outcomes and covariates alone. In this application, the panel data
dimensions cover mainland US states, $N=48$, for the years 1962-2015, $T=53$.

The identified network structure of tax competition differs markedly from
the assumption of competition between geographic neighbors. The identified
economic network has fewer edges, and we identify non-adjacent states that
influence tax setting behaviors. Differences in the structure of the identified
economic and geography-based networks are reflected in\ the far lower
clustering coefficient in the former ($.042$ versus $.419$). With the
recovered social interactions matrix we establish, beyond geography, which
covariates correlate to the existence of ties between states and so shed new
light on hypotheses for social interactions in tax setting: factor mobility
and yardstick competition (\citealp{Tiebout1956}; \citealp{Shleifer1985}; %
\citealp{BesleyCase1994}). The identified network highlights significant
predictors of tax competition between states beyond distance:\ political
homophily \emph{reduces} the likelihood of a link, suggesting any yardstick
competition driving social interactions occurs when voters compare their
governor to those of the opposing party in other states. Tax haven states
appear to be less influential in tax setting behaviors, easing concerns
over a race-to-the-bottom in tax setting. Labor mobility between states does
not robustly predict the existence of economic ties between states in tax
setting behavior.

Given the relatively long study period in this application, at a final stage
of analysis we extend our method to allow the strength of social
interactions in tax competition ($\rho _{0}$, $\gamma _{0}$) and the
structure of links in the economic network ($W_{0}$) to vary over time as we
change the weight placed on observations from any given time period. We
document the gradual increase in strength of social interactions over time,
and the changing nature of the network of interactions. We utilize these
findings to conduct counterfactual simulations of the general equilibrium
propagation of tax shocks from a given state to all other mainland US\
states, and how these general equilibrium effects of the same policy shock
vary as we place weight on observations later in our study period.

Our paper contributes to the literature on the identification of social
interactions models. The first generation of papers studied the case where $%
W_{0}$ is known, so only the endogenous and exogenous social effects
parameters needed to be identified. It is now established that if the known $%
W_{0}$ differs from the linear-in-means example where all units are linked
with equal weights, $\rho _{0}$ and $\gamma _{0}$ can be identified (%
\citealp{Bramoulleetal2009}; \citealp{DeGiorgietal2010}). Intuitively,
identification in those cases can use peers-of-peers, are not
necessarily connected to individual $i$ and can be used to leverage
variation from exclusion restrictions in (\ref{model motivation}), or can
use groups of different sizes within which all individuals interact with
each other (\citealp{Lee2007}). \citet{Bramoulleetal2009} show these
conditions are met if $I,$ $W_{0}$, and $W_{0}^{2}$ are linearly independent,
which is shown to hold generically by \citet{Blumeetal2015}. However, as
made precise in Section \ref{sec:Identification}, the linear algebraic
arguments employed by \citet{Bramoulleetal2009} or \citet{Blumeetal2015} do
not apply when $W_{0}$ is unobserved, and other arguments have to be used
instead.\footnote{%
Alternative identification approaches when $W_{0}$ is known focus on higher
moments (variances and covariances across individuals) of outcomes (%
\citealp{DePaula2017}) and rely on additional restrictions on higher
moments of $\epsilon _{it}$. Note that (\ref{model motivation}) is a spatial
autoregressive model. In that literature, $W_{0}$ is also typically assumed to be 
known (\citealp{Anselin2010}).}

\citet{Blumeetal2015} investigate the case when $W_{0}$ is \emph{partially}
observed and show that if
two individuals are \emph{known} not to be directly connected, the
parameters of interest in a model related to (\ref{model motivation}) can be
identified. Blume \emph{et al.} (2011) take an alternative approach: suggesting a parameterization of $W_{0}$ according to a
pre-specified distance between nodes. We do not impose such restrictions,
but note that partial observability of $W_{0}$ or placing additional structure on $W_{0}$ is complementary to our approach, as it reduces the number of
parameters in $W_{0}$ to be retrieved. \citet{BonaldiHortacsuKastl2015} and %
\citet{Manresa2016} estimate models like (\ref{model motivation}) when $%
W_{0} $ is not observed, but where $\rho _{0}$ is set to zero so there are
no endogenous social effects. They use sparsity-inducing methods from the
statistics literature, but the presence of $\rho _{0}$ in our case
complicates identification because it introduces issues of simultaneity that
we address.\footnote{\citet{Manresa2016} allows for unit-specific $\beta
_{0} $ parameters. While in many applications those are taken to be
homogeneous, we also discuss extensions on how heterogeneity in those
parameters can be handled when $\rho _{0}\neq 0$ in Appendix B.}

\citet{Rose2015} also presents related identification results for linear
models like (\ref{model motivation}), assuming the sparsity of the neighborhood
structure. Intuitively,
given two observationally equivalent systems, sparsity guarantees the
existence of pairs that are not connected in either. Since observationally
equivalent systems are linked via the reduced-form coefficient matrix, this
pair allows one to identify certain parameters in the model. Having
identified those parameters, \citet{Rose2015} shows that one can proceed to
identify other aspects of the structure (see also \citealp{GautierRose2016}%
). This is related to the ideas in Blume \emph{et al.} (2015),
who show identification results can be leveraged if individuals are \emph{%
known} not to be connected. Our main identification results do $not$ rely on
properties of sparse networks, and make use of plausible and intuitive
conditions, whereas the auxiliary rank conditions necessary may be computationally complex to verify. 
\textcolor{black}{More recently,
\citet{Lewbeletal2019} propose an estimation strategy for the parameters
$\rho_0$, $\beta_0$, and $\gamma_0$ of model \eqref{model motivation} in the
absence of network links if many different groups can be observed.} 
\textcolor{black}{\citet{Battaglinietal2019} estimate a structural model
specifically for the case of unobserved social connections in the US
Congress.}

Finally, in the statistics literature, \citet{LamSouza2019} study the
penalized estimation of model \eqref{model motivation} when $W_{0}$ is not
observed, assuming the model and social interactions are identified. The
statistical literature on graphical models has investigated the estimation
of neighborhoods defined by the covariance structure of the random variables
at hand (\citealp{MeinshausenBuhlmann2006}). This corresponds to a model
where $y_{t}=(I-\rho _{0}W_{0})^{-1}\epsilon _{t}$ is jointly normal
(abstracting from covariates). On a graph with $N$ nodes corresponding to
the variables in the model, an edge between two nodes (variables) $i$ and $j$
is absent when these two variables are conditionally independent given the
other nodes. In the
model above, the inverse covariance matrix is $(I-\rho _{0}W_{0})^{\top
}\Sigma _{\epsilon }^{-1}(I-\rho _{0}W_{0})$, where $\Sigma _{\epsilon }$ is
the variance covariance structure for $\epsilon _{t}$. The discovery of zero
entries in this matrix is not equivalent to the identification of $W_{0}$
and involves $\Sigma _{\epsilon }$ (as do identification strategies using
higher moments when $W_{0}$ is known).\footnote{%
\citet{MeinshausenBuhlmann2006}'s and \citet{LamSouza2019}'s neighborhood
estimates rely on (penalized) regressions of $y_{it}$ on $y_{1t},\dots
,y_{i-1,t},y_{i+1,t},\dots ,y_{N,t}$, which do not address the endogeneity
in estimating $W_{0}$.} 

We build on these papers by studying the problem where $W_{0}$ is
potentially entirely unknown to the researcher. In so doing, we open up the
study of social interactions to realms where social network data does not
exist. In our case, we consider the definition of the network as the one
that mediates, together with the variables $x_{it}$, the outcome process $%
y_{it}$ according to Equation \eqref{model motivation}. The identified
network may be a combination of elicited types of social interactions --
such as friendship formation, lending and borrowing relations, links with
relatives -- or different from elicited data, as long as the links are
relevant in determining the outcomes. In our case, and in line with the
literature, the network ties $W_{ij}$ are considered to be deterministic
parameters or predetermined. Alternatively, the networks are assumed to be
the outcome of a stochastic process, such as the latent space model (%
\citealp{hoff2002latent}; \citealp{Brezaetal2017}) or Exponential Random
Graphs models (\citealp{HollandLeinhardt1981}).

Our conclusions discuss how our approach can be modified, and assumptions
weakened, to integrate partial knowledge of $W_{0}$. We discuss further
applications and the steps required to simultaneously identify models of
network formation and the structure of social interactions. The practical
use of our proposed method has already been demonstrated in applications.
For example, \citet{fetzeretal2020} study the impact on conflict of the
transition of security responsibilities between international and Afghan
forces. Our proposed method is used to control for violation of SUTVA-type
hypotheses that might occur because of spillover and displacement effects of
insurgent forces across districts. Since the pattern of displacement is
unobserved -- and, in fact, insurgents have incentives to obfuscate their
strategy -- the current method is applied to fully recover the network and
bound the effects of the end of the military occupation on conflict.%
\footnote{\citet{Zhou2019} applies our identification results, focusing on
unobserved networks with grouped heterogeneity, to suggest a nonlinear least
squares procedure for estimation on a single network observation.%
}

We proceed as follows. Section 2 presents our core result:\ the sufficient
conditions under which the social interactions matrix, endogenous and
exogenous social effects are globally identified. Section 3 describes the
high-dimensional techniques used for estimation based on the
Adaptive Elastic Net GMM method and presents simulation results from
stylized and real-world networks. Section 4 applies our methods to study tax
competition between US\ states. Section 5 concludes. The Appendix provides
proofs and further details on estimation and simulations.

\section{Identification\label{sec:Identification}}

\subsection{Setup\label{subsec:setup}}

Consider a researcher with panel data covering $i=1,\dots ,N$ individuals
repeatedly observed over $t=1,\dots ,T$ instances. The number of individuals 
$N$ in the network is fixed but potentially large. The aim is to use this
data to identify a social interactions model with no data on actual social
ties. For expositional ease, we first consider identification in a simpler
version of the canonical model in (\ref{model motivation}), where we drop
individual-specific ($\alpha _{i}$) and time-constant fixed effects ($\alpha
_{t}$) and assume $x_{it}$ is a one-dimensional regressor for individual $i$
and instance $t$. We later extend the analysis to include
individual-specific, time-constant fixed effects and allow for
multidimensional covariates $x_{k,it}$, $k=1,\dots ,K$. We adopt the
subscript \textquotedblleft 0\textquotedblright\ to denote parameters
generating the data, and non-subscripted parameters are generic values in
the parameter space:%
\begin{equation}
y_{it}=\rho _{0}\sum_{j=1}^{N}W_{0,ij}y_{jt}+\beta _{0}x_{it}+\gamma
_{0}\sum_{j=1}^{N}W_{0,ij}x_{jt}+\epsilon _{it}.  \label{eq:modelSF}
\end{equation}%
As the outcomes for all individuals $i=1,\dots ,N$ obey equations analogous to (%
\ref{eq:modelSF}), the system of equations can be more compactly written in
matrix notation as:%
\begin{equation}
y_{t}=\rho _{0}W_{0}y_{t}+\beta _{0}x_{t}+\gamma _{0}W_{0}x_{t}+\epsilon
_{t}.  \label{eq:Model}
\end{equation}%
The vector of outcomes $y_{t}=(y_{1t},\dots ,y_{Nt})^{\prime }$ assembles
the individual outcomes in instance $t$; the vector $x_{t}$ does the same
with individual characteristics. $y_{t}$, $x_{t}$, and $\epsilon _{t}$ have
dimension $N\times 1$, the social interactions matrix $W_{0}$ is $N\times N$%
, and $\rho _{0}$, $\beta _{0}$, and $\gamma _{0}$ are scalar parameters. We
do not make any distributional assumptions on $\epsilon _{t}$ beyond $%
\mathbb{E}(\epsilon _{t}|x_{t})=0$ (or $\mathbb{E}(\epsilon _{t}|z_{t})=0$
for an appropriate instrumental variable $z_{t}$ if $x_{t}$ is endogenous).
We assume the network structure is predetermined and constant, and that the
number of individuals $N$ is fixed and repeated. In reality, networks may
evolve over time. We thus later expand the method for dynamic network cases.
The network structure $W_{0}$ is a parameter to be identified and estimated.%

The social interaction model \eqref{eq:Model} has been widely studied (%
\citealp{Manski1993}; \citealp{Manresa2016}; and \citealp{Blumeetal2015}, among many
others), but it is also restrictive in at least two senses. First, we
consider $W_{ij}$ to be fixed and predetermined, and not through models
of strategic network formation (\citealp{jackson1996strategic}; %
\citealp{dePaulaRichardsTamer2018}) or of stochastic nature, as in the class
of Exponential Random Graphs (\citealp{HollandLeinhardt1981}) or Latent
Distance models (\citealp{hoff2002latent}; \citealp{Brezaetal2017}). If
there is feedback between outcome determination and link formation, and
especially if this involves unobservables, it would be important to model
network formation more explicitly.

A regression of outcomes on covariates corresponds, then, to the reduced
form for (\ref{eq:Model}),%
\begin{equation}
y_{t}=\Pi _{0}x_{t}+\nu _{t},  \label{eq2:modelRF}
\end{equation}%
with $\Pi _{0}=(I-\rho _{0}W_{0})^{-1}(\beta _{0}I+\gamma _{0}W_{0})$ and $%
\nu _{t}\equiv (I-\rho _{0}W_{0})^{-1}\epsilon _{t}$.   If $W_{0}$ is observed, \citet{Bramoulleetal2009} note that a structure $%
(\rho ,\beta ,\gamma )$ that is observationally equivalent to $(\rho
_{0},\beta _{0},\gamma _{0})$ is such that $(I-\rho _{0}W_{0})^{-1}(\beta
_{0}I+\gamma _{0}W_{0})=(I-\rho W_{0})^{-1}(\beta I+\gamma W_{0})$. This can
be written as a linear equation in $I,W_{0}$, and $W_{0}^{2}$, and
identification is established if those matrices are linearly independent. If 
$W_{0}$ is not observed, the putative unobserved structure comprises $W_{0}$,
and an observationally equivalent parameter vector will instead satisfy $%
(I-\rho _{0}W_{0})^{-1}(\beta _{0}I+\gamma _{0}W_{0})=(I-\rho W)^{-1}(\beta
I+\gamma W)$. Following the strategy in \citet{Bramoulleetal2009} would lead
to an equation in $I,W,W_{0}$, and $WW_{0}$, so the insights obtained in
that paper do \emph{not} carry over to the case we study when $W_{0}$
is unknown.

We establish identification of the structural parameters of the model,
including the social interactions matrix $W_{0}$, from the coefficients
matrix $\Pi _{0}$. Without data on the network $W_{0}$, we treat it as an
additional parameter in an otherwise standard model relating outcomes and
covariates. Our identification strategy relies on how changes in covariates $%
x_{it}$ reverberate through the system and impact $y_{it}$, as well as
outcomes for other individuals. These are summarized by the entries of the
coefficient matrix $\Pi _{0}$, which, in turn, encode information about $%
W_{0}$ and $(\rho _{0},\beta _{0},\gamma _{0})$. A non-zero partial effect of $%
x_{it}$ on $y_{jt}$ indicates the existence of direct \emph{or} indirect
links between $i$ and $j$. When $\rho _{0}=0$ (and $\Pi _{0}=\beta
_{0}I+\gamma _{0}W_{0}$), only direct links  produce such a
correlation. When $\rho \neq 0$, both direct and indirect connections may
generate a non-zero response, but distant connections will lead to a lower
response. Our results formally determine sufficient conditions to precisely
disentangle these forces.

We set out six assumptions underpinning our main identification results.
Three of these are entirely standard. A fourth is a normalization required
to separately identify ($\rho _{0}$, $\gamma _{0}$) from $W_{0}$, and the
fifth is closely related to known results on the identification of ($\rho
_{0}$, $\gamma _{0}$) when $W_{0}$ is known (\citealp{Bramoulleetal2009}).
The sixth assumption pertains to the relation between the nature of repeated
multiple observations of the outcome and covariates and restrictions on the
stability of $W$. These Assumptions (A1-A6) deliver an identified set of up
to two points.

Our first assumption explicitly states that no individuals affect
themselves and is a standard condition in social interaction models:

\begin{itemize}
\item[(A1)] $(W_{0})_{ii}=0$, $i=1,\dots ,N$.
\end{itemize}

Assumption (A1) rules out applications with self-influence. For example,
Input-Output matrices typically feature $(W_{0})_{ii}>0$, as firms tend to
source from other firms in the same industry. With Assumption (A1), we can
omit elements on the diagonal of $W_{0}$ from the parameter space. We thus
can denote a generic parameter vector as $\theta =\left( W_{12},\dots
,W_{N,N-1},\rho ,\gamma ,\beta \right) ^{\prime }\in \mathbb{R}^{m}$, where $%
m=N\left( N-1\right) +3$, and $W_{ij}$ is the $(i,j)$-th element of $W$.
Reduced-form parameters can be tied back to the structural model (\ref%
{eq:Model}) by letting $\Pi :\mathbb{R}^{m}\rightarrow \mathbb{R}^{N^{2}}$
define the relation between structural and reduced-form parameters:%
\begin{equation*}
\Pi (\theta )=\left( I-\rho W\right) ^{-1}\left( \beta I+\gamma W\right) ,
\end{equation*}%
where $\theta \in \mathbb{R}^{m}$, and $\Pi _{0}\equiv \Pi (\theta _{0})$.

As $\epsilon _{t}$ (and, consequently, $\nu _{t}$) is mean-independent from $%
x_{t}$, $\mathbb{E}[\epsilon _{t}|x_{t}]=0$, the matrix $\Pi _{0}$ can be
identified as the linear projection of $y_{t}$ on $x_{t}$. We do not impose
additional distributional assumptions on the disturbance term, except for
conditions that allow us to identify the reduced-form parameters in %
\eqref{eq2:modelRF}. If $x_{t}$ is endogenous, i.e., $\mathbb{E}[\epsilon
_{t}|x_{t}]\neq 0$, a vector of instrumental variables $z_{t}$ may still be
used to identify $\Pi _{0}$. In either case, identification of $\Pi _{0}$
requires variation of the regressor across individuals $i$ and through
instances $t$. In other words, either $\mathbb{E}[x_{t}x_{t}^{\prime }]$ (if
exogeneity holds) or $\mathbb{E}[x_{t}z_{t}^{\prime }]$ (otherwise) is
full-rank.

Our next assumption controls the propagation of shocks and guarantees that they
die as they reverberate through the network. This provides adequate
stability and is related to the concept of stationarity in network models.
It implies the maximum eigenvalue norm of $\rho _{0}W_{0}$ is less than one
and ensures $(I-\rho _{0}W_{0})$ is a non-singular matrix. As the variance
of $y_{t}$ exists, the transformation $\Pi (\theta _{0})$ is well-defined,
and the Neumann expansion $(I-\rho _{0}W_{0})^{-1}=\sum_{j=0}^{\infty }(\rho
_{0}W_{0})^{j}$ is appropriate.

\begin{itemize}
\item[(A2)] \textcolor{black}{$\sum_{j=1}^N|\rho_0 (W_0)_{ij}| < 1$} for
every $i=1,\dots ,N$, 
\textcolor{black}{$\|W_0\| < C$ for some positive $C
\in \mathbb{R}$} and $|\rho _{0}|<1$.
\end{itemize}

\noindent We next assume that network effects do not cancel out, another
standard assumption. As we will show, this assumption rules out the
pathological case in which endogenous and exogenous effects exactly cancel
each other out:

\begin{itemize}
\item[(A3)] $\beta _{0}\rho _{0}+\gamma _{0}\neq 0$.
\end{itemize}

\noindent The need for this assumption can be shown by expanding the
expression for $\Pi (\theta _{0})$, which is possible by (A2):%
\begin{equation}
\Pi (\theta _{0})=\beta _{0}I+(\rho _{0}\beta _{0}+\gamma
_{0})\sum_{k=1}^{\infty }\rho _{0}^{k-1}W_{0}^{k}.  \label{eq:RFexpanded}
\end{equation}%
If Assumption (A3) were violated, $\beta _{0}\rho _{0}+\gamma _{0}=0$ and $%
\Pi _{0}=\beta _{0}I$, so the endogenous and exogenous effects would balance each
other out, and network effects would be altogether eliminated in the reduced form.%
\footnote{%
\textcolor{black}{One important case is when networks do not
determine outcomes, which we interpret as $\rho_0=\gamma_0=0$ or with $W_0$
representing the empty network. From equation \eqref{eq:RFexpanded}, it is
clear that if $\Pi(\theta_0)$ is \emph{not} diagonal with constant entries,
then it must be that $(\rho_0\beta_0+\gamma_0)\neq0$, which implies that
$\rho_0\neq0$ or $\gamma_0\neq0$, and also that $W_0$ is non-empty. Taken
together, this suggests that the observation that $\Pi(\theta_0)$ is not
diagonal is sufficient to ensure that network effects are present and
Assumption (A3) is not violated.}}

Identification of the social effects parameters $(\rho _{0},\gamma _{0})$
requires that at least one row of $W_{0}$ adds to a fixed and known number.
Otherwise, $\rho _{0}$ and $\gamma _{0}$ cannot be separately identified
from $W_{0}$. Clearly, no such condition would be required if $W_{0}$ were
observed.

\begin{itemize}
\item[(A4)] There is an $i$ such that $\sum_{j=1,\dots ,N}(W_{0})_{ij}=1$.
\end{itemize}

Letting $W_{y}\equiv \rho _{0}W_{0}$ and $W_{x}\equiv \gamma _{0}W_{0}$
denote the matrices that summarize the influence of peers' outcomes (the
endogenous social effects) and characteristics on one's outcome (the
exogenous social effects), respectively, the assumption above can be seen as
a normalization. In this case, $\rho _{0}$ and $\gamma _{0}$ represent the
row-sum for individual $i$ in $W_{y}$ and $W_{x}$, respectively.\footnote{%
\textcolor{black}{Alternatively, one could normalize $\rho^*=1$
and rescale the network accordingly. In this case,
$W^*=\rho_0W_0$ would be identified instead. Also,
$W_x=\frac{\gamma_0}{\rho_0}W^*$ so $\gamma_0$ would be identified relative
to $\rho_0$. $W_y$ and $W_x$ would be unchanged.}}

The fifth assumption allows for a specific kind of network asymmetry. We
require the diagonal of $W_{0}^{2}$ not to be constant as one of our
sufficient conditions for identification.

\begin{itemize}
\item[(A5)] There exists $l,$ $k$ such that $(W_{0}^{2})_{ll}\neq
(W_{0}^{2})_{kk}$, i.e., the diagonal of $W_{0}^{2}$ is not proportional to $%
\iota $, where $\iota$ is the $N\times1$ vector of ones.
\end{itemize}

In unweighted networks, the diagonal of the square of the social
interactions matrix captures the number of reciprocated links for each
individual or, in the case of undirected networks, the popularity of those
individuals. Assumption (A5) hence intuitively suggests differential
popularity across individuals in the social network.

This assumption is related to the network asymmetry condition proposed
elsewhere, such as in \citet{Bramoulleetal2009}. They show that when $W_{0}$
is known, the structural model (\ref{eq:modelSF}) is identified if $I$, $%
W_{0}$, and $W_{0}^{2}$ are linearly independent. Given the remaining
assumptions, this condition is satisfied if (A5) is satisfied, but the
converse is \emph{not} true: one can construct examples in which $I$, $W_{0}$%
, and $W_{0}^{2}$ are linearly independent when $W_{0}^{2}$ has a constant
diagonal, so $\Pi _{0}$ does not pin down $\theta _{0}$. See Example 1
in Appendix A. The strengthening of this hypothesis is the formal price to
pay for the social interactions matrix $W_{0}$ being unknown to the
researcher. 

Before proceeding to our formal results, we provide a very simple
illustration to shed light on how the assumptions above come together to
provide identification. Suppose the observed reduced-form matrix is,%
\begin{equation*}
\Pi _{0}=\frac{1}{455}\left[ 
\begin{array}{ccc}
275 & 310 & 0 \\ 
310 & 275 & 0 \\ 
0 & 0 & 182%
\end{array}%
\right] ,
\end{equation*}%
and that, following (A4), the first row is normalized to one. From the third
row and column of $\Pi _{0}$, we see there is no path of any length
connecting the individual in row 3 to or from those in rows 1 or 2, since her
outcome is not affected by their covariates and their outcomes are not
affected by her covariates. In other words, individual 3, is isolated and $%
(W_{0})_{13}=(W_{0})_{23}=(W_{0})_{31}=(W_{0})_{32}=0$. On the other hand,
individuals 1 and 2 cannot be isolated, as their covariates are correlated
with the other individual's outcome, reflecting (A5).\footnote{%
If on the other hand, $(W_{0})_{ij}=0.5,i\neq j$ in violation of (A5), and
all agents were connected, the model would not be identified.} Due to the
row-sum normalization of the first row, $(W_{0})_{12}=1$. Using (A3), it can
be seen that $W_{0}$ is symmetric if $\Pi _{0}$ is symmetric. We thus find
that $(W_{0})_{21}=1$. This and (A1) map all elements of $W_{0}$, and thus,%
\begin{equation*}
W_{0}=\left[ 
\begin{array}{ccc}
0 & 1 & 0 \\ 
1 & 0 & 0 \\ 
0 & 0 & 0%
\end{array}%
\right] .
\end{equation*}
As the third individual is isolated, she will only be affected by her
exogenous $x_{i}$ and not by endogenous or exogenous peer effects. Hence, the 
$(3,3)$ element of $\Pi _{0}$ is equal to $\beta _{0}=\frac{182}{455}=.4$. \
To find $\rho _{0}$, note that $(I-\rho _{0}W_{0})\Pi _{0}=\beta
_{0}I+\gamma _{0}W_{0}$. Hence, focusing on the (1,1) elements of the
matrices above, we find that $\frac{275}{455}-\rho _{0}\frac{310}{455}=.4$,
implying $\rho _{0}=.3$ (complying with (A2)). Finally, $\gamma _{0}$ is
identified from entry $(1,2)$, giving $\gamma _{0}=\frac{310}{455}-.3\frac{%
275}{455}=.5$.

Our final assumption articulates the need for a constant network $W_0$
observed over multiple instances of $y_t$ and $x_t$:

\begin{itemize}
\item[(A6)] $y_t$ and $x_t$ are observed for individuals $i=1,\dots,N$, and
instances $t=1,\dots,T$, and the network $W_0$ does not depend on $t$
\end{itemize}

Here, ``instances'' can refer to time but also to settings in which the same
units are observed over multiple episodes. For example, if $i$ are firms,
then $t$ can be segmented markets in which they operate. For simplicity, we
refer to an instance as a time period. If $\Pi_0$ is known, the main
identification result we articulate below will state that $W_0$, $\rho_0$, $%
\beta_0$ and $\gamma_0$ are globally identified. However, in practice, $\Pi_0$
is rarely observed and thus all quantities need to be estimated. For this
purpose, when $\Pi_0$ is not known, multiple observations of $y_t$ and $x_t$
with a constant $W_0$ are required to implement the estimator. We expand on
estimation requirements in Section \ref{sec:estimation}.

Importantly, the main identification results (for a given $\Pi_0$) could, in
principle, be applied for each time period $t$. That is, one can write a
version of Equation \eqref{eq:Model} as 
\begin{eqnarray*}
y_t&=&\rho_0W_{0t}y_t+\beta_0x_t+\gamma_0W_{0t}x_t+\epsilon_t,
\end{eqnarray*}
where $W_{0t}$ is time-varying and, consequently, the reduced-form
interaction matrix $\Pi_{0t}=(I-\rho_0 W_{0t})^{-1}(I\beta_0+W_{0t}\gamma_0)$
is also time-varying. If the reduced-form matrices were known, the
identification results we develop below could be applied to the reduced-form
element by element for each $\Pi_{0t}$. Again, one rarely observes $\Pi_{0t}$%
, for all $t\in[1,T]$. This observation will motivate an extension of the
method, presented in Section \ref{subsec:timevaryingW}, where Assumption
(A6) is relaxed and $W_{0}$ is allowed to vary with $t$.

In an extension, we allow the network $W_{0t}$ to vary over time and
introduce kernel weights. Akin to the nonparametric regression $%
Y=f(X)+\epsilon $, $f$ is identified if $\mathbb{E}(\epsilon |X)=0$, and it
is possible to estimate $f$ using neighboring observations if $f$ is
sufficiently smooth or varies slowly. Similar considerations extend to
varying-coefficient models and, in particular, time-varying coefficient
models where local stability conditions as those discussed in %
\citet{Dahlhaus2012} are usually invoked (see also %
\citet{HastieTibshirani1993} and their Example (e)).

\subsection{Main Identification Results\label{subsec:mainresults}}

Under the assumptions above, we can begin to identify parameters related to
the network. These results are then useful for our main identification
theorems. Let $\lambda _{0j}$ denote an eigenvalue of $W_{0}$ with
corresponding eigenvector $v_{0,j}$ for $j=1,\dots ,N$. Assumptions (A2) and
(A3) allow us to identify the eigenvectors of $W_{0}$ directly from the
reduced form. As $|\rho _{0}|<1$:%
\begin{eqnarray}
\Pi _{0}v_{0,j} &=&\beta _{0}v_{0,j}+(\rho _{0}\beta _{0}+\gamma
_{0})\sum_{k=1}^{\infty }\rho _{0}^{k-1}W_{0}^{k}v_{0,j}  \notag \\
&=&\left[ \beta _{0}+(\rho _{0}\beta _{0}+\gamma _{0})\sum_{k=1}^{\infty
}\rho _{0}^{k-1}\lambda _{0,j}^{k}\right] v_{0,j}  \notag \\
&=&\frac{\beta _{0}+\gamma _{0}\lambda _{0,j}}{1-\rho _{0}\lambda _{0,j}}%
v_{0,j}.  \label{eq:eigenvaluePiW}
\end{eqnarray}%
The infinite sum converges as $|\rho _{0}\lambda _{0,j}|<1$ by (A2). The
equation above implies that $v_{0,j}$ is also an eigenvector of $\Pi _{0}$
with the associated eigenvalue $\lambda _{\Pi ,j}=\frac{\beta _{0}+\gamma
_{0}\lambda _{0,j}}{1-\rho _{0}\lambda _{0,j}}$. The fact that eigenvectors
of $W_{0}$ are also eigenvectors of $\Pi _{0}$ has a useful implication:
eigencentralities may be identified from the reduced form, even when $W_{0}$
is not identified. As detailed in \citet{DePaula2017} and %
\citet{Jacksonetal2017}, such eigencentralities often play an important role
in empirical work as they allow a mapping back to underlying models of
social interaction.\footnote{%
To identify the eigencentralities, we identify the eigenvector that
corresponds to the dominant eigenvalue. If $W_{0}$ is non-negative and
irreducible, this is the (unique) eigenvector with strictly positive
entries, by the Perron-Frobenius theorem for non-negative matrices (see %
\citealp{HornJohnson2013}, p. 534).}

Now let $\Theta \equiv \{\theta \in \mathbb{R}^{m}:$ A$\text{ssumptions
(A1)-(A6) are satisfied}\}$ be the structural parameter space of interest.
Our identification argument is structured as follows: a) we first establish
local identification of the mapping $\Pi (\theta )$ using classical results
on the rank of the gradient of \citealp{Rothenberg1971} (Theorem 1); b) we
then show that $\Pi (\theta )$ is proper (Corollary 1); and c) has a
connected image (Lemma 2, in the Appendix); d) allowing us to state the
cardinality of the pre-image $\Pi ^{-1}(\bar{\Pi})$ is constant for any $%
\bar{\Pi}$ in the image of $\Pi (\cdot )$, and that the cardinality is at
most 2 (Theorem 2). We then provide additional conditions to narrow the
identified set to a singleton (Corollaries 2-4).

We now formally present our results. Our first theorem establishes local
identification of the mapping. A parameter point $\theta _{0}$ is locally
identifiable if there exists a neighborhood of $\theta _{0}$ containing no
other $\theta $ which is observationally equivalent. Using classical results
in \citet{Rothenberg1971}, we show that our assumptions are sufficient to
ensure that the Jacobian of $\Pi $ relative to $\theta $ is non-singular,
which, in turn, suffices to establish local identification.

\begin{thm}
\label{thm:localidentification} Assume (A1)-(A6). $\theta_0\in\Theta$ is
locally identified.
\end{thm}

An immediate consequence of local identification is that the set $\{\theta
\in \Theta :\Pi (\theta )=\Pi (\theta _{0})\}$ is discrete (i.e., its
elements are isolated points). The following corollary establishes that $\Pi 
$ is a proper function, i.e. the inverse image $\Pi ^{-1}(K)$ of any compact
set $K\subset \mathbb{R}^{N^{2}}$ is also compact (\citealp{KrantzParks2013}%
, p.\ 124). Since it is discrete, the identified set must be finite.

\begin{cor}
\label{thm:finiteidentification} Assume (A1)-(A6). Then $\Pi(\cdot)$ is a
proper mapping. Moreover, the set $\{\theta:\Pi(\theta) = \Pi(\theta_0)\}$
has a finite number of elements.
\end{cor}

\noindent Under additional assumptions, the identified set is at most a
singleton in each of the partitioning sets $\Theta _{-}\equiv \Theta \cap
\{\rho \beta +\gamma <0\}$ and $\Theta _{+}\equiv \Theta \cap \{\rho \beta
+\gamma >0\}$.\footnote{%
The global inversion results we use are related to, but different from,
variations on a classic inversion result of Hadamard that has been used in
the literature. In contrast, we employ results on the cardinality of the
pre-image of a function, relying on less stringent assumptions. While the
Hadamard result requires the image of the function to be simply-connected
(Theorem 6.2.8 of \citealp{KrantzParks2013}), the results we rely on do not.}

Since $\Theta =\Theta _{-}\cup \Theta _{+}$, if the sign of $\rho _{0}\beta
_{0}+\gamma _{0}$ is unknown, the identified set contains, at most, two
elements. In the theorem that follows, we show global identification only
for $\theta \in \Theta _{+}$, since arguments are mirrored for $\theta \in
\Theta _{-}$.

\begin{thm}
\label{thm:globalidentification} Assume (A1)-(A6). Then for every $%
\theta\in\Theta_{+}$, we have $\Pi(\theta)=\Pi(\theta_0)\Rightarrow\theta=%
\theta_0$. That is, $\theta_0$ is globally identified with respect to the
set $\Theta_{+}$.
\end{thm}

\noindent Similar arguments apply if Theorem \ref{thm:globalidentification}
instead were to be restricted to $\theta\in\Theta_{-}$. The proof of the
corollary below is immediate and therefore omitted.

\begin{cor}
\label{cor:globalidentification1} Assume (A1)-(A6). If $\rho _{0}\beta
_{0}+\gamma _{0}>0$, then the identified set contains at most one element,
and similarly if $\rho _{0}\beta _{0}+\gamma _{0}<0$. Hence, if the sign of $%
\rho _{0}\beta _{0}+\gamma _{0}$ is unknown, the identified set contains, at
most, two elements.\footnote{%
\textcolor{black}{Under some special
conditions, the mirror image of $\theta_0$ can be characterized from
equation \eqref{eq:RFexpanded}. If $-W_0$ satisfies Assumption (A4), we may
set $\rho^*=-\rho_0$, $\beta^*=\beta_0$, $\gamma^*=-\gamma_0$ and
$W^*=-W_0$. Then, $\rho_0\beta_0+\gamma_0=-(\rho^*\beta^*+\gamma^*)$. Also
note that $\sum_{k=1}^\infty \rho_0^{k-1}W_0^k=-\sum_{k=1}^\infty
(\rho^*)^{k-1}(W^*)^k$, so $(\rho _{0}\beta
_{0}+\gamma_{0})\sum_{k=1}^{\infty }\rho _{0}^{k-1}W_{0}^{k}=(\rho ^*\beta
^*+\gamma^*)\sum_{k=1}^{\infty }(\rho^*)^{k-1}(W^*)^{k}$. It follows that
$\Pi(\theta_0)=\Pi(\theta^*)$, where
$\theta^*=(\rho^*,\beta^*,\gamma^*,W^*)$.}}
\end{cor}

We now turn our attention to the problem of identifying the sign of $\rho
_{0}\beta _{0}+\gamma _{0}$ from the observation of $\Pi _{0}$. This would
then allow us to establish global identification using Theorem \ref%
{thm:globalidentification}. It is apparent from (\ref{eq:RFexpanded}) that
if $\rho _{0}>0$ and $(W_{0})_{ij}\geq 0$, for all $i,j=\{1,\dots ,N\}$, the
off-diagonal elements of $\Pi _{0}$ identify the sign of $\rho _{0}\beta
_{0}+\gamma _{0}$.

\begin{cor}
\label{cor:globalidentification2} Assume (A1)-(A6). If $\rho_0 > 0$ and $%
(W_0)_{ij} \ge 0$, the model is globally identified.
\end{cor}

Real-world applications often suggest endogenous social interactions are
positive ($\rho _{0}>0$), in which case global identification is fully
established by Corollary \ref{cor:globalidentification2}. On the other hand,
if $\rho _{0}<0$ (e.g., if outcomes are strategic substitutes), $\rho
_{0}^{k} $ in (\ref{eq:RFexpanded}) alternates signs with $k$, and the
off-diagonal elements no longer carry the sign of $\rho _{0}\beta
_{0}+\gamma _{0}$. Nonetheless, if $W_{0}$ is non-negative and irreducible
(i.e., not permutable into a block-triangular matrix or, equivalently, a
strongly connected social network), the model is also identifiable without
further restrictions on $\rho _{0}$:

\begin{cor}
\label{cor:globalidentification3} Assume (A1)-(A6), $(W_0)_{ij} \ge 0$ and $%
W_0$ is irreducible. If $W_0$ has at least two real eigenvalues or $|\rho_0|
< \sqrt{2}/2$, then the model is globally identified.
\end{cor}

\noindent Corollary \ref{cor:globalidentification3} requires that $W_0$ be
irreducible, i.e., that it is not permutable into a block upper-triangular
matrix. In the context of directed graphs, this is similar to requiring that
the matrix be strongly connected, that is, that any node can be reached from
any other node. The corollary then rules out cases when the network is not
connected, for example, if there are two disjoint groups (with no connection
across groups), or a star network pointing from the center towards the
edges. The corollary holds if there are at least two real eigenvalues, or if 
$\rho _{0}$ is appropriately bounded. Since $W_{0}$ is non-negative, it has
at least one real eigenvalue by the Perron-Frobenius theorem. If $W_{0}$ is
symmetric, for example, its eigenvalues are all real, and Corollary \ref%
{cor:globalidentification3} holds. It also holds if $(W_{0})_{ij}\leq 0$, as
we can rewrite the model as $\rho W_{0}=-\rho |W_{0}|$, where $|W_{0}|$ is
the matrix whose entries are the absolute values of the entries in $W_{0}$.
However, Corollary \ref{cor:globalidentification3} rules out cases that mix
positive $(W_0)_{ij}\geq 0$ and negative interactions $(W_0)_{ij}\leq 0$. In
any case, the bound on $|\rho _{0}|$ is sufficient and holds in most (if not
all) empirical estimates we are aware of obtained from either elicited or
postulated networks, and in our application on tax competition.

\subsection{Extensions}

We present three extensions of the method for individual fixed effects,
common shocks, and time-varying $W$. Appendix B describes extensions for
multivariate covariates and heterogeneous $\beta _{0}$.

\subsubsection{Individual Fixed Effects\label{subsec:corrfixedeffects}}

We observe outcomes for $i=1,\dots ,N$ individuals repeatedly through $%
t=1,\dots ,T$ instances. If $t$ corresponds to time, it is natural to think
of there being unobserved heterogeneity across individuals, $\alpha _{i}$,
to be accounted for when estimating $\Pi _{0}$. The structural model (\ref%
{eq:modelSF}) is then,%
\begin{equation*}
y_{it}=\rho _{0}\sum_{j=1}^{N}W_{0,ij}y_{jt}+\beta _{0}x_{it}+\gamma
_{0}\sum_{j=1}^{N}W_{0,ij}x_{jt}+\alpha _{i}+\epsilon _{it},
\end{equation*}%
which can be written in matrix form as%
\begin{equation*}
y_{t}=\rho _{0}W_{0}y_{t}+x_{t}\beta _{0}+W_{0}x_{t}\gamma _{0}+\alpha
^{\ast }+\epsilon _{t},
\end{equation*}%
where $\alpha ^{\ast }$ is the vector of fixed effects. Individual-specific
and time-constant fixed effects can be eliminated using the standard
subtraction of individual time averages. Defining $\bar{y}%
_{t}=T^{-1}\sum_{t=1}^{T}y_{t}$, $\bar{x}_{t}=T^{-1}\sum_{t=1}^{T}x_{t}$, and 
$\bar{\epsilon}_{t}=T^{-1}\sum_{t=1}^{T}\epsilon _{t}$, 
\begin{equation*}
y_{t}-\bar{y}_{t}=\rho _{0}W_{0}\left( y_{t}-\bar{y}_{t}\right) +\left(
x_{t}-\bar{x}_{t}\right) \beta _{0}+W_{0}\left( x_{t}-\bar{x}_{t}\right)
\gamma _{0}+\epsilon _{t}-\bar{\epsilon}_{t},
\end{equation*}%
if $W_{0}$ does not change with time. Identification from the reduced form
follows from previous theorems, since $\Pi _{0}$ is unchanged when
regressing $y_{t}-\bar{y}_{t}$ on $x_{t}-\bar{x}_{t}$.\footnote{%
\textcolor{black}{As is the case in panel data, this would require strict
exogeneity ($\mathbb{E}[\epsilon _{s}|x_{t}]=0$ for any $s$ and $t$) or
predetermined errors ($\mathbb{E}[\epsilon _{s}|x_{t}]=0$ for $s \ge t$) so
that the matrix $\Pi _{0}$ can be consistently estimated.}}

\subsubsection{Common Shocks\label{subsec:commonshocks}}

We next allow for unobserved common shocks to all individuals in the network
in the same instance $t$. Such correlated effects\emph{\ }$\alpha _{t}$ can
confound the identification of social interactions. As we have not placed
any distributional assumption on the covariance matrix of the disturbance
term, our analysis readily incorporates correlated effects that are
orthogonal to $x_{t}$. When this is not the case, one possibility is to
model the correlated effects $\alpha _{t}$ explicitly. The model then is,%
\begin{equation*}
y_{t}=\rho _{0}W_{0}y_{t}+x_{t}\beta _{0}+\gamma _{0}W_{0}x_{t}+\alpha
_{t}\iota +\epsilon _{t},
\end{equation*}%
where $\alpha _{t}$ is a scalar capturing shocks in the network common to
all individuals. Let $\Pi _{01}=\left( I-\rho _{0}W_{0}\right) ^{-1}$ and $%
\Pi _{02}=\left( \beta _{0}I+\gamma _{0}W_{0}\right) $ such that $\Pi
_{0}=\Pi _{01}\Pi _{02}$. The reduced-form model is%
\begin{equation*}
y_{t}=\Pi _{0}x_{t}+\alpha _{t}\Pi _{01}\iota +v_{t}.
\end{equation*}%
We propose a transformation to eliminate the correlated effects:\ exclude
the individual-invariant $\alpha _{t}$, subtracting the mean of the
variables in a given period (global differencing). For this purpose, define $%
H=\frac{1}{n}\iota \iota ^{\prime }$. We note that in empirical and
theoretical work, it is customary to strengthen Assumption (A4) and require
that \emph{all} rows of $W_{0}$ sum to one if no individual is isolated (see
for example \citealp{Blumeetal2015}). This strengthened assumption is
usually referred to as row-sum normalization, and is stated below:

\begin{itemize}
\item[(A4$^\prime$)] For all $i=1,...N$, we have that $\sum_{j=1,\dots
,N}(W_{0})_{ij}=1$.
\end{itemize}

This can be written compactly as $W_{0}\iota =\iota $. In this case, $W_{0}$
can be interpreted as the normalized adjacency matrix. Under row-sum
normalization we have that,%
\begin{eqnarray*}
\left( I-H\right) y_{t} &=&\left( I-H\right) \left( I-\rho _{0}W_{0}\right)
^{-1}\left( \beta _{0}I+\gamma _{0}W_{0}\right) x_{t}+\left( I-H\right)
\left( I-\rho _{0}W_{0}\right) ^{-1}\epsilon _{t} \\
&=&\left( I-H\right) \Pi _{0}x_{t}+\left( I-H\right) v_{t},
\end{eqnarray*}%
because $\left( I-H\right) \left( I-\rho _{0}W_{0}\right) ^{-1}\alpha
_{t}\iota =0$ if Assumption (A4$^\prime$) holds. It then follows that $\tilde{\Pi}%
_{0}=(I-H)\Pi _{0}$ is identified. The next proposition shows that, under
row-sum normalization of $W_{0}$, $\Pi _{0}$ is identified from $\tilde{\Pi}%
_{0}$ (and, as a consequence, the previous results immediately apply).

\begin{prop}
\label{prop:globalidentification1} If $W_0$ is non-negative, irreducible, and
row-sum normalized, $\Pi _{0}$ is identified from $\tilde{\Pi}_{0}$.
\end{prop}

Under row-sum normalization of $W_{0}$, a common group-level shock affects
individuals homogeneously since $(I-\rho _{0}W_{0})^{-1}\alpha _{t}\iota
=\alpha _{t}(I+\rho _{0}W_{0}+\rho _{0}^{2}W_{0}^{2}+\cdots )\iota =\frac{%
\alpha _{t}}{1-\rho _{0}}\iota $, which is a vector with no variation across
entries. Consequently, global differencing eliminates correlated effects and 
$\left( I-H\right) \left( I-\rho _{0}W_{0}\right) ^{-1}\alpha _{t}\iota
=\left( I-\rho _{0}W_{0}\right) ^{-1}\alpha _{t}\left( I-H\right) \iota =0$.
Absent row-sum normalization, global differencing does not ensure correlated
effects are eliminated. To see this, note that $(I-\rho _{0}W_{0})^{-1}$ is
no longer row-sum normalized and $\alpha _{t}(I-\rho _{0}W_{0})^{-1}\iota $
does not have constant entries.

The next proposition makes this point formally: that the stronger Assumption
(A4$^\prime$) is \emph{necessary} to eliminate group-level shocks by showing it is
not possible to construct a data transformation that eliminates group
effects in the absence of row-sum normalization.

\begin{prop}
\label{prop:globalidentification2} Define $r_{W_{0}}=(I-\rho
_{0}W_{0})^{-1}\iota $. If in space $\Theta =\{\theta \in \mathbb{R}^{m}:$
Assumptions (A1)-(A6) are satisfied$\}$, there are $N$ matrices $%
W_{0}^{(1)},\dots ,W_{0}^{(N)}$ such that $[r_{W_{0}^{(1)}}\;\cdots
\;r_{W_{0}^{(N)}}]$ has rank $N$, then the only transformation such that $(I-%
\tilde{H})(I-\rho _{0}W_{0})^{-1}\iota =0$ is $\tilde{H}=I$.
\end{prop}

It is useful to be able to test for row-sum normalization (A4$^\prime$) as it
enables common shocks to be accounted for in the social interactions model.
This is possible as%
\begin{eqnarray}
\Pi _{0}\iota &=&\beta _{0}\iota +(\rho _{0}\beta _{0}+\gamma
_{0})\sum_{k=1}^{\infty }\rho _{0}^{k-1}W_{0}^{k}\iota  \notag \\
&=&\left[ \beta _{0}+(\rho _{0}\beta _{0}+\gamma _{0})\sum_{k=1}^{\infty
}\rho _{0}^{k-1}\right] \iota  \notag \\
&=&\frac{\beta _{0}+\gamma _{0}}{1-\rho _{0}}\iota .  \label{eq:PiRowSum}
\end{eqnarray}%
The last equality follows from the observation that, under row-normalization
of $W_{0}$, $W_{0}^{k}\iota =W_{0}\iota =\iota $, $k>0$. This implies $\Pi
_{0}$ has constant row-sums, which suggests row-sum normalization is
testable. In the Appendix, we derive a Wald test statistic to do so.\footnote{%
For ease of explanation, in the Appendix, we derive the test under the
asymptotic distribution of the OLS estimator. The test generally holds with
minor adjustments for estimators with known asymptotic distributions.}

\subsubsection{Time-varying $W$\label{subsec:timevaryingW}}

We now relax Assumption (A6), which states that $W_0$ does not vary across
the time periods $t=1,\dots,T$. The version of Equation \eqref{eq:Model}
with time-varying network is 
\begin{eqnarray*}
y_t&=&\rho_0W_{0t}y_t+\beta_0x_t+\gamma_0W_{0t}x_t+\epsilon_t,
\end{eqnarray*}
with the reduced-form matrix $\Pi_{0t}=(I-\rho_0
W_{0t})^{-1}(I\beta_0+W_{0t}\gamma_0)$. We note that the identification results
developed in Subsection \ref{subsec:mainresults} can, in principle, be
applied element by element to each $\Pi_{0t}$, leading to the identification
of a time-varying $W_{0t}$ (and, potentially, of the parameters $\rho_0$, $%
\beta_0 $ and $\gamma_0$).

In practice, implementing any estimation strategy with a time-varying $%
\Pi_{0t}$ (or $W_{0t}$) is not feasible using only observation from the
single time period $t$. We instead adopt a kernel-weighted version. Define
period-specific weights $\omega _{t}$, and consider the transformed data $%
\tilde{y}_{s}=\omega _{s}(t)y_{s}$ and $\tilde{x}_{s}=\omega_{s}(t)x_{s}$, $%
s=1,\dots,T$. Evidently, uniform weights $\omega _{s}=1, s=1,\dots,T$ are
equivalent to the strategy not considering time-varying networks, and assuming that the networks are fixed within those windows.
Alternatively, one could estimate $W_{t}$ in time windows by setting $%
\omega_{t}=1[\underbar{t}\leq t\leq \bar{t}]$, where $\underbar{t}$ and $%
\bar{t}$ are the start and end of the time window for which $W_{t}$ is
estimated. In this case, the minimum effective window length $\bar{t}-%
\underbar{t}$ can be computed as we discuss in Section 3.1. In the context
of DSGE models with time-varying parameters, \citet{kapetanios2019time}
suggests a Gaussian kernel with positive weights throughout the entire
sample. As in nonparametric regression with smooth kernel weights, it also
assumes that the network evolves slowly over time. We further discuss this
strategy in the estimation section, and it is implemented in the empirical
application section below.

\section{Implementation}

We now transition from our core identification results to their practical
implementation. In practice, Ordinary Least Squares (OLS) can only be used to
estimate $\theta $ if $T\gg N$, which is in practice unlikely to be met, as
this is a high-dimensional problem.   Our preferred approach makes use of penalized estimation techniques that
can be used for any given $T$. More specifically, we make use of the
Adaptive Elastic Net GMM (\citealp{CanerZhang2014}), which is based on the
penalized GMM objective function. Given the identification results presented
in Section \ref{sec:Identification}, the population moments used in forming
the GMM objective function will be satisfied at the true parameter vector.

After setting out the estimation procedure, we showcase the method using
Monte Carlo simulations based on stylized and real-world network structures.
In each case, we take a fixed network structure $W_{0}$, and simulate panel
data as if the data generating process were given by the model in (\ref%
{model motivation}). We apply the method to the simulated panel data to
recover estimates of all elements in $W_{0}$, as well as the endogenous and exogenous
social effect parameters.

\subsection{Estimation\label{sec:estimation}}

The parameter vector to be estimated is high-dimensional: $\theta =\left(
W_{12},\dots ,W_{N,N-1},\rho ,\gamma ,\beta \right) ^{\prime }\in \mathbb{R}%
^{m}$, where $m=N\left( N-1\right) +3$ and $W_{ij}$ is the $(i,j)$-th
element of the $N\times N$ social interactions matrix $W_{0}$. 
\textcolor{black}{To be clear, in a network with $N$ individuals, there are
$N(N-1)$ potential interactions because an individual could interact with
everyone else but herself (which would violate Assumption A1). As a
consequence, even with a modest $N$, there are many more parameters to
estimate, and $m$ is large. For example, a network with $N=50$ implies more
than 2,000 parameters to estimate. While we consider $N$ (and thus
$m$) fixed, we still refer to $\theta$ as high-dimensional.} OLS estimation
requires $m\ll NT(\Rightarrow N\ll T)$, so many more time periods than
individuals: a requirement often met in finance data sets (%
\citealp{vanVliet2018}) 
\textcolor{black}{or in other fields (see, e.g.,
Section 4.2 in \citealp{rothenhausler2015})}. Instead, to estimate a large
number of parameters with limited data, we utilize high-dimensional
estimation methods, which are the focus of a rapidly growing literature.

Sparsity is a key assumption underlying many high-dimensional estimation
techniques. In the context of social interactions, we say that $W_{0}$ is
sparse if $\tilde{m}$, the number of non-zero elements of $W_{0}$, is such
that $\tilde{m}\ll NT$. The notion of sparsity thus depends on the number of
time periods. Sparsity corresponds to assuming that individuals influence or
are influenced by a small number of others, relative to the overall size of
the potential network and the time horizon in the data. As such, sparsity is
typically \emph{not} a binding constraint in social networks analysis.%
\footnote{%
Common stylized networks are sparse, such as the star, lattice (each
individual is a source of spillover only to one other individual), or
interactions in pairs, triads or small groups (\citealp{DeGiorgietal2010}%
). Real-world economic networks are also sparse. The sparsity in 
\textit{AddHealth} friendship network is around $98$\%. Sparsity of the
production networks in the US\ is above $99$\% (\citealp{Atalayetal2011}).}

In the estimation of sparse models, the \textquotedblleft effective number
of parameters\textquotedblright\ (or \textquotedblleft effective degrees of
freedom\textquotedblright ) relates to the number of variables with non-zero
estimated coefficients (\citealp{TibshiraniTaylor2012}). In the context of
the current social network model, this is equivalent to $m$ parameters,
where $m=dN(N-1)+K$ and $d$ is the network density defined as $\tilde{m}/(N(N-1))$. The Adaptive Elastic Net
GMM estimator presented by \citet{CanerZhang2014} converges at a rate of $%
\sqrt{NT/\tilde{m}}=\sqrt{NT/[dN(N-1)+K]}=O(\sqrt{T/(dN)})$ (see remark 7 in %
\citealp{CanerZhang2014}). Hence the quality of the large sample results
relies on a comparison between $T$ and $dN$. In line with this, we thus
require $NT\gg dN(N-1)+K$. For example, in the high-school network of %
\citet{Coleman1964} that is part of our simulation exercise, $N=70$ and $%
d=0.076$. Assuming $K=3$, $N(N-1)\times d+3=370.1$.\footnote{%
As pointed out by a referee, variation in $x$ will also matter for
estimation precision. This is reflected in the asymptotic distribution for
this estimator, shown later in this subsection.}

Finally, to reiterate, our identification results themselves do \emph{not}
depend on the sparsity of networks. In particular, Assumptions (A1)-(A6) 
\emph{do not} impose restrictions on the number of links in $W_{0}$, or $%
\tilde{m}$.\footnote{%
If $N\rightarrow \infty $, Assumption (A2) would imply vanishing $%
(W_{0})_{ij}$ entries. As highlighted previously, we consider $N$ to be
fixed, in line with many practical applications. Furthermore, Assumption
(A2) is used to represent inverse matrices as Neumann series in our
identification results. What is necessary for this to hold is that a
sub-multiplicative norm on $\rho W$ be less than one. Here we use a specific
norm (i.e., the maximum row-sum norm), but other (induced) norms are also
possible (i.e., the 2-norm or the 1-norm) (see \citealp{HornJohnson2013},
Chapter 5.6).} The identification results presented in Section \ref%
{sec:Identification} apply more broadly and irrespective of the estimation
procedure.

Our preferred approach estimates the interaction matrix in the reduced form
while penalizing and imposing sparsity on the structural object $W_{0}$. 
\textcolor{black}{We impose sparsity and penalization in the structural-form matrix $W_{0}$ because this is a weaker requirement than imposing sparsity and
penalization in the reduced-form matrix $\Pi _{0}$.}\footnote{%
\textcolor{black}{Note that even if $W$ is sparse, $\Pi$ may not be sparse.}
In Appendix C.1, we show that $[\Pi _{0}]_{ij}=0$ if, and only if, there are no paths
between $i$ and $j$ in $W_{0}$, so the pair is not connected. So,
sparsity in $\Pi _{0}$ is understood as $W_{0}$ being ``sparsely connected'',
which is a stronger assumption than sparsity in $W_{0}$.} To do so, we use
the Adaptive Elastic Net GMM (\citealp{CanerZhang2014}), which is based on
the penalized GMM objective function,%
\begin{equation}
G_{NT}(\theta ,p)\equiv g_{NT}\left( \theta \right) ^{\prime
}M_{T}g_{NT}\left( \theta \right) +p_{1}\sum_{\substack{ i,j=1  \\ i\neq j}}%
^{N}\left\vert W_{i,j}\right\vert +p_{2}\sum_{\substack{ i,j=1  \\ i\neq j}}%
^{N}\left\vert W_{i,j}\right\vert ^{2}  \label{eq:enetGMMs1}
\end{equation}%
where $\theta =(W_{1,2},\dots ,W_{N,N-1}\,\rho ,\gamma ,\beta )^{\prime }$
with dimension $m=N(N-1)+3$, and $p_{1}$ and $p_{2}$ are the penalization
terms. The term $g_{NT}\left( \theta \right) ^{\prime }M_{T}g_{NT}\left(
\theta \right) $ is the unpenalized GMM objective function with moment
conditions based on orthogonality between the structural disturbance term
and the covariates: $g_{NT}\left( \theta \right) =\sum_{t=1}^{T}\left[
x_{1t}e_{t}(\theta )^{\prime }\;\cdots \;x_{Nt}e_{t}(\theta )^{\prime }%
\right] ^{\prime }$, $e_{t}(\theta )=y_{t}-\rho Wy_{t}-\beta x_{t}-\gamma
Wx_{t}$. There are $q\equiv N^{2}$ moment conditions since $x_{it}$ is
orthogonal to $e_{jt}$ for each $i,j=1,\dots ,N$. Hence, the GMM weight
matrix $M_{T}$ is of dimension $N^{2}\times N^{2}$, symmetric, and positive
definite. For simplicity, we use $M_{T}=I_{N^{2}\times N^{2}}$. Note that if 
$x_{t}$ is econometrically endogenous, one can also exploit moment
conditions with respect to available instrumental variables.\footnote{%
For expositional ease, we describe estimation in the context of the reduced-form model (\ref{eq2:modelRF}), thereby abstaining from individual fixed or
correlated effects. As the GMM estimator uses moments between the structural
disturbance terms and covariates, this endogeneity is built into the
estimation procedure.} Given the identification results presented in Section %
\ref{sec:Identification}, if $\theta \neq \theta _{0}$ and does not belong
to the identified set, then $\Pi (\theta )\neq \Pi (\theta _{0})$.
Consequently, the populational version of the GMM objective function is
uniquely minimized at the true parameter vector $\theta _{0}$.

The penalization terms in Equation (\ref{eq:enetGMMs1}) are what makes this
different from a standard GMM problem. The first term, $p_{1}\sum_{{%
i,j=1,i\neq j}}^{N}\left\vert W_{i,j}\right\vert $, penalizes the sum of the
absolute values of $W_{ij}$, i.e., the sum of the strength of links, for all
node-pairs. Depending on the choice of $p_{1}$, some $W_{i,j}$'s will be
estimated as exact zeros. A larger share of parameters will be estimated as
zeros if $p_{1}$ increases. The second term, $p_{2}\sum_{{i,j=1,i\neq j}%
}^{N}\left\vert W_{i,j}\right\vert ^{2}$, penalizes the sum of the square of
the parameters. This term has been shown to provide better model-selection
properties, especially when explanatory variables are correlated (%
\citealp{ZouZhang2009}). The first-stage estimate is%
\begin{equation}  \label{eq:1ststage}
\tilde{\theta}(p)=(1+p_{2}/T)\cdot \underset{\theta \in \mathbb{R}^{m}}{\arg
\min }\;\;G_{NT}(\theta ,p)
\end{equation}%
where $(1+p_{2}/T)$ is a bias-correction term also used by %
\citet{CanerZhang2014}.

Implementing the numerical optimization embedded in Equation %
\eqref{eq:1ststage} is computationally challenging, as $m=N(N-1)+3$ may
entail a large number of function arguments. We instead implement the
following modification to use fast Least-Angle Regression (LARS) algorithms (%
\citealp{efron2004least}). For any given $\rho $, $\beta$, and $\gamma $,
the expression for $e_{t}(\theta )$ is linear in $W$:%
\begin{equation*}
e_{t}(\theta )=y_{t}-x_{t}\beta -W(\rho y_{t}+x_{t}\gamma )\;\;=\;\;\tilde{y}%
_{it}(\beta )-W\tilde{x}_{t}(\rho ,\gamma )
\end{equation*}%
where $\tilde{y}_{it}(\beta )\equiv y_{t}-x_{t}\beta $ and $\tilde{x}%
_{t}(\rho ,\gamma )\equiv \rho y_{t}+x_{t}\gamma $ and, following the
strategy above, is instrumented with $x_{t}$. This motivates a two-step
optimization routine: 
\begin{equation*}
\underset{\theta \in \Theta =\Theta _{1}\times \Theta _{2}}{\min }%
G_{NT}(\theta ,p)=\underset{(\rho ,\beta ,\gamma )\in \Theta _{1}}{\min }\;\;%
\left[ \underset{W_{ij}\in \Theta _{2}}{\min }G_{NT}(\theta ,p)\right] ,
\end{equation*}%
where the expression in brackets has a computationally efficient solution
through the LARS algorithm. The numerical optimization is then subsequently
conducted over the parameter space of $(\rho ,\beta ,\gamma )$ only. We also
impose row-sum normalization. Details of the implementation are expanded in
Appendix Subsection C.2.

A second (adaptive) step provides improvements by re-weighting the
penalization by the inverse of the first-step estimates (\citealp{Zou2006}):%
\begin{equation}
\tilde{\theta}^*(p)=\left( 1+p_{2}/T\right) \cdot \underset{\theta \in 
\Theta}{\arg \min }\;\left\{ g_{NT}\left( \theta \right) ^{\prime
}M_{T}g_{NT}\left( \theta \right) +p_{1}^{\ast }\sum_{\substack{ \{i,j:%
\tilde{W}_{ij}\neq 0,  \\ i,j=1,\dots ,N,  \\ i\neq j\}}}\frac{|W_{i,j}|}{|%
\tilde{W}_{i,j}|^{c }}+p_{2}\sum_{\substack{ \{i,j:\tilde{W}_{ij}\neq 0,  \\ %
i,j=1,\dots ,N,  \\ i\neq j\}}}\left\vert W_{i,j}\right\vert ^{2}\right\} ,
\label{eq:enetGMMs2}
\end{equation}%
where $\tilde{W}_{i,j}$ is the $\left( i,j\right) $-th element of the
first-step estimate of $W$. We follow \citet{CanerZhang2014} and set $c =2.5$%
. If $\tilde{W}_{i,j}<0.05$, we set $\tilde{W}_{i,j}=0.05$. This ensures
that the second-stage estimates can be non-zero even if the first-stage estimates
were zero or small. The computational improvement -- described above for the
first-stage estimator -- is also applied in the adaptive stage.

As a third and final step, we fix the support of $\tilde{\theta}^{\ast }(p)$%
, $\mathcal{S}=\{\rho ,\beta ,\gamma \}\cup \{W_{ij}:\tilde{W}_{ij}^{\ast
}\neq 0\}$ and estimate the final parameters without penalization. This
takes as arguments only the elements of $\tilde{\theta}^{\ast }(p)$ that
were estimated as non-zero in the adaptive step. In essence, this step boils
down to a standard GMM approach, 
\begin{equation}
\hat{\theta}_{\mathcal{S}}(p)=\underset{\theta \in \mathcal{S}}{\arg \min }%
\;\left\{ g_{NT}\left( \theta \right) ^{\prime }M_{T}g_{NT}\left( \theta
\right) \right\}.  \label{eq:enetGMMs3}
\end{equation}%
Importantly, \citet{CanerZhang2014} show that the third-step estimator is
asymptotically normal, with a known and easy-to-compute distribution, 
\begin{equation*}
\delta ^{\prime }\left[ \big(\hat{G}^{\prime }M_{T}\hat{G}\big)^{-1}\cdot %
\big(\hat{G}^{\prime }M_{T}\Omega M_{T}\hat{G}\big)\cdot \big(\hat{G}%
^{\prime }M_{T}\hat{G}\big)^{-1}\right] ^{1/2}\cdot \sqrt{NT/\tilde{m}}\cdot
(\hat{\theta}_{\mathcal{S}}-\theta _{0})\overset{d}{\longrightarrow }N(0,1),
\end{equation*}%
where $\hat{G}\equiv \hat{G}(\hat{\theta})=\nabla g_{NT}(\theta )$ and $%
\Omega \equiv E[g_{NT}(\theta )g_{NT}(\theta )^{\prime }]$.\footnote{%
This applies in the case of small $p_{2}$. In the case of large $p_{2}$,
the asymptotic distribution is pre-multiplied by $K_{n}=\frac{I+p_{2}\left[ 
\hat{G}(\hat{\theta})^{\prime }\hat{\Omega}^{-1}\hat{G}(\hat{\theta})\right]
^{-1}}{1+p_{2}/{NT}}$. See Theorem 4 of \citet{CanerZhang2014}.} This allows
us to conduct hypothesis testing and inference on the $\rho $, $\beta $, $%
\gamma $ and the non-zero elements of $W$.

We write $p=(p_{1},p_{1}^{\ast },p_{2})$ as the final set of penalization
parameters. Conditional on $p$, the estimate of the procedure is $\hat{\theta%
}(p)$. As in \citet[p.
35]{CanerZhang2014}, the penalization parameters $p$ are chosen by the BIC
criterion. This balances model fit with the number of parameters included in
the model.\footnote{\textcolor{black}{Following \citealp{CanerZhang2014},}
the choice of $p$, which we denote as $\hat{p}$, is the one that minimizes 
\begin{equation*}
\text{BIC}(p)=\log \left[ g_{NT}\Big(\hat{\theta}(p)\Big)^{\prime
}M_{T}g_{NT}\Big(\hat{\theta}(p)\Big)\right] +A\Big(\hat{\theta}(p)\Big)%
\cdot \frac{\log T}{T}
\end{equation*}%
where $A\Big(\hat{\theta}(p)\Big)$ counts the number of non-zero
coefficients among $\{W_{1,2},\dots ,W_{N,N-1}\}$, and larger than a
numerical tolerance, which we set at $10^{-5}$. See also %
\citet{zouhastietibshirani2007}.}

In Appendix C.2, 
we provide further implementation details, including the choice of initial
conditions. Of course, other estimation methods are available, and our
identification results do not hinge on any particular estimator. Our aim is
to demonstrate the practical feasibility of using the Adaptive Elastic Net
estimator rather than claim it is the optimal estimator.\footnote{%
See the alternative approaches of \citet{GautierTsybakov2014}, %
\citet{Manresa2016}, \citet{lamsouza2016}, and \citet{GautierRose2016}.} 

\subsection{Simulations}

\label{subsec:simulation}

We showcase the method using Monte Carlo simulations. We describe the
simulation procedures, results, and robustness checks in more detail in 
Appendix D.1. Here, we just provide a brief overview to highlight how well
the method works to recover social networks even in relatively short panels.

For each simulated network, we take a fixed network structure $W_{0}$ and
simulate panel data as if the data generating process were given by (\ref%
{model motivation}). We then apply the method to the simulated panel data to
recover estimates of all elements in $W_{0}$, as well as the endogenous and
exogenous social effect parameters ($\rho _{0}$, $\gamma _{0}$). Our result
identifies entries in $W_{0}$ and so naturally recovers links of varying
strength. It is long recognized that link strength might play an important
role in social interactions (\citealp{Granovetter1973}). Data limitations
often force researchers to postulate some ties to be weaker than others
(say, based on interaction frequency). In contrast, our approach identifies
the continuous strength of ties, $W_{0,ij}$, where $W_{0,ij}>0$ implies node 
$j$ influences node $i$.

The stylized networks we consider are a random network and a political
party network in which two groups of nodes each cluster around a central
node. The real-world networks we consider are the high-school friendship
network in \citet{Coleman1964} from a small high school in Illinois, and one
of the village networks elicited in \citet{banerjeeetal2013} from rural
Karnataka, India.

Summary statistics for each network are presented in 
\textcolor{bluez}{Panel
A} of \textcolor{bluez}{Table A1}. 
The four networks differ in their size, complexity, and the relative
importance of strong and weak ties. For example, the Erd\"{o}s-Renyi network
only has strong ties, while the political party network has twice as many strong
as weak ties. For the real-world networks, the mean out-degree distributions
are higher, so the majority of ties are weak, with the high school network
having around $80$\% of its edges being weak ties. All four networks are
also sparse.

For the stylized networks, we assess the performance of the estimator for a
fixed network size, $N=30$. We simulate the real-world networks using
non-isolated nodes in each (so $N=70$ and $65$ respectively).\footnote{%
Like \citet{Bramoulleetal2009}, we exclude isolated nodes because they do
not conform to row-sum normalization.}

We evaluate the procedure over varying panel lengths (starting from short
panels with $T=5$), using various metrics. Given our core contribution is to
identify the social interactions matrix, we first examine the proportion of
true zero entries in $W_{0}$ estimated as zeros and the proportion of true
non-zero entries estimated as non-zeros. A global perspective of the
proximity between the true and estimated networks can be inferred from their
average absolute distance between elements. This is the mean absolute
deviation of $\hat{W}$ and $\hat{\Pi}$ relative to their true values,
defined as $MAD(\hat{W})=\frac{1}{N(N-1)}\sum_{i,j,i\neq j}|\hat{W}%
_{ij}-W_{ij,0}|$ and $MAD(\hat{\Pi})=\frac{1}{N(N-1)}\sum_{i,j,i\neq j}|\hat{%
\Pi}_{ij}-\Pi _{ij,0}|$. As these metrics are closer to zero, more of the
elements in the true matrix are correctly estimated. Finally, we evaluate
the procedure's performance using averaged estimates of the endogenous
and exogenous social effect parameters, $\hat{\rho}$ and $\hat{\gamma}$. In
keeping with the estimation strategy in our empirical application, we report
unpenalised GMM.

\subsection{Results}

\textcolor{bluez}{Figure A1} shows the simulation results as evaluated using
the six metrics described above. \textcolor{bluez}{Panel A} shows that for
each network, the proportion of zero entries in $W_{0}$ correctly estimated
as zeros is above $95$\% even when $T=5$. The proportion approaches $100$\%
as $T$ grows. Conversely, \textcolor{bluez}{Panel B} shows the proportion of
strong non-zero entries estimated as non-zeros (defined as larger than 0.3)
is also high for a small $T$. It is above $70$\% from $T=5$ for the Erd\"{o}s-Renyi network, being at least $85$\% across networks for $T=25$, and
increasing as $T$ grows. As discussed above, the Adaptive Elastic Net
estimator may only recover strong edges well, and not necessarily the weaker
ones, due to the well-known issue with shrinkage estimators that they tend to
shrink small parameters to zero. We return to this issue below.

\textcolor{bluez}{Panels C and D} show that for each simulated network, the
mean absolute deviation between estimated and true networks for $\hat{W}$
and $\hat{\Pi}$ falls quickly with $T$ and is close to zero for large sample
sizes. Finally, \textcolor{bluez}{Panels E and F} show that biases in the
endogenous and exogenous social effects parameters, $\hat{\rho}$ and $\hat{%
\gamma}$, also fall in $T$ (we do not report the bias in $\hat{\beta}$ since
it is close to zero for all $T$). The fact that biases are not zero is as
expected for a small $T$, being analogous to well-known results for
autoregressive time series models.%
\textcolor{black}{\footnote{\textcolor{black}{The bias in spatial auto-regressive models with a small
number of observations \emph{even when the network is observed} is similarly
documented by
\citet{Smith2009}, \citet{NeumanMizruchi2010}, \citet{Wangetal2014}, and others.}}}

\textcolor{bluez}{Figure A2} shows that, as $T$ increases, the procedure detects weaker links. The figure also shows that, with low sample sizes,
weak edges are generally not detected. This pattern is consistent with the
well-known fact that small parameters are likely shrunk to zero due to the
penalization (\citealp{BelloniChernozhukov2011b}). The absence of weak edges
also implies that the strength of strong edges may also be over-estimated,
since rows are normalized to one. In \textcolor{bluez}{Panel A}, we show the
distribution of the estimates of $\hat{W}_{ij}$, with $T=25$ and for the
high-school network. We show the distribution for the five most common
values of $W_{0,ij}$. We find that most edges weaker than $.5$ are not
detected; edges with a strength of $.75$ are substantially more likely to be
estimated as non-zeros. When they are detected as non-zeros, they are more
likely to be over-estimated. When we estimate $W$ with $T=150$, %
\textcolor{bluez}{Panel B} shows that virtually all edges with strength
greater than $.5$ are estimated as non-zeros, and most edges with strength $%
.375$ are also detected. We further see are more continuous distribution of
estimates of edge strength. Only edges smaller than $.25$ are not detected. %
\textcolor{bluez}{Panels C and D} show a similar conclusion for the village
network.

\textcolor{bluez}{Figure A3} shows the simulated and actual networks under $%
T=100$ time periods. The network size is set to $N=30$ in the two stylized
networks, $N=70$ for the high school network, and $N=65$ for the village
household network. In comparing the simulated and true networks, %
\textcolor{bluez}{Figure A3} distinguishes between kept edges, added edges,
and removed edges. Kept edges are depicted in blue:\ these links are
estimated as non-zero in at least $5$\% of the iterations, and are also
non-zero in the true network. Added edges are depicted in green:\ these
links are estimated as non-zero in at least $5$\% of the iterations but the
edge is zero in the true network. Removed edges are depicted in red:\ these
links are estimated as zero in at least $5$\% of the iterations but are
non-zero in the true network. \textcolor{bluez}{Figure A3} further
distinguishes between strong and weak links:\ strong links are shown as
solid edges ($W_{0,ij}>.3$), and weak links are shown as dashed edges.

\textcolor{bluez}{Panel A of Figure A3} compares the simulated and true Erd\"{o}s-Renyi networks. All links are recovered. For the political party network,
Panel B shows that all strong edges are correctly estimated. However, around
half the weak edges are recovered (blue dashed edges), with the others being
missed (red dashed edges). As discussed above, this is not surprising given
that shrinkage estimators force small non-zero parameters to zero. Hence,
a larger $T$ is needed to achieve similar performance to the other
simulated networks in terms of detecting weak links. For the more complex
and larger real-world networks, \textcolor{bluez}{Panel C} shows that in the
high-school network, the strong edges are all recovered. However, around half
the weak edges are missing (red dashed edges), and there are a relatively
small number of added edges (green edges):\ these amount to $87$ edges, or
approximately $1.9 $\% of the $4,534$ zero entries in the true high-school
network. A similar pattern of results is seen in the village network in %
\textcolor{bluez}{Panel D}: the strong edges are all recovered, and here the
majority of weak edges are also recovered.

\textcolor{bluez}{Panel B of Table A1} compares the network- and node-level
statistics calculated from the recovered social interactions matrix $\hat{W}$
to those in Panel A from the true interactions matrix $W_{0}$. The random Erd\"{o}s-Renyi network is perfectly recovered. For the political party network,
the number of recovered edges is slightly lower than in the true network ($41$
vs. $45$), and all edges are classified as strong. The mean of the in- and
out-degree distributions are slightly lower in the recovered network, and
all three nodes with the highest out-degree are correctly captured (nodes $1$%
, $11$, and $28$), include both party leaders (individuals $1$ and $11$%
). We then move to discussing the performance in the two real-world
networks. In the high-school network, 30\% of all edges are correctly
recovered, and they are all strong edges. As already noted in %
\textcolor{bluez}{Figure A2}, weak edges are not well estimated in the high-school network. This draws two main consequences. First, the average in- and
out- degrees are smaller in the recovered network relative to the true
network. Second, we over-estimate the number of strong edges ($61$ vs. $113$%
). This is a downside of row-sum normalization: because some weak edges get
estimated as zeros, the non-zeros are over-estimated so that the row adds to
one. We do, however, recover all three individuals with the highest
out-degree. Finally, in the village network, half the edges are recovered.
The same phenomena of underestimating weak and overestimating strong
edges are again observed. We again recover the three households with the
highest out-degree (nodes $16$, $35$, and $57$).

In the Appendix, we show the robustness of the simulation results to (i)\
varying network sizes and (ii)\ alternative parameter choices and richening the
structure of shocks across nodes. We also demonstrate the gains from using
the Adaptive Elastic Net GMM estimator over alternative estimators, such as
the Adaptive Lasso and OLS.

\section{Application: Tax Competition between US States}

\label{sec:application} We apply our results to shed new light on a classic
social interactions problem:\ tax competition between US states (%
\citealp{Wilson1999}). Defining competing ``neighbors'' remains the central
empirical challenge in this literature. Theory provides some guidance on the
issue through two mechanisms driving interactions across jurisdictions:\
factor mobility and yardstick competition.

On factor mobility, \citet{Tiebout1956} first argued that labor and capital
can move in response to differential tax rates across jurisdictions. Factor
mobility leads naturally to the postulated social interactions matrix being
(i) geographic neighbors, given labor mobility, and (ii) jurisdictions with
similar economic or demographic characteristics, given capital mobility (%
\citealp{Caseetal1989}).

Yardstick competition is driven by voters making comparisons between states
to learn about their own politician's quality (\citealp{Shleifer1985}). %
\citet{BesleyCase1995} formalize the idea in a model where voters use taxes
set by governors in other states to infer their own governor's quality.
Yardstick competition leads naturally to the postulated interactions matrix
being ``political neighbors'': states that voters make comparisons to.

In this application, the number of nodes and time periods is relatively
low:\ the data covers mainland US states, $N=48$, for the years 1962-2015, $T=53$%
. Our approach identifies the structure of social interactions among
``economic neighbors'', denoted ${W}_{econ}$. We contrast this against a null
that state taxes are influenced by geographic neighbors, $W_{geo}$, as shown
in Panel A of \textcolor{bluez}{Figure 1A}. With $W_{econ}$ recovered, we
can establish, beyond geography, what predicts the strength of ties between
states and provide fresh insights on drivers of tax competition.

Before using the real data, we confirm the estimator's performance when the
true network is $W_{geo}$ in simulated settings. In line with the findings
of the previous section, Panel B of Figure 1A shows that (i) the procedure recovers strong edges frequently (more specifically, $89$\% of the true
strong edges are recovered) and (ii)\ performance deteriorates when recovering
weak edges ($72$\%). In all cases, the estimator does not add edges not in
the true network. This suggests recovered economic links that deviate from
geographic links may indeed carry signal, while weak links may not get
detected. Finally, the estimator for $\rho $ and $\gamma $ may show some
downward bias with the sample sizes in the application, consistent with the
simulations in \textcolor{bluez}{Appendix Figure A1}.

\subsection{Data and Empirical Specification}

We denote state tax liabilities for state $i$ in year $t$ as $\tau _{it}$,
covering state taxes collected from real per capita income, sales, and
corporate taxes. We extend the sample used by \citet{BesleyCase1995}, that
runs from 1962-1988 ($T=26$).\footnote{\citet{BesleyCase1995} test their
political agency model using a two-equation set-up: (i) on gubernatorial
re-election probabilities and (ii)\ on tax setting. Our application focuses
on the latter because this represents a social interaction problem. They use
two tax series: (i) TAXSIM data (from the NBER), which runs from 1977-1988 and
(ii)\ state tax liabilities series constructed from data published annually
in the Statistical Abstract of the US, that runs from 1962-1988. All their
results are robust to either series. We extend the second series.} The
outcome considered, $\Delta \tau _{it}$, is the change in tax liabilities
between years $t$ and $(t-2)$ because it might take a governor more than a
year to implement a tax program. Their model implies a standard social
interactions specification for the tax setting behavior of state governors:%
\begin{equation}
\Delta \tau _{it}=\rho_0 \sum_{j=1}^{N}W_{0,ij}\Delta \tau
_{jt}+\sum_{k=1}^{K}\sum_{j=1}^{N}W_{0,ij}x_{jkt}\gamma
_{0,k}+\sum_{k=1}^{K}\beta _{0,k}x_{ikt}+\alpha _{i}+\alpha _{t}+\epsilon
_{it},  \label{BC Estimation}
\end{equation}%
where $k=1,\dots ,K$ are the covariates for state $i$ in period $t$.
Tax setting behavior is determined by (i)\ endogenous social effects arising
through neighbors' tax changes ($\sum_{j=1}^{N}W_{0,ij}\Delta \tau _{jt});$
(ii) exogenous social effects arising through neighbors' characteristics ($%
\sum_{j=1}^{N}W_{0,ij}x_{jkt})$; and (iii)\ state $i$'s characteristics ($%
x_{ikt} $), including income per capita, the unemployment rate, and the
proportions of young and elderly in the state's population. All
specifications include state and time effects ($\alpha _{i},$ $\alpha _{t}$%
). Due to the inclusion of the time effects $\alpha _{t}$, we normalize the
rows of $W_{econ}$ to one. \textcolor{bluez}{Table A6} presents descriptive
statistics for the \citet{BesleyCase1995} sample and our extended sample.

Much of the earlier literature on tax competition has focused on endogenous
social effects and ignored exogenous social effects by setting $\gamma =0$.
Our identification result allows us to relax this restriction and estimate
the full typology of social effects described by \citet{Manski1993}. This is
important because only endogenous social effects lead to social multipliers
from tax competition, and they are crucial to identify as they can lead to a
race-to-the-bottom or suboptimal public goods provision (%
\citealp{BrennanBuchanan1980}; \citealp{Wilson1986}; %
\citealp{OatesSchwab1988}).

\subsection{Preliminary Findings}

\textcolor{bluez}{Table 1} presents our preliminary findings and comparison
to \citet{BesleyCase1995}.  Throughout this section, we refer to \lq\lq OLS estimates\rq\rq\ as the estimates of the main equation (\ref{BC Estimation}) when $W_0$ is postulated as $W_{geo}$ or $W_{econ}$ and $\rho_0$, $\gamma_{0,k}$, and $\beta_{0,k}$ are estimated by OLS.\footnote{We postulate that $W$ is $W_{econ}$ obtained by running the procedure in Section 3, retrieving $\hat{W}$, and re-running model (\ref{BC Estimation}) with $W=\hat{W}$. For such OLS estimates, we use robust standard errors and ignore the sampling uncertainty in the estimated $W_{econ}$.} \textcolor{bluez}{Column 1} shows those estimates where the postulated social interactions matrix is
based on geographic neighbors, exogenous social effects are ignored and the
panel includes all $48$ mainland states but runs only from 1962-1988 as in %
\citet{BesleyCase1995}. Social interactions influence gubernatorial tax
setting behavior:\ $\widehat{\rho }_{OLS}=.375$. \textcolor{bluez}{Column 2}
shows this to be robust to instrumenting neighbors' tax changes using the
instrument set proposed by \citet{BesleyCase1995}:\ namely, instrumenting
for $\Delta \tau _{jt}$ using geographic neighbors' lagged changes in per capita income and unemployment rates. These instruments are in the spirit of
using exogenous social effects to instrument for neighbors' tax changes. $%
\widehat{\rho }_{2SLS}$ is more than double the magnitude of $\widehat{\rho }%
_{OLS}$, suggesting tax setting behaviors across jurisdictions are strategic
complements.

\textcolor{bluez}{Columns 3 and 4} replicate both specifications over the
longer sample, confirming Besley and Case's (1995) \nocite{BesleyCase1995}
finding on social interactions to be robust. $\widehat{\rho }_{2SLS}$ is
again more than double $\widehat{\rho }_{OLS}$. The result in %
\textcolor{bluez}{Column 4} implies that for every dollar increase in the
average tax rates among geographic neighbors, a state increases its own
taxes by $64$ cents. This is similar to the headline estimate of %
\citet{BesleyCase1995}.\footnote{%
Nor is the magnitude very different from earlier work examining fiscal
expenditure spillovers. For example, \citet{Caseetal1989} find that US state
governments' levels of per-capita expenditure are significantly impacted by
the expenditures of their neighbors, with
a one-dollar increase in neighbors' expenditures leading to a seventy-cent increase in
own-state expenditures.}

\subsection{Endogenous and Exogenous Social Interactions}

We now move beyond much of the earlier literature to first establish whether
there are endogenous and exogenous social interactions in tax setting. We
first focus on the endogenous and exogenous social interaction parameters,
and in the next subsection, we detail the identified social interactions
matrix, $\hat{W}_{econ}$. To do so, we need to modify slightly how we
instrument for neighbors' tax changes:\ the instrument set proposed by %
\citet{BesleyCase1995} based on geographic neighbors' characteristics will
generally be weaker when estimating the full specification in (\ref{BC
Estimation}) because the instruments are now directly controlled for in (\ref%
{BC Estimation}). We use an Adaptive Elastic Net GMM approach, which
instruments neighbors' tax changes with the characteristics of all other
states. With the inclusion of endogenous and exogenous social effects, this
represents our preferred approach.

\textcolor{bluez}{Columns 1 and 2 of Table 2} show OLS\ and GMM estimates
for $\rho $ obtained from the Adaptive Elastic Net procedure, where we still
set $\gamma =0$ but use our preferred instrument set:\ $\widehat{\rho }%
_{GMM}=.709>$ $\widehat{\rho }_{OLS}=.649$. 
\textcolor{bluez}{Columns 3 and
4} estimate the full model in (\ref{BC Estimation}). Relative to when
exogenous social effects are assumed away ($\gamma =0$), the OLS\ and GMM
estimates of $\rho $ are smaller, but we continue to find robust evidence of
endogenous social interactions in tax setting. The specification in %
\textcolor{bluez}{Column 4} represents our preferred one: $\widehat{\rho }%
_{GMM}=.452$ (with a standard error of $.132$). This value meets the
requirements on $\rho $ in Corollaries 3 and 4 for global identification.%
\footnote{\textcolor{bluez}{Table A7} shows the full set of exogenous social
effects (so \textcolor{bluez}{Columns 1 and 2} refer to the same
specifications as \textcolor{bluez}{Columns 3 and 4 in Table 2}). Exogenous
social effects operate through economic neighbors' unemployment rate,
demographic characteristics, and their governor's age.}

\subsection{Identified Social Interactions Matrix}

\textcolor{bluez}{Figure 1B} shows how the structures of economic ($\hat{W}%
_{econ}$) and geographic networks ($W_{geo}$) differ, where connected edges
imply that two states are linked in at least one direction (state $i$
causally impacts state taxes in $j$, and/or \emph{vice versa}). This
comparison makes clear whether all states geographically adjacent to $i$
matter for its tax setting behavior and whether there are non-adjacent
states that influence its tax rate.

The left-hand panel of \textcolor{bluez}{Figure 1B}\ shows the network of
geographic neighbors (whose edges are colored blue), onto which we
superimpose edges \emph{not} identified as links in $W_{econ}$;\ dropped
edges are in red. The vast majority of geographically adjacent states are
irrelevant for tax setting behavior. The right-hand panel of %
\textcolor{bluez}{Figure 1B}\ adds new edges identified in $\hat{W}_{econ}$
that are \emph{not} part of $W_{geo}$;\ these added edges are in green and
represent non-geographically adjacent states through which social
interactions occur. For tax-setting behavior, economic distance is
imperfectly measured if we simply assume interactions depend only on
physical distance.

\textcolor{bluez}{Table 3} summarizes the comparison between $W_{geo}$ and $%
\hat{W}_{econ}$. $W_{geo}$ has $214$ edges, while $\hat{W}_{econ}$ has only $%
49$. $\hat{W}_{econ}$ and $W_{geo}$ have $9$ edges in common. Hence, the
vast majority of geographical neighbors ($205/214=96$\%) are not relevant
for tax setting. $\hat{W}_{econ}$ has $40$ edges that are absent in $W_{geo}$%
, and the identified social interactions are more spatially dispersed than
under the assumption of geographic networks. This is reflected in\ the far
lower clustering coefficient in $\hat{W}_{econ}$ than in $W_{geo}$ ($.042$
versus $.419$).\footnote{%
The clustering coefficient is the frequency of the number of fully connected
triplets over the total number of triplets. 
}

\subsection{Links and Reciprocity}

Our estimation strategy identifies the continuous strength of links, $%
W_{0,ij}$, where $W_{0,ij}>0$ is interpreted as state $j$ influencing
outcomes in state $i$. This is useful because recent developments in tax
competition theory, using insights from the social networks literature,
suggest links need not be reciprocal (\citealp{JanebaOsterleh2013}).

\textcolor{bluez}{Table 3} reveals that only $12.2$\% of edges in $\hat{W}%
_{econ}$ are reciprocal (all edges in $W_{geo}$ are reciprocal by
definition). Hence, tax competition is both spatially disperse and
asymmetric. In most cases where tax setting in state $i$ is influenced by
taxes in state $j$, the opposite is not true.

Given common time shocks $\alpha _{t}$ in (\ref{BC Estimation}), row-sum
normalization is required and ensures $\sum_{j}W_{0,ij}=1$. Hence, for every
state $i$, there will be at least one economic neighbor state $j^{\ast }$
that impacts it, so $W_{0,ij^{\ast }}>0$. This just reiterates that social
interactions matter. On the other hand, our procedure imposes no restriction
on the derived columns of $\hat{W}_{econ}$. It could be that a state does
not affect any other state. To see this in more detail, the final rows of %
\textcolor{bluez}{Table 3} report the degree distribution across states,
splitting for in-networks and out-networks. In $W_{geo}$, the in-degree is
by construction equal to the out-degree, as all ties are reciprocal. The
greater sparsity of the network of economic neighbors is reflected in the
degree distribution being lower for $\hat{W}_{econ}$ than $W_{geo}$. In $%
\hat{W}_{econ}$, the dispersion of in- and out-degree networks is very
different (as measured by the standard deviation), being nearly nine times
higher for the in-degree. Hence one reason for so few reciprocal ties
being in the economic network is that out-degree network ties are rarely
also in-degree ties.

This asymmetry in $\hat{W}_{econ}$ further suggests that some highly
influential states drive tax setting behavior in other states. To see which
states these are, \textcolor{bluez}{Figure 2} shows a histogram for the
number of out-degree links from states. Twenty states have an out-degree of
zero, so their tax rates have no direct impact on any other state's tax
setting behavior. The most influential states in terms of the highest
out-degree are Alabama (directly impacts tax setting behavior in five
other states) and South Carolina, Pennsylvania, and Montana (which each
directly impact tax setting behavior in four other states). Taking South
Carolina as an example, the four states that it directly impacts 
include its geographic neighbor, Georgia, as well as non-geographic
neighbors Missouri, Montana, and Virginia.

\subsection{Factor Mobility or Yardstick Competition?}

We use $\hat{W}_{econ}$ to shed light on the roles of factor mobility and
yardstick competition in driving tax competition. To do so, we estimate the
factors correlated with the existence of links between states $i$ and $j$ in 
$\hat{W}_{econ}$ relative to $\hat{W}_{geo}$. For state pairs with non-zero
links in either $\hat{W}_{econ}$ or $\hat{W}_{geo}$, we define a dummy
outcome $\hat{W}_{econ,ij}=1$ if a link between states $i$ and $j$ is
estimated under $\hat{W}_{econ}$ and $\hat{W}_{econ,ij}=0$ if a link between
states $i$ and $j$ exists under $\hat{W}_{geo}$ but not under $\hat{W}%
_{econ} $. We examine correlates of links using the following dyadic
regression:%
\begin{equation}
\hat{W}_{econ,ij}=\lambda _{0}+\lambda _{1}X_{ij}+\lambda _{2}X_{i}+\lambda
_{3}X_{j}+u_{ij},
\end{equation}%
estimated using a linear probability model. The elements $X_{ij},$ $X_{i},$
and $X_{j}$ correspond to characteristics of the pair of states ($i,j$), 
state $i$, and state $j$, respectively. Covariates are time-averaged over
the sample period, and robust standard errors are reported.

\textcolor{bluez}{Table 4} presents the dyadic regression results. %
\textcolor{bluez}{Column 1} controls only for the distance between states $i$
and $j$:$\ $this is highly predictive of an economic link between them. This
reflects that the economic network of state $i$ often comprises states are
in the same region, but not necessarily contiguous to state $i$. %
\textcolor{bluez}{Column 2} adds two $X_{ij}$ covariates to capture the
economic and demographic homophily between states $i$ and $j$. GDP homophily
is the absolute difference in the states' GDP per capita. Demographic
homophily is the absolute difference in the share of young people (aged
5-17) plus the absolute difference of the share of elderly people (aged 65+)
across the states. GDP\ homophily does not predict economic ties, whereas
demographic homophily does.

\textcolor{bluez}{Columns 3-5} then sequentially add in several sets of
controls. For labor mobility, we use net state-to-state migration data to
control for the net migration flow of individuals from state $i$ to state $i$
(defined as the flow from $i$ to $j$ minus the flow from $j$ to $i$).%
\footnote{%
We also experimented with alternative measures of labor migration, and
the results were qualitatively the same. State-to-state migration data are based
on year-to-year address changes reported on individual income tax returns
filed with the IRS. The data cover filing years 1991 through 2015 and
include the number of returns filed, which approximates the number of
households that migrated, and the number of personal exemptions claimed, which
approximates the number of individuals who migrated. The data are available
at https://www.irs.gov/statistics/soi-tax-stats-migration-data (accessed
September 2017).} We then add a political homophily variable between states.
For any given year, this is set to one if a pair of states have governors of
the same political party. As this is time-averaged over our sample, this
element captures the share of the sample in which the states have governors
of the same party. Lastly, we include whether state $j$ is considered a tax
haven (and so might have disproportionate influence on other states). Based
on \citet{Findleyetal2012}, the following states are coded as tax havens:
Nevada, Delaware, Montana, South Dakota, Wyoming, and New York.

\textcolor{bluez}{Column 5} shows that with this full set of controls,
distance remains a robust predictor of the existence of economic links
between states. However, the identified economic network highlights
additional significant predictors of tax competition between states:\
political homophily \emph{reduces} the likelihood of a link, suggesting any
yardstick competition driving social interactions occurs when voters compare
their governor to those of the opposing party in other states. Tax haven
states appear to be especially less influential in the tax setting behaviors
of other states. This mirrors what was observed in 
\textcolor{bluez}{Figure
2}, where some of the prominent tax havens -- Nevada, Delaware, and New York
-- were all identified to have zero out-degree links. The relatively weak
influence of tax haven states eases concerns over a race-to-the-bottom in
tax setting behaviors.

\textcolor{bluez}{Column 6} controls for state $i$ and state $j$ fixed
effects. This reinforces the idea that distance and political homophily
correlate to the strength of influence states tax setting has on others (the
tax haven dummy cannot be separately identified in this specification).
Labor mobility between states does not robustly predict the existence of
economic ties.

\subsection{Dynamics}

As in our identification result, our empirical approach has taken the
network structure as fixed over the entire sample period. In the context of
tax competition over our study period, this might be a strong assumption. We
examine the issue in more detail by allowing the estimated $W_{econ}$ matrix
to vary over time by changing the weight placed on observations from any
given time period. More precisely, for any given time period $t$, we weight
observations using a Gaussian kernel with its center varying
period-by-period from $1962$ to $2015$. The variance of the kernel is set
such that $75$\% of the weight is given to the first half of the data (i.e.,
pre-$1988$) when the kernel is centered in $1962$. 
\textcolor{bluez}{Figure
A4} shows the kernel employed as we vary its center: the solid kernel is
centered in $1962$, the start of our sample -- when we place the most weight
on observations from $1962$. The static case considered previously is akin
to using a uniform kernel over all periods. This kernel weighting approach
is outlined in Section \ref{subsec:timevaryingW}. We fully describe the
algorithm in Appendix Section C.2.

We begin by considering time-varying estimates of the endogenous and
exogenous social interaction parameters from the full model in (\ref{BC
Estimation}). The results for the endogenous social interactions parameter
are shown in \textcolor{bluez}{Figure 3}, where the shaded areas are the $95$\%
confidence intervals of the period-by-period estimates. Panel A shows that OLS\
estimates of $\widehat{\rho }$ drift up over time, so the strength of
endogenous interactions increases from around $.35$ in the late 1960s to $.50
$ by the 2010s. In all periods, we can reject the null that the
endogenous social effect is zero. Recall the earlier static estimate was $%
\widehat{\rho }_{OLS}=.375$.

Panel B shows the estimated endogenous social effect when we use GMM based
on the characteristics of all other states as IVs. This also drifts up from
around $.35$ in the late 1960s to $.50$ by the 2010s. In the majority of
periods, we can reject the null of no endogenous social effect.\footnote{%
Standard errors are estimated without imposing the restriction that
parameters vary slowly over time, and fluctuations across periods reflect
variations in the network across periods.}

\textcolor{bluez}{Figure 4} shows the evolution of $\hat{W}_{econ}$ over
time as we center the kernel on different periods, following the same color-coding as in \textcolor{bluez}{Figure 1}. In all periods, geography-based
edges play little role, and over time, the economic network becomes denser.
This highlights\ not only that economic networks for tax competition always
differ starkly from geography-based networks, but that the nature of
economic networks relevant to tax competition has changed steadily over
time.

\textcolor{bluez}{Figure 5} shows how the features of $\hat{W}_{econ}$ evolve as
we place greater weight from early to later periods. For each statistic, we
plot the period-by-period estimate when we center the kernel in any given
 period. The resulting smoothed estimates are then shown. To ease
exposition, networks edges with $W_{0,ij}<1/47$ are removed. This cutoff is
chosen as, in theory, states can only link at maximum with $47$ other
states. \textcolor{bluez}{Panel A} shows the share of edges that are kept
from the previous estimate (centered in the previous period). We see
relatively high stability in $\hat{W}_{econ}$ with the smoothed estimate
suggesting more than $60$\% of edges always being kept from one estimate to
the next, with this stability increasing from the late 1980s.

\textcolor{bluez}{Panel B} shows how the overlap between $\hat{W}_{econ}$
and $W_{geo}$ varies over time, as measured by the share of edges that are
only present in $\hat{W}_{econ}$. There is little overlap between the two
networks over the entire sample. The smoothed estimate suggests that at
least $80$\% of identified edges in $\hat{W}_{econ}$ are never in $W_{geo}$.
The divergence between economic and geographic neighbors becomes starker
from the mid-$1980$s onwards.

\textcolor{bluez}{Panels C and D} show how the clustering and reciprocity of
links in $\hat{W}_{econ}$ vary as we shift the weight to later observations.
Clustering of $\hat{W}_{econ}$ increases from the $1960$s through to the
early $2000$s. Thereafter, social interactions in tax competition become
sparser. We also observe a reversal in the extent to which social
interactions are reciprocal, with reciprocity rising to a peak in the early
1980s -- when $20$\% of ties were reciprocal -- and slowly falling
thereafter.

Taken together, the results suggest the nature of tax competition between
US\ states has changed over time through two mechanisms: (i) the strength of
endogenous social interactions ($\widehat{\rho }$) has increased over time and
(ii) the\ network of states interacted with ($\hat{W}_{econ}$) varies
over time. This has important implications for policy evaluation:\ the same
intervention might have different spillover effects if implemented at
different moments in time due to the evolution of $\widehat{\rho }$ and $%
\hat{W}_{econ}$. We consider this next using counterfactual simulations.

\subsection{Counterfactuals}

We use a counterfactual exercise to contrast how shocks to tax setting in a
given state propagate under $\hat{W}_{econ}$, relative to what would have
been predicted under $W_{geo}$. We do so for both static and dynamic estimates
of $\hat{W}_{econ}$. We focus on South Carolina (SC), a state with one of the
highest out-degree, as shown in \textcolor{bluez}{Figure 2}. We consider a
scenario in which SC exogenously increases its taxes per capita by $10$\%.
We measure the differential change in equilibrium state taxes in state $j$
under the two network structures using the following statistic:%
\begin{equation}
\Upsilon _{j}=\log (\Delta \tau _{jt}|\hat{W}_{econ})-\log (\Delta \tau
_{jt}|W_{geo}),  \label{Impulse}
\end{equation}%
so that positive (negative) values imply equilibrium taxes being higher
(lower) under $\hat{W}_{econ}$.\footnote{%
For $W_{geo}$, we calculate the counterfactual at $\hat{\rho}_{GMM}=.452$,
the endogenous effect parameter estimated in our preferred specification, %
\textcolor{bluez}{Column 4 of Table 2}.}

Starting with the static case, \textcolor{bluez}{Panel A of Figure 6} shows
for each mainland US state the spillover effects through the economic
network of tax competition. This highlights positive spillovers on tax rates
in many states that are not geographic neighbors of SC. 
\textcolor{bluez}{Panel
B} graphs $\Upsilon _{j}$ to make precise how spillovers derived from $\hat{W%
}_{econ}$ diverge from those predicted under $W_{geo}$. In $26$ states, $%
\Upsilon _{j}$ is smaller than $.01$\% because both networks predict
negligible spillovers to those states. In the remaining $22$ mainland
states, there is a wide discrepancy between the equilibrium state tax rates
predicted under $\hat{W}_{econ}$ relative to $W_{geo}$: $\Upsilon _{j}$
varies from $-1$ to $4.03$. The long-run effect in SC itself is also higher
under $\hat{W}_{econ}$ than under $W_{geo}$. The former states that given feedback
effects, the long-run increase in tax rates in SC from a $10$\% increase is $%
11.4$\%, while the geographically based network implies a smaller equilibrium
increase of $10.3$\%.

As $\hat{W}_{econ}$ is spatially more dispersed than $W_{geo}$, the general
equilibrium effects are different under the two network structures. %
\textcolor{bluez}{Table 5} summarizes the general equilibrium implications
for tax inequality under $\hat{W}_{econ}$ and the $W_{geo}$ counterfactual.
The average tax rate increase under $\hat{W}_{econ}$ is three times that
estimated under $W_{geo}$. Moreover, the dispersion of tax rates across
states increases under $\hat{W}_{econ}$ relative to $W_{geo}$. Finally,
assuming interactions are based solely on geographic neighbors, we miss the
fact that many states have relatively small tax increases.

We can repeat the exercise using the dynamically estimated economic network.
Throughout, we calculate the general equilibrium effects of the \emph{same}
policy experiment: SC increasing its taxes per capita by $10$\%. These
general equilibrium effects vary over the sample period because the strength
of social interactions in tax competition vary ($\widehat{\rho }_{GMM}$), as
shown in \textcolor{bluez}{Figure 3}, and identified economic neighbors vary
over time ($\hat{W}_{econ}$), as shown in \textcolor{bluez}{Figure 4}. The
results are summarized in \textcolor{bluez}{Figure 7}. Placing weight on the
early or later part of the sample generates similar changes in average tax
rates and their variance in general equilibrium -- with both being lower
than simulated under the static model. Placing more weight on the middle of
the sample period generates higher changes in average tax rates and their
variance in general equilibrium.

The differential general equilibrium impacts found as we place different
weights across sample observations links to recent discussions on the external
validity of internally valid causal impacts based on micro-evidence. While
the earlier literature has emphasized the potential interaction of treatment
effects with aggregate shocks (\citealp{rosenzweig2020external}) or how
behavioral responses to social insurance policies vary over the business
cycle (\citealp{kroft2016should}), our analysis provides another explanation
for the changing impacts of policies where social interactions determine
behavior:\ changes in the strength of social interactions and the network of
economic interactions.

\section{Discussion}

In a canonical social interactions model, we provide sufficient conditions
under which the social interactions matrix, and endogenous and exogenous social
effects are all globally identified, even absent information on social
links. Our identification strategy is novel and may bear fruit in other
areas. The method is immediately applicable to other classic social
interactions problems, but where data on social links is either missing or
partial. In fields such as macroeconomics, political economy, and trade,
there are core areas of research where social interactions across
jurisdictions/countries etc. drive key outcomes, panel data exist over many
periods, and the number of nodes is relatively fixed. Moreover, while our
discussion and application have focused on a continuous policy response
(state taxes), our methods can also be applied to the extensive margin of
policy adoption and diffusion. Such diffusion models might generate network
interactions where some states influence the later adoption of economic and
social policies in other jurisdictions. This issue is studied by %
\citet{dellavigna2022policy} in the context of US\ state policies -- they
examine the diffusion of over $700$ policies in the past $70$ years. Their
work also suggests the nature of interactions across states has changed:\
while geographic proximity is a good predictor of policy diffusion, they
also find that since $2000$, political alignment across states has become the
strongest predictor of diffusion.

In finance, high-frequency panel data is readily available and relevant for
the study of core research questions. For example, a long-standing question
has been whether CEOs are subject to relative performance evaluation, and if
so, what is the comparison set of firms/CEOs used (\citealp{edmansgabaix2016}%
). 
More generally, our method can be readily applied to a large class of
economic questions around contagion, risk, and the fragility of economic and
environmental systems. For example, since the financial crisis of 2008, it
has become clear that linkages between actors such as firms or banks are
complex and often hidden, yet because endogenous network interactions cause
feedback loops and have multiplier effects, they can have enormous
implications for the evolution of a financial crisis or the propagation of supply
shocks in aggregate. Identifying such synchronicity is a critical first step
to putting in place policies to reduce the fragility of economic systems (%
\citealp{vanVliet2018}; \citealp{elliott2022networks}; %
\citealp{goldstein2022synchronicity}).

Advances in the availability of administrative data, data from social media,
mobile technologies, and online economic transactions all offer new
possibilities to identify social interactions with long panels or high-frequency data collection, where data on social ties will typically be
missing.

Three further directions for future research are of note. First, under
partial observability of $W_{0}$ (as in \citealp{Blumeetal2015}), the number
of parameters in $W_{0}$ to be retrieved falls quickly. Our approach can
then still be applied to complete knowledge of $W_{0}$, such as if Aggregate
Relational Data is available, and this could be achieved with potentially
weaker assumptions for identification, and in even shorter panels. To
illustrate possibilities, \textcolor{bluez}{Figure A5} shows results from a
final simulation exercise in which we assume the researcher starts with
partial knowledge of $W_{0}$. We do so for Banerjee \emph{et al.} (2013)%
\nocite{banerjeeetal2013} village family network, showing simulation results
for scenarios in which the researcher knows the social ties of the three
(five, ten) households with the highest out-degree. For comparison, we also
show the earlier simulation results when $W_{0}$ is entirely unknown. This
clearly illustrates that with partial knowledge of the social network,
performance on all metrics improves rapidly for any given $T$.

Second, we have developed our approach in the context of the canonical
linear social interactions model (\ref{model motivation}). This builds on %
\citet{Manski1993} when $W_{0}$ is known to the researcher, and the
reflection problem is the main challenge in identifying endogenous and
exogenous social effects. However, the reflection problem is functional-form
dependent and may not apply to many non-linear models (\citealp{Blumeetal2011}%
, \citealp{Blumeetal2015}). An important topic for future research is to
extend the analysis to non-linear social interaction settings. Relatedly,
the canonical social interaction model assumes that the same $W_{0}$ governs
the endogenous and exogenous channels. Despite the relaxation we propose in
Section 2.3.3, we see this as a limitation of the current method, and future
research is needed to allow for a fully flexible approach.

Finally, an important part of the social networks literature examines
endogenous network formation (\citealp{Jacksonetal2017}; %
\citealp{DePaula2017}). Our analysis allows us to begin probing the issue in
two ways. First, the kind of dyadic regression analysis in Section 4 on the
correlates of entries in $W_{0,ij}$ suggests factors driving link formation
and dissolution. Second, this leads naturally to a broad agenda going
forward, to address the challenge of simultaneously identifying and
estimating time varying models of network formation and social interaction,
all in cases where data on social networks is not required.

\singlespacing
\bibliographystyle{ectaCUSTOM}
\bibliography{bibTRIMMED}

\begin{thebibliography}{83}
\newcommand{\enquote}[1]{``#1''}
\expandafter\ifx\csname natexlab\endcsname\relax\def\natexlab#1{#1}\fi

\bibitem[\protect\citeauthoryear{Acemoglu, Carvalho, Ozdaglar, and
  Tahbaz-Salehi}{Acemoglu \emph{et~al.}}{2012}]{Acemogluetal2012}
\textsc{Acemoglu, D., V.~Carvalho, A.~Ozdaglar, and A.~Tahbaz-Salehi} (2012).
  The Network Origins of Aggregate Fluctuations. \emph{Econometrica}, 80,
  1977--2016.

\bibitem[\protect\citeauthoryear{Ambrosetti and Prodi}{Ambrosetti and
  Prodi}{1972}]{AmbrosettiProdi1972}
\textsc{Ambrosetti, A. and G.~Prodi} (1972).  On the Inversion of Some
  Differentiable Mappings with Singularities between Banach Spaces.
  \emph{Annali di Matematica Pura ed Applicata}, 93, 231--46.

\bibitem[\protect\citeauthoryear{Ambrosetti and Prodi}{Ambrosetti and
  Prodi}{1995}]{AmbrosettiProdi1995}
---\hspace{-.1pt}---\hspace{-.1pt}--- (1995). \emph{A Primer of Nonlinear
  Analysis}, Cambridge University Press.

\bibitem[\protect\citeauthoryear{Anselin}{Anselin}{2010}]{Anselin2010}
\textsc{Anselin, L.} (2010).  Thirty Years of Spatial Econometrics.
  \emph{Papers in Regional Science}, 89, 3--25.

\bibitem[\protect\citeauthoryear{Atalay, Hortacsu, Roberts, and
  Syverson}{Atalay \emph{et~al.}}{2011}]{Atalayetal2011}
\textsc{Atalay, E., A.~Hortacsu, J.~Roberts, and C.~Syverson} (2011).  Network
  Structure of Production. \emph{Proceedings of the American Mathematical
  Society}, 108, 5199--202.

\bibitem[\protect\citeauthoryear{Ballester, Calvo-Armendol, and
  Zenou}{Ballester \emph{et~al.}}{2006}]{Ballesteretal2006}
\textsc{Ballester, C., A.~Calvo-Armendol, and Y.~Zenou} (2006).  Who's Who in
  Networks. Wanted: The Key Player. \emph{Econometrica}, 74, 1403--17.

\bibitem[\protect\citeauthoryear{Banerjee, Chandrasekhar, Duflo, and
  Jackson}{Banerjee \emph{et~al.}}{2013}]{banerjeeetal2013}
\textsc{Banerjee, A., A.~G. Chandrasekhar, E.~Duflo, and M.~O. Jackson} (2013).
   The Diffusion of Microfinance. \emph{Science}, 341, 1236498.

\bibitem[\protect\citeauthoryear{Battaglini, Patacchini, and
  Rainone}{Battaglini \emph{et~al.}}{2022}]{Battaglinietal2019}
\textsc{Battaglini, M., E.~Patacchini, and E.~Rainone} (2022).  Endogenous
  social interactions with unobserved networks. \emph{The Review of Economic
  Studies}, 89, 1694--1747.

\bibitem[\protect\citeauthoryear{Belloni and Chernozhukov}{Belloni and
  Chernozhukov}{2011}]{BelloniChernozhukov2011b}
\textsc{Belloni, A. and V.~Chernozhukov} (2011).  High Dimensional Sparse
  Models: an Introduction. Tech. Rep. 1106.5242, arXiv.org Collection,
  http://arxiv.org/abs/1106.5242.

\bibitem[\protect\citeauthoryear{Besley and Case}{Besley and
  Case}{1994}]{BesleyCase1994}
\textsc{Besley, T. and A.~Case} (1994).  Unnatural Experiments? Estimating the
  Incidence of Endogenous Policies. \emph{NBER Working Paper 4956}.

\bibitem[\protect\citeauthoryear{Besley and Case}{Besley and
  Case}{1995}]{BesleyCase1995}
---\hspace{-.1pt}---\hspace{-.1pt}--- (1995).  Incumbent Behavior:
  Vote-seeking, Tax-setting, and Yardstick Competition. \emph{American Economic
  Review}, 85, 25--45.

\bibitem[\protect\citeauthoryear{Blume, Brock, Durlauf, and Ioannides}{Blume
  \emph{et~al.}}{2011}]{Blumeetal2011}
\textsc{Blume, L., W.~A. Brock, S.~N. Durlauf, and Y.~Ioannides} (2011).
  Identification of Social Interactions. in \emph{Handbook of Social
  Economics}, ed. by J.~Behabib, A.~Bisin, and M.~O. Jackson, Noth-Holland,
  vol.~1B.

\bibitem[\protect\citeauthoryear{Blume, Brock, Durlauf, and Jayaraman}{Blume
  \emph{et~al.}}{2015}]{Blumeetal2015}
\textsc{Blume, L.~E., W.~A. Brock, S.~N. Durlauf, and R.~Jayaraman} (2015).
  Linear Social Interactions Models. \emph{Journal of Political Economy}, 123,
  444--96.

\bibitem[\protect\citeauthoryear{Bonaldi, Hortacsu, and Kastl}{Bonaldi
  \emph{et~al.}}{2015}]{BonaldiHortacsuKastl2015}
\textsc{Bonaldi, P., A.~Hortacsu, and J.~Kastl} (2015).  An Empirical Analysis
  of Funding Costs Spillovers in the EURO-zone with Application to Systemic
  Risk. Princeton University Working Paper.

\bibitem[\protect\citeauthoryear{Bramoull{\'e}, Djebbari, and
  Fortin}{Bramoull{\'e} \emph{et~al.}}{2009}]{Bramoulleetal2009}
\textsc{Bramoull{\'e}, Y., H.~Djebbari, and B.~Fortin} (2009).  Identification
  of Peer Effects Through Social Networks. \emph{Journal of Econometrics}, 150,
  41--55.

\bibitem[\protect\citeauthoryear{Brennan and Buchanan}{Brennan and
  Buchanan}{1980}]{BrennanBuchanan1980}
\textsc{Brennan, G. and J.~Buchanan} (1980). \emph{The Power to Tax: Analytical
  Foundations of a Fiscal Constitution}, Cambridge University Press.

\bibitem[\protect\citeauthoryear{Breza, Chandrasekhar, McCormick, and
  Pan}{Breza \emph{et~al.}}{2020}]{Brezaetal2017}
\textsc{Breza, E., A.~G. Chandrasekhar, T.~H. McCormick, and M.~Pan} (2020).
  Using aggregated relational data to feasibly identify network structure
  without network data. \emph{American Economic Review}, 110, 2454--84.

\bibitem[\protect\citeauthoryear{Bursztyn, Ederer, Ferman, and
  Yuchtman}{Bursztyn \emph{et~al.}}{2014}]{Bursztynetal2014}
\textsc{Bursztyn, L., F.~Ederer, B.~Ferman, and N.~Yuchtman} (2014).
  Understanding Mechanisms Underlying Peer Effects: Evidence From a Field
  Experiment on Financial Decisions. \emph{Econometrica}, 82, 1273--301.

\bibitem[\protect\citeauthoryear{Caner, Han, and Lee}{Caner
  \emph{et~al.}}{2018}]{Caneretal2018}
\textsc{Caner, M., X.~Han, and Y.~Lee} (2018).  Adaptive Elastic Net GMM
  Estimation With Many Invalid Moment Conditions: Simultaneous Model and Moment
  Selection. \emph{Journal of Business and Economic Statistics}, 36, 24--46.

\bibitem[\protect\citeauthoryear{Caner and Zhang}{Caner and
  Zhang}{2014}]{CanerZhang2014}
\textsc{Caner, M. and H.~H. Zhang} (2014).  Adaptive Elastic Net for
  Generalized Method of Moments. \emph{Journal of Business and Economic
  Statistics}, 32, 30--47.

\bibitem[\protect\citeauthoryear{Case, Hines, and Rosen}{Case
  \emph{et~al.}}{1989}]{Caseetal1989}
\textsc{Case, A., J.~R. Hines, and H.~S. Rosen} (1989).  Copycatting: Fiscal
  Policies of States and Their Neighborrs. \emph{NBER Working Paper 3032}.

\bibitem[\protect\citeauthoryear{Chandrasekhar and Lewis}{Chandrasekhar and
  Lewis}{2016}]{ChandrasekharLewis2016}
\textsc{Chandrasekhar, A. and R.~Lewis} (2016).  Econometrics of Sampled
  Networks. Working Paper.

\bibitem[\protect\citeauthoryear{Chaney}{Chaney}{2014}]{Chaney2014}
\textsc{Chaney, T.} (2014).  The Network Structure of International Trade.
  \emph{American Economic Review}, 104, 3600--34.

\bibitem[\protect\citeauthoryear{Coleman}{Coleman}{1964}]{Coleman1964}
\textsc{Coleman, J.~S.} (1964). \emph{Introduction to Mathematical Sociology},
  London Free Press Glencoe.

\bibitem[\protect\citeauthoryear{Conley and Udry}{Conley and
  Udry}{2010}]{ConleyUdry2010}
\textsc{Conley, T.~G. and C.~R. Udry} (2010).  Learning About a New Technology:
  Pineapple in Ghana. \emph{American Economic Review}, 100, 35--69.

\bibitem[\protect\citeauthoryear{Dahlhaus}{Dahlhaus}{2012}]{Dahlhaus2012}
\textsc{Dahlhaus, R.} (2012).  13 - Locally Stationary Processes. in \emph{Time
  Series Analysis: Methods and Applications}, ed. by T.~{Subba Rao}, S.~{Subba
  Rao}, and C.~Rao, Elsevier, vol.~30 of \emph{Handbook of Statistics},
  351--413.

\bibitem[\protect\citeauthoryear{De~Giorgi, Pellizzari, and Redaelli}{De~Giorgi
  \emph{et~al.}}{2010}]{DeGiorgietal2010}
\textsc{De~Giorgi, G., M.~Pellizzari, and S.~Redaelli} (2010).  Identification
  of Social Interactions through Partially Overlapping Peer Groups.
  \emph{American Economic Journal: Applied Economics}, 2, 241--75.

\bibitem[\protect\citeauthoryear{de~Marco, Gorni, and Zampieri}{de~Marco
  \emph{et~al.}}{2014}]{Marcoetal2014}
\textsc{de~Marco, G., G.~Gorni, and G.~Zampieri} (2014).  Global Inversion of
  Functions: an Introduction. \emph{ArXiv:1410.7902v1}.

\bibitem[\protect\citeauthoryear{de~Paula}{de~Paula}{2017}]{DePaula2017}
\textsc{de~Paula, A.} (2017).  Econometrics of Network Models. in
  \emph{Advances in Economics and Econometrics: Theory and Applications}, ed.
  by B.~Honore, A.~Pakes, M.~Piazzesi, and L.~Samuelson, Cambridge University
  Press.

\bibitem[\protect\citeauthoryear{de~Paula, Tamer, and Richards-Shubik}{de~Paula
  \emph{et~al.}}{2018}]{dePaulaRichardsTamer2018}
\textsc{de~Paula, A., E.~Tamer, and S.~Richards-Shubik} (2018).  Identifying
  Preferences in Networks with Bounded Degree. \emph{Econometrica}, 86,
  263--288.

\bibitem[\protect\citeauthoryear{DellaVigna and Kim}{DellaVigna and
  Kim}{2022}]{dellavigna2022policy}
\textsc{DellaVigna, S. and W.~Kim} (2022).  Policy Diffusion and Polarization
  across US States. \emph{National Bureau of Economic Research}.

\bibitem[\protect\citeauthoryear{Edmans and Gabaix}{Edmans and
  Gabaix}{2016}]{edmansgabaix2016}
\textsc{Edmans, A. and X.~Gabaix} (2016).  Executive Compensation: a modern
  primer. \emph{Journal of Economic literature}, 54, 1232--87.

\bibitem[\protect\citeauthoryear{Efron, Hastie, Johnstone, and
  Tibshirani}{Efron \emph{et~al.}}{2004}]{efron2004least}
\textsc{Efron, B., T.~Hastie, I.~Johnstone, and R.~Tibshirani} (2004).  Least
  angle regression. \emph{The Annals of statistics}, 32, 407--499.

\bibitem[\protect\citeauthoryear{Elliott and Golub}{Elliott and
  Golub}{2022}]{elliott2022networks}
\textsc{Elliott, M. and B.~Golub} (2022).  Networks and Economic Fragility.
  \emph{Annual Review of Economics}, 14, 665--696.

\bibitem[\protect\citeauthoryear{Erd\"{o}s and Renyi}{Erd\"{o}s and
  Renyi}{1960}]{ErdosRenyi1960}
\textsc{Erd\"{o}s, P. and A.~Renyi} (1960).  On the Evolution of Random Graphs.
  \emph{Publ. Math. Inst. Hung. Acad. Sci}, 5, 17--60.

\bibitem[\protect\citeauthoryear{Fetzer, Souza, Vanden~Eynde, and
  Wright}{Fetzer \emph{et~al.}}{2021}]{fetzeretal2020}
\textsc{Fetzer, T., P.~C. Souza, O.~Vanden~Eynde, and A.~L. Wright} (2021).
  Security transitions. \emph{American Economic Review}, 111, 2275--2308.

\bibitem[\protect\citeauthoryear{Findley, Nielson, and Sharman}{Findley
  \emph{et~al.}}{2012}]{Findleyetal2012}
\textsc{Findley, M., D.~Nielson, and J.~Sharman} (2012).  Global Shell Games:
  Testing Money Launderers' and Terrorist Financiers' Access to Shell
  Companies. Griffith University Working Paper.

\bibitem[\protect\citeauthoryear{Gautier and Rose}{Gautier and
  Rose}{2016}]{GautierRose2016}
\textsc{Gautier, E. and C.~Rose} (2016).  Inference in Social Effects when the
  Network is Sparse and Unknown. (in preparation).

\bibitem[\protect\citeauthoryear{Gautier and Tsybakov}{Gautier and
  Tsybakov}{2014}]{GautierTsybakov2014}
\textsc{Gautier, E. and A.~Tsybakov} (2014).  High-Dimensional Instrumental
  Variables Regression and Confidence Sets. Working Paper CREST.

\bibitem[\protect\citeauthoryear{Gefang, Hall, and Tavlas}{Gefang
  \emph{et~al.}}{2023}]{gefanghalltavlas2023}
\textsc{Gefang, D., S.~G. Hall, and G.~S. Tavlas} (2023).  Identifying Spatial
  Interdependence in Panel Data with Large N and Small T. .

\bibitem[\protect\citeauthoryear{Goldstein, Kopytov, Shen, and Xiang}{Goldstein
  \emph{et~al.}}{2022}]{goldstein2022synchronicity}
\textsc{Goldstein, I., A.~Kopytov, L.~Shen, and H.~Xiang} (2022).
  Synchronicity and Fragility. \emph{Manuscript}.

\bibitem[\protect\citeauthoryear{Granovetter}{Granovetter}{1973}]{Granovetter1973}
\textsc{Granovetter, M.} (1973).  The Strength of Weak Ties. \emph{American
  Journal of Sociology}, 6, 1360--80.

\bibitem[\protect\citeauthoryear{Hastie and Tibshirani}{Hastie and
  Tibshirani}{1993}]{HastieTibshirani1993}
\textsc{Hastie, T. and R.~Tibshirani} (1993).  Varying-Coefficient Models.
  \emph{Journal of the Royal Statistical Society. Series B (Methodological)},
  55, 757--796.

\bibitem[\protect\citeauthoryear{Hoff, Raftery, and Handcock}{Hoff
  \emph{et~al.}}{2002}]{hoff2002latent}
\textsc{Hoff, P.~D., A.~E. Raftery, and M.~S. Handcock} (2002).  Latent space
  approaches to social network analysis. \emph{Journal of the american
  Statistical association}, 97, 1090--1098.

\bibitem[\protect\citeauthoryear{Holland and Leinhardt}{Holland and
  Leinhardt}{1981}]{HollandLeinhardt1981}
\textsc{Holland, P.~W. and S.~Leinhardt} (1981).  An Exponential Family of
  Probability Distributions for Directed Graphs. \emph{Journal of the American
  Statistical Association}, 76, 33--50.

\bibitem[\protect\citeauthoryear{Horn and Johnson}{Horn and
  Johnson}{2013}]{HornJohnson2013}
\textsc{Horn, R.~A. and C.~R. Johnson} (2013). \emph{Matrix Analysis},
  Cambridge University Press.

\bibitem[\protect\citeauthoryear{Jackson, Rogers, and Zenou}{Jackson
  \emph{et~al.}}{2017}]{Jacksonetal2017}
\textsc{Jackson, M., B.~Rogers, and Y.~Zenou} (2017).  The Economic
  Consequences of Social Network Structure. \emph{Journal of Economic
  Literature}, 55, 49--95.

\bibitem[\protect\citeauthoryear{Jackson and Wolinsky}{Jackson and
  Wolinsky}{1996}]{jackson1996strategic}
\textsc{Jackson, M. and A.~Wolinsky} (1996).  A Strategic Model of Social and
  Economic Networks. \emph{Journal of Economic Theory}, 71, 44--74.

\bibitem[\protect\citeauthoryear{Janeba and Osterleh}{Janeba and
  Osterleh}{2013}]{JanebaOsterleh2013}
\textsc{Janeba, E. and S.~Osterleh} (2013).  Tax and the City -- A Theory of
  Local Tax Competition. \emph{Journal of Public Economics}, 106, 89--100.

\bibitem[\protect\citeauthoryear{Kapetanios, Masolo, Petrova, and
  Waldron}{Kapetanios \emph{et~al.}}{2019}]{kapetanios2019time}
\textsc{Kapetanios, G., R.~M. Masolo, K.~Petrova, and M.~Waldron} (2019).  A
  time-varying parameter structural model of the UK economy. \emph{Journal of
  Economic Dynamics and Control}, 106, 103705.

\bibitem[\protect\citeauthoryear{Kline and Tamer}{Kline and
  Tamer}{2016}]{KlineTamer2016}
\textsc{Kline, B. and E.~Tamer} (2016).  Bayesian Inference in a Class of
  Partially Identified Models. \emph{Quantitative Economics}, 7, 329--366.

\bibitem[\protect\citeauthoryear{Krantz and Parks}{Krantz and
  Parks}{2013}]{KrantzParks2013}
\textsc{Krantz, S.~G. and H.~R. Parks} (2013). \emph{The Implicit Function
  Theorem}, Birkhauser.

\bibitem[\protect\citeauthoryear{Kroft and Notowidigdo}{Kroft and
  Notowidigdo}{2016}]{kroft2016should}
\textsc{Kroft, K. and M.~J. Notowidigdo} (2016).  Should unemployment insurance
  vary with the unemployment rate? Theory and evidence. \emph{The Review of
  Economic Studies}, 83, 1092--1124.

\bibitem[\protect\citeauthoryear{Lam and Souza}{Lam and
  Souza}{2016}]{lamsouza2016}
\textsc{Lam, C. and P.~C. Souza} (2016).  Detection and Estimation of Block
  Structure in Spatial Weight Matrix. \emph{Econometric Reviews}, 35,
  1347--1376.

\bibitem[\protect\citeauthoryear{Lam and Souza}{Lam and
  Souza}{2020}]{LamSouza2019}
---\hspace{-.1pt}---\hspace{-.1pt}--- (2020).  Estimation and selection of
  spatial weight matrix in a spatial lag model. \emph{Journal of Business \&
  Economic Statistics}, 38, 693--710.

\bibitem[\protect\citeauthoryear{Lee}{Lee}{2004}]{Lee2004}
\textsc{Lee, L.-F.} (2004).  Asumptotic Distributions of Quasi-Maximum
  Likelihood Estimators for Spatial Autoregressive Models. \emph{Econometrica},
  72, 25.

\bibitem[\protect\citeauthoryear{Lee}{Lee}{2007}]{Lee2007}
---\hspace{-.1pt}---\hspace{-.1pt}--- (2007).  Identification and Estimation of
  Econometric Models with Group Interactions, Contextual Factors and Fixed
  Effects. \emph{Journal of Econometrics}, 60, 531--42.

\bibitem[\protect\citeauthoryear{Lewbel, Qu, and Tan}{Lewbel
  \emph{et~al.}}{2022}]{Lewbeletal2019}
\textsc{Lewbel, A., X.~Qu, and X.~Tan} (2022).  Social Networks with Unobserved
  Links. \emph{Journal of Political Economy (forthcoming)}.

\bibitem[\protect\citeauthoryear{Manresa}{Manresa}{2016}]{Manresa2016}
\textsc{Manresa, E.} (2016).  Estimating the Structure of Social Interactions
  Using Panel Data. Manuscript.

\bibitem[\protect\citeauthoryear{Manski}{Manski}{1993}]{Manski1993}
\textsc{Manski, C.~F.} (1993).  Identification of Endogenous Social Effects:
  the reflection problem. \emph{The Review of Economic Studies}, 60, 531--42.

\bibitem[\protect\citeauthoryear{Meinshausen and Buhlmann}{Meinshausen and
  Buhlmann}{2006}]{MeinshausenBuhlmann2006}
\textsc{Meinshausen, N. and P.~Buhlmann} (2006).  High-Dimensional Graphs and
  Variable Selection with the Lasso. \emph{The Annals of Statistics}, 34,
  1436--1462.

\bibitem[\protect\citeauthoryear{Neuman and Mizruchi}{Neuman and
  Mizruchi}{2010}]{NeumanMizruchi2010}
\textsc{Neuman, E.~J. and M.~S. Mizruchi} (2010).  Structure and bias in the
  network autocorrelation model. \emph{Social Networks}, 32, 290--300.

\bibitem[\protect\citeauthoryear{Oates and Schwab}{Oates and
  Schwab}{1988}]{OatesSchwab1988}
\textsc{Oates, W. and R.~Schwab} (1988).  Economic Competition Among
  Jurisdictions: Efficiency-enhancing or Distortion-inducing?. \emph{Journal of
  Public Economics}, 35, 333--54.

\bibitem[\protect\citeauthoryear{Rose}{Rose}{2015}]{Rose2015}
\textsc{Rose, C.} (2015).  Essays in Applied Microeconometrics. Ph.D. thesis,
  University of Bristol.

\bibitem[\protect\citeauthoryear{Rosenzweig and Udry}{Rosenzweig and
  Udry}{2020}]{rosenzweig2020external}
\textsc{Rosenzweig, M.~R. and C.~Udry} (2020).  External validity in a
  stochastic world: Evidence from low-income countries. \emph{The Review of
  Economic Studies}, 87, 343--381.

\bibitem[\protect\citeauthoryear{Rothenberg}{Rothenberg}{1971}]{Rothenberg1971}
\textsc{Rothenberg, T.} (1971).  Identification in Parametric Models.
  \emph{Econometrica}, 39, 577--91.

\bibitem[\protect\citeauthoryear{Rothenh{\"a}usler, Heinze, Peters, and
  Meinshausen}{Rothenh{\"a}usler \emph{et~al.}}{2015}]{rothenhausler2015}
\textsc{Rothenh{\"a}usler, D., C.~Heinze, J.~Peters, and N.~Meinshausen}
  (2015).  BACKSHIFT: Learning causal cyclic graphs from unknown shift
  interventions. in \emph{Advances in Neural Information Processing Systems},
  1513--1521.

\bibitem[\protect\citeauthoryear{Sacerdote}{Sacerdote}{2001}]{Sacerdote2001}
\textsc{Sacerdote, B.} (2001).  Peer Effects with Random Assingment: Results
  for Dartmouth Roommates. \emph{The Quarterly Journal of Economics}, 116,
  681--704.

\bibitem[\protect\citeauthoryear{Shleifer}{Shleifer}{1985}]{Shleifer1985}
\textsc{Shleifer, A.} (1985).  A Theory of Yardstick Competition. \emph{Rand
  Journal of Economics}, 16.

\bibitem[\protect\citeauthoryear{Smith}{Smith}{2009}]{Smith2009}
\textsc{Smith, T.~E.} (2009).  Estimation bias in spatial models with strongly
  connected weight matrices. \emph{Geographical Analysis}, 41, 307--332.

\bibitem[\protect\citeauthoryear{Strang}{Strang}{2006}]{strang2006}
\textsc{Strang, G.} (2006). \emph{Linear algebra and its applications},
  Belmont, CA Thomson, Brooks/Cole.

\bibitem[\protect\citeauthoryear{Tibshirani and Taylor}{Tibshirani and
  Taylor}{2012}]{TibshiraniTaylor2012}
\textsc{Tibshirani, R.~J. and J.~Taylor} (2012).  Degrees of freedom in lasso
  problems. \emph{The Annals of Statistics}, 40, 1198--1232.

\bibitem[\protect\citeauthoryear{Tiebout}{Tiebout}{1956}]{Tiebout1956}
\textsc{Tiebout, C.} (1956).  A Pure Theory of Local Expenditures.
  \emph{Journal of Political Economy}, 64, 416--24.

\bibitem[\protect\citeauthoryear{van Vliet}{van Vliet}{2018}]{vanVliet2018}
\textsc{van Vliet, W.} (2018).  Connections as Jumps: Estimating Financial
  Interconnectedness from Market Data. Manuscript.

\bibitem[\protect\citeauthoryear{Wang, Neuman, and Newman}{Wang
  \emph{et~al.}}{2014}]{Wangetal2014}
\textsc{Wang, W., E.~J. Neuman, and D.~A. Newman} (2014).  Statistical power of
  the social network autocorrelation model. \emph{Social Networks}, 38, 88--99.

\bibitem[\protect\citeauthoryear{Wilson}{Wilson}{1986}]{Wilson1986}
\textsc{Wilson, J.} (1986).  A Theory of Interregional Tax Competition.
  \emph{Journal of Urban Economics}, 19, 296--315.

\bibitem[\protect\citeauthoryear{Wilson}{Wilson}{1999}]{Wilson1999}
---\hspace{-.1pt}---\hspace{-.1pt}--- (1999).  Theories of Tax Competition.
  \emph{National Tax Journal}, 52, 269--304.

\bibitem[\protect\citeauthoryear{Wooldridge}{Wooldridge}{2002}]{Wooldridge2002}
\textsc{Wooldridge, J.} (2002). \emph{Econometric Analysis of Cross Section and
  Panel Data}, The MIT Press.

\bibitem[\protect\citeauthoryear{Zhou}{Zhou}{2019}]{Zhou2019}
\textsc{Zhou, W.} (2019).  A Network Social Interaction Model with
  Heterogeneous Links. \emph{Economics Letters}, 180, 50--53.

\bibitem[\protect\citeauthoryear{Zou}{Zou}{2006}]{Zou2006}
\textsc{Zou, H.} (2006).  The Adaptive Lasso and Its Oracle Properties.
  \emph{Journal of the American Statistical Association}, 101, 1418--29.

\bibitem[\protect\citeauthoryear{Zou and Hastie}{Zou and
  Hastie}{2005}]{zou2005regularization}
\textsc{Zou, H. and T.~Hastie} (2005).  Regularization and Variable Selection
  via the Elastic Net. \emph{Journal of the Royal Statistical Society: Series B
  (Statistical Methodology)}, 67, 301--320.

\bibitem[\protect\citeauthoryear{Zou, Hastie, and Tibshirani}{Zou
  \emph{et~al.}}{2007}]{zouhastietibshirani2007}
\textsc{Zou, H., T.~Hastie, and R.~Tibshirani} (2007).  On the ``Degrees of
  Freedom'' of the LASSO. \emph{The Annals of Statistics}, 35, 2173--2192.

\bibitem[\protect\citeauthoryear{Zou and Zhang}{Zou and
  Zhang}{2009}]{ZouZhang2009}
\textsc{Zou, H. and H.~H. Zhang} (2009).  On the Adaptive Elastic-net with a
  Diverging Number of Parameters. \emph{Ann. Statist.}, 37, 1733--51.

\end{thebibliography}

\newpage

\appendix

\section{Proofs}

\label{app:proofs}

\subsubsection*{Example 1}
To see how the assumption of \citet{Bramoulleetal2009} grants identification
when $W_{0}$ is known$,$ choose constants $c_{1}$, $c_{2}$, and $c_{3}$ such
that $c_{1}I+c_{2}W_{0}+c_{3}W_{0}^{2}=0.$ Focusing on the diagonal elements of
this condition, we see that if the diagonal of $W_{0}^{2}$ is not
proportional to the diagonal of $I$, then $c_{1}=c_{3}=0$ because $\text{diag%
}(W_{0})=0$. It follows that $c_{2}=0$ if at least one (off-diagonal)
element of $W_{0}$ is non-zero. However, the converse is not true, so 
if Assumptions A1-A6 do not hold, one can construct examples where $\Pi _{0}$
does not pin down $\theta _{0}$. Take, for instance, $N=5$ with $\theta _{0}$
and $\theta $ where $\beta =\beta _{0}=1$, $\rho =1.5$, $\rho _{0}=0.5$, $%
\gamma =-2.5$, and $\gamma _{0}=0.5$:%
\begin{equation*}
W_{0}=\left[ 
\begin{array}{ccccc}
0 & 0.5 & 0 & 0 & 0.5 \\ 
0.5 & 0 & 0.5 & 0 & 0 \\ 
0 & 0.5 & 0 & 0.5 & 0 \\ 
0 & 0 & 0.5 & 0 & 0.5 \\ 
0.5 & 0 & 0 & 0.5 & 0%
\end{array}%
\right] \quad \text{and}\quad W=\left[ 
\begin{array}{ccccc}
0 & 0 & 0.5 & 0.5 & 0 \\ 
0 & 0 & 0 & 0.5 & 0.5 \\ 
0.5 & 0 & 0 & 0 & 0.5 \\ 
0.5 & 0.5 & 0 & 0 & 0 \\ 
0 & 0.5 & 0.5 & 0 & 0%
\end{array}%
\right] .
\end{equation*}%
Both $W$ and $W_{0}$ violate (A5) ($(W^{2})_{kk}=(W_{0}^{2})_{kk}=0.5$ for
any $k$), and $\rho $ violates (A2). Nonetheless, $I,W_{0}$, and $W_{0}^{2}$
are linearly independent and, likewise, so are $I,W$, and $W^{2}$. In this
case, \emph{both} parameter sets produce $\Pi =(I-\rho _{0}W_{0})^{-1}(\beta
_{0}I+\gamma _{0}W_{0})=(I-\rho W)^{-1}(\beta I+\gamma W)$. This arises even
as $W$ and $W_{0}$ represent very different network structures:\ any pair
connected under $W$ is not connected under $W_{0}$ and \textit{vice versa}.

\subsubsection*{Theorem 1}

\begin{proof}
The local identification result follows \citet{Rothenberg1971}.  \textcolor{black}{Under the assumptions in our model,} the parameter space $\Theta \subset \mathbb{R}^m$ is an open set (recall that $m = N(N-1) + 3$).  \textcolor{black}{This corresponds to Assumption I in \citet{Rothenberg1971}.} We have that
\begin{eqnarray*}
\frac{\partial\Pi}{\partial W_{ij}} & = & \rho\left(I-\rho W\right)^{-1}\Delta_{ij}\left(I-\rho W\right)^{-1}\left(\beta I+\gamma W\right)+\left(I-\rho W\right)^{-1}\gamma\Delta_{ij}\\
\frac{\partial\Pi}{\partial\rho} & = & \left(I-\rho W\right)^{-1}W\left(I-\rho W\right)^{-1}\left(\beta I+\gamma W\right)\\
\frac{\partial\Pi}{\partial\gamma} & = & \left(I-\rho W\right)^{-1}W\\
\frac{\partial\Pi}{\partial\beta} & = & \left(I-\rho W\right)^{-1},
\end{eqnarray*}
where $\Delta_{ij}$ is the $N \times N$ matrix with 1 in the $\left(i,j\right)$-th
position and zero elsewhere. Write the $N^{2}\times m$
derivative matrix $\nabla_{\Pi}\equiv\frac{\partial\text{vec}\left(\Pi\right)}{\partial\theta'}$.
By assumption, row $i$ in matrix $W$ sums up to one, incorporated through the
restriction that $ \varphi \equiv\sum_{j=1,j\neq i}^{N}W_{ij}-1=0$
for the unit-normalized row $i$. The derivative of the restriction $\varphi$ is the $m$-dimensional 
vector $\nabla_{W}' \equiv \frac{\partial\varphi}{\partial\theta'}=\left[e_{i}' \otimes \iota'_{N-1}\; 0_{1\times3} \right]$ (where $e_i$ is an $N$-dimensional vector with 1 in the $i$th component and zero otherwise).
Following Theorem 6 of \citet{Rothenberg1971}, the structural parameters
$\theta \in \Theta$ are locally identified if, and only if, the matrix $\nabla\equiv\left[\nabla_{\Pi}'\;\nabla_{W}'\right]'$
has rank $m$.\footnote{For a parameter vector to be locally identified, \citet{Rothenberg1971} requires that the derivative matrix $\nabla$ have rank $m$ at that point and that this vector be (rank-)regular.  A (rank-)regular point of the parameter space is one for which there is a neighborhood where the rank of $\nabla$ is constant (see Definition 4 in \citealp{Rothenberg1971}).  Because we show that the derivative matrix has rank $m$ at every point in the parameter space, this also guarantees that every point in the parameter space is (rank-)regular.}

If $\nabla$ does not have rank $m$, there is a nonzero vector
$\mathbf{c} \equiv \left(c_{W_{12}},\dots,c_{W_{N,N-1}},c_{\rho},c_{\gamma},c_{\beta}\right)'$
such that $\nabla \cdot \mathbf{c} = 0$.  This implies that
\begin{eqnarray}
c_{W_{12}}\frac{\partial\Pi}{\partial W_{12}}+\cdots+c_{W_{N,N-1}}\frac{\partial\Pi}{\partial W_{N,N-1}}+c_{\rho}\frac{\partial\Pi}{\partial\rho}+c_{\gamma}\frac{\partial\Pi}{\partial\gamma}+c_{\beta}\frac{\partial\Pi}{\partial\beta} & = & 0\label{eq:sysPi}
\end{eqnarray}
and, for the unit-normalized row $i$ (see A4),
\begin{eqnarray}
\sum_{j\ne i, j=1,\dots,n} c_{W_{ij}} & = & 0.\label{eq:sysW}
\end{eqnarray}
Pre-multiplying equation (\ref{eq:sysPi}) by $\left(I-\rho W\right)$
and substituting the derivatives, 
\begin{eqnarray*}
\sum_{i,j=1,i\neq j}^{N}c_{W_{ij}}\left[\rho\Delta_{ij}\left(I-\rho W\right)^{-1}\left(\beta I+\gamma W\right)+\gamma\Delta_{ij}\right]+\\
+c_{\rho}W\left(I-\rho W\right)^{-1}\left(\beta I+\gamma W\right)+c_{\gamma}W+c_{\beta}I & = & 0.
\end{eqnarray*}
Define $C\equiv\sum_{i,j=1,i\neq j}^{N}c_{W_{ij}}\Delta_{ij}$. Since the spectral radius of $\rho W$ is strictly less than one by A2, one can show (by representing $\left(I-\rho W\right)^{-1}$ as a Neumann series, for instance) that
$\left(\beta I+\gamma W\right)$ and $\left(I-\rho W\right)^{-1}$
commute.  Then, the expression above is equivalent to
\begin{eqnarray*}
\rho C\left(\beta I+\gamma W\right)\left(I-\rho W\right)^{-1}+\gamma C+c_{\rho}W\left(\beta I+\gamma W\right)\left(I-\rho W\right)^{-1}+c_{\gamma}W+c_{\beta}I & = & 0.
\end{eqnarray*}
Post-multiplying by $\left(I-\rho W\right)$, we obtain
\begin{eqnarray*}
\rho C\left(\beta I+\gamma W\right)+\gamma C\left(I-\rho W\right)+c_{\rho}W\left(\beta I+\gamma W\right)+c_{\gamma}W\left(I-\rho W\right)+c_{\beta}\left(I-\rho W\right) & = & 0
\end{eqnarray*}
which, upon rearrangement, yields
\begin{eqnarray}
\left(\gamma+\rho\beta\right)C+c_{\beta}I+\left(\beta c_{\rho}-c_{\beta}\rho+c_{\gamma}\right)W+\left(c_{\rho}\gamma-\rho c_{\gamma}\right)W^{2} & = & 0.\label{eq:sys1}
\end{eqnarray}
Because $C_{ii}=0$ and $W_{ii}=0$ (by A1), we have that $c_{\beta}+\left(c_{\rho}\gamma-\rho c_{\gamma}\right) (W^{2})_{ii}=0$
for all $i=1,\dots,N$. Since by Assumption (A5) there isn't a constant
$\kappa$ such that $\text{diag}\left(W_{0}^{2}\right)=\kappa\iota$,
then $c_{\beta}=c_{\rho}\gamma-\rho c_{\gamma}=0$. Plugging back
in (\ref{eq:sys1}), we obtain 
\begin{eqnarray*}
\left(\gamma+\rho\beta\right)C+\left(\beta c_{\rho}+c_{\gamma}\right)W & = & 0.
\end{eqnarray*}
which implies that $C=-\frac{\beta c_{\rho}+c_{\gamma}}{\gamma+\rho\beta}W$,
since $\gamma+\rho\beta\neq0$ by Assumption (A3). Taking the sum of the elements in row $i$, we get
\begin{eqnarray*}
\left(\gamma+\rho\beta\right) \sum_{j\ne i, j=1,\dots,n} c_{W_{ij}} +\left(\beta c_{\rho}+c_{\gamma}\right) & = & 0.
\end{eqnarray*}
Note that, by equation (\ref{eq:sysW}), $\sum_{j\ne i, j=1,\dots,n} c_{W_{ij}}=0$. So, $\beta c_{\rho}+c_{\gamma}=0$
and $C=-\frac{\beta c_{\rho}+c_{\gamma}}{\gamma+\rho\beta}W=0$. This implies that $c_{W_{ij}}=0$ for any $i$ and $j$.  Combining
$\beta c_{\rho}+c_{\gamma}=0$ with $c_{\rho}\gamma-\rho c_{\gamma}=0$ obtained above,
we get that $c_{\rho}\left(\rho\beta+\gamma\right)=0$. Since $\rho\beta+\gamma \ne 0$,
then $c_{\rho}=0$. Given that $\beta c_\rho + c_\gamma = 0$, it follows that $c_{\gamma}=0$.  This shows that $\theta \in \Theta$ is locally identified.	
\end{proof}

\subsubsection*{Corollary 1}

\begin{proof}
The parameter $\theta_0$ being locally identified (see Theorem 1
) implies that the set $\{\theta:\Pi(\theta) = \Pi(\theta_0)\}$ is discrete.  If it is also compact, then the set is finite.  To establish that, we now show that $\Pi$ is a proper function: the inverse image $\Pi^{-1}(K)$ of any compact set $K\subset \mathbb{R}^m$ is also compact (see \citealp{KrantzParks2013}, p. 124). 

Let $\mathcal{A}$ be a compact set in the space of $N \times N$ real matrices.  Since it is a compact set in a finite dimensional space, it is closed and bounded.  Since $\Pi$ is a continuous function of $\theta$, the pre-image of a compact set, which is closed, is also closed.  Because $\|W\|$ is bounded and $\rho \in (-1,1)$, their corresponding coordinates in $\theta \in \Pi^{-1}(\mathcal{A})$ are bounded.  Suppose the coordinates for $\beta$ or $\gamma$ in $\theta \in \Pi^{-1}(\mathcal{A})$ are not bounded.  One can find a sequence $(\theta_k)_{k=1}^\infty$ such that $|\beta_k| \rightarrow \infty$ or $|\gamma_k| \rightarrow \infty$.  

% ----
Denote the Frobenius norm of the matrix $A$ as $\left\Vert A\right\Vert $.
By the submultiplicative property $\left\Vert AB\right\Vert \leq\left\Vert A\right\Vert \cdot\left\Vert B\right\Vert $,
\begin{eqnarray*}
\left\Vert \beta I+\gamma W\right\Vert  & \leq & \left\Vert I-\rho W\right\Vert \cdot \left\Vert \left(I-\rho W\right)^{-1} \left(\beta I+\gamma W\right)\right\Vert = \left\Vert I-\rho W\right\Vert \cdot \left\Vert \Pi\right\Vert .
\end{eqnarray*}
It follows that
\begin{eqnarray*}
\frac{\left\Vert \beta I+\gamma W\right\Vert }{\left\Vert I-\rho W\right\Vert } & \leq & \left\Vert \Pi\right\Vert .
\end{eqnarray*}
Given $W$ has zero main diagonal,
\begin{eqnarray*}
\left\Vert \beta I+\gamma W\right\Vert ^{2} & = & \beta^{2}\left\Vert I\right\Vert ^{2}+\gamma^{2}\left\Vert W\right\Vert ^{2}\quad=\quad\beta^{2} N+\gamma^{2}\left\Vert W\right\Vert ^{2}.
\end{eqnarray*}
Also, $\left\Vert I-\rho W\right\Vert^2 =N+\rho^{2}\left\Vert W\right\Vert ^{2}\leq N+ \rho^2 C$,
for some constant $C\in\mathbb{R}$ by Assumption (A2).  We then
have that 
\begin{eqnarray*}
\frac{\sqrt{\beta^{2}N+\gamma^{2}\left\Vert W\right\Vert ^{2}}}{\sqrt{N+\rho^2 C}} & \leq & \left\Vert \Pi\right\Vert .
\end{eqnarray*}
By Assumption (A2), the denominator above is bounded.  Hence, $\left|\beta_k\right|\rightarrow\infty\Rightarrow\left\Vert \Pi(\theta_k)\right\Vert \rightarrow\infty$. We now use the fact that $\sum_{j}W_{ij}=1$ to show that there is
a lower bound on $\left\Vert W\right\Vert ^{2}$, and so $\left|\gamma_k\right|\rightarrow\infty\Rightarrow\left\Vert \Pi(\theta_k)\right\Vert \rightarrow\infty$.
To see this, note that
\begin{eqnarray*}
\min_{\text{s.t. }\sum_{j}W_{ij}=1}\left\Vert W\right\Vert ^{2} & \geq & \min_{\text{s.t. }\sum_{j}W_{ij}=1}\sum_{j=1}^{N}W_{ij}^{2}.
\end{eqnarray*}
Since the objective function in above is convex, it can be shown that it is minimized at $W_{ij}=\frac{1}{N-1}, j \ne i$ and, consequently, $\left\Vert W\right\Vert ^{2}\geq (N-1)\frac{1}{(N-1)^{2}}=\frac{1}{N-1}$.  Hence, if $\left|\gamma_k\right|\rightarrow\infty$, the numerator in the lower bound for $\left\Vert \Pi\right\Vert$ above also goes to infinity.  Consequently, $\mathcal{A}$ is not compact. 

Therefore, if $\mathcal{A}$ is compact, the coordinates in $\theta \in \Pi^{-1}(\mathcal{A})$ corresponding to $\beta$ and $\gamma$ are also bounded.  Hence, $ \Pi^{-1}(\mathcal{A})$ is bounded (and closed).  Consequently, it is compact.

For a given reduced-form parameter matrix $\Pi$, the set $\{\theta:\Pi(\theta) = \Pi(\theta_0)\}$ is then compact.  Since it is also discrete, it is finite.	
\end{proof}

The following lemmas are used in proving Theorem 2.

\begin{lem}
\label{lem:gamma0globallyidentified} Assume (A1)-(A6). If $\gamma_0=0$, $W_0$
is such that $(W_0)_{1,2}=(W_0)_{2,1}=1$ and $(W_0)_{ij}=0$ otherwise, with $%
\rho_0\neq0$ and $\beta_0\neq0$, then $\theta_0\in\Theta$ is identified.
\end{lem}

\begin{proof}
Take $\theta=(W_{12},\dots,W_{N,N-1},\rho,\gamma,\beta)\in \Theta$ possibily different from $\theta_0$ such that the models are observationally equivalent, so $\Pi_0=\Pi$. Then,
\begin{eqnarray*}
(I-\rho_0W_0)^{-1}(\beta_0I+\gamma_0W_0)&=&(I-\rho W)^{-1}(\beta I+\gamma W).	
\end{eqnarray*}
Since $\gamma_0=0$ and $(I-\rho W)^{-1}$ and $(\beta I+\gamma W)$ commute (see the proof for Theorem 1 
), it follows that
\begin{eqnarray*}
\Pi_0=\Pi  &\Leftrightarrow&  \beta_0(I-\rho_0W_0)^{-1}=(\beta I+\gamma W)(I-\rho W)^{-1}	
\end{eqnarray*}
or, equivalently,
\begin{eqnarray*}
\beta_0(I-\rho W)&=&(I-\rho_0W_0)(\beta I+\gamma W).
\end{eqnarray*}
This last equation implies that
\begin{eqnarray}\label{eq:lem1main}
(\beta_0-\beta)I-(\gamma+\beta_0\rho)W+\rho_0\beta W_0+\rho_0\gamma W_0W&=&0.
\end{eqnarray}
We first note that $(W_0W)_{N,N}=0$ since $(W_0)_{N,i}=0$ for any $1 \le i \le N$ and, by Assumption (A1), $(W)_{N,N}=(W_0)_{N,N}=0$. So $\beta_0=\beta$. Taking elements $(i,j)$ such that $i\geq3$ and $i\neq j$ in equation (\ref{eq:lem1main}), and using the fact that $\beta_0=\beta$, we find that $-(\gamma+\beta_0\rho)(W)_{ij}=-(\gamma+\beta\rho)(W)_{ij}=0$ for any $(i,j)$ such that $i\geq3$ and $i\neq j$. By Assumption (A3), $\gamma+\beta\rho \ne 0$ and it follows that $(W)_{ij}=0$ for any $(i,j)$ such that $i\geq3$ and $i\neq j$.  In fact, since $(W)_{i,i} = 0$ by Assumption (A1), we get that $(W)_{ij}=0$ for any $(i,j)$ such that $i\geq3$.

Using Assumption (A1) and since $\beta_0=\beta$, elements $(1,1)$ and $(2,2)$ in equation (\ref{eq:lem1main}) imply that $\rho_0\gamma (W)_{2,1}=\rho_0\gamma (W)_{1,2}=0$. Given that $\rho_0\neq0$, we get that $\gamma (W)_{2,1}=\gamma (W)_{1,2}=0$. From element $(1,2)$ in equation (\ref{eq:lem1main}), we find that $-(\gamma+\beta_0\rho)(W)_{1,2}+\rho_0\beta=0$ or, equivalently, $(\rho_0-\rho(W)_{1,2})\beta_0-\gamma (W)_{1,2}=0$. Given that $\gamma (W)_{1,2}=0$ and $\beta_0\neq0$, it must be that $\rho_0-\rho(W)_{1,2}=0$. Making the analogous argument for element $(2,1)$, we would also obtain that $\rho_0-\rho(W)_{2,1}=0$.

If both $(W)_{1,2}$ and $(W)_{2,1}$ are equal to zero, using the fact that $W_{ij}=0$ for any $(i,j)$ such that $i\geq3$, we would then obtain that $W^2$ is equal to a matrix of zeros,
which contradicts Assumption (A5). Hence, $(W)_{1,2}\neq0$ or $(W)_{2,1}\neq0$. If $(W)_{1,2}\neq0$, using the fact that $\gamma (W)_{1,2}=0$, we get that $\gamma=0$. Equivalently, if $(W)_{2,1}\neq0$, and using the fact that $\gamma (W)_{2,1}=0$, we again get that $\gamma=0$. So, in either case, $\gamma=\gamma_0=0$. 

Taking element $(1,j)$ in equation (\ref{eq:lem1main}), with $j\geq 3$, we get that $-(\gamma+\rho\beta_0)W_{1,j}+\gamma\rho_0(W)_{2,j}=-\rho\beta_0 W_{1,j}=0$. Similarly, element $(2,j)$, with $j\geq 3$ implies that $-(\gamma+\rho\beta_0)W_{2,j}+\gamma\rho_0(W)_{1,j}=-\rho\beta_0 W_{2,j}=0$.   Then, from $-\rho\beta_0 (W)_{1,j}=-\rho\beta_0 (W)_{2,j}=0$ for $j\geq 3$, it follows that $-\rho (W)_{1,j}=-\rho (W)_{2,j}=0$ since $\beta_0 \ne 0$. 

From $\rho_0-\rho(W)_{1,2}=0$, if $(W)_{1,2}\neq0$, we get that $\rho=\rho_0/(W)_{1,2}\neq0$. Equivalently, if $(W)_{2,1} \ne 0$, we get that $\rho=\rho_0/(W)_{2,1}\neq0$. Since $(W)_{1,2}\ne 0$ or $(W)_{2,1}\ne 0$, we obtain that $\rho\neq0$. Then, because $-\rho (W)_{1j}=\rho (W)_{2j}=0$ for $j\geq 3$, we have that $(W)_{1j}=(W)_{2j}=0$ for $j\geq 3$.

Given that $\rho_0-\rho(W)_{1,2}=0$, $\rho_0-\rho(W)_{2,1}=0$ and $\rho\neq0$, we obtain that $(W)_{1,2}=(W)_{2,1}=\frac{\rho_0}{\rho}$. Since $(W)_{1,j}=0$ for $j\neq2$, $(W)_{2,j}=0$ for $j\neq1$ and $(W)_{ij} = 0$ for $i \ge 3$, by Assumption (A5) we get that $(W)_{1,2}=(W)_{2,1}=1$ and $\rho=\rho_0$. Hence, $((W)_{1,2},\dots,(W)_{N,N-1},\rho,\gamma,\beta)=((W_0)_{1,2},\dots,(W_0)_{N,N-1},\rho_0,\gamma_0,\beta_0)$.
\end{proof}

\begin{lem}
\label{lem:pathconnected} Assume (A1)-(A2) and (A4)-(A6). The image of $%
\Pi(\cdot)$, for $\theta\in\Theta_{+}$, is path-connected and, therefore,
connected.
\end{lem}

\begin{proof}
Take $\theta$ and $\theta^* \in \Theta_+$.  Consider first the subvectors corresponding to the adjacency matrices $W$ and $W^*$.  Without loss of generality, let $1,\dots,N$ be ordered such that $({W}^2)_{11}>({W}^2)_{22}$. Consider the adjacency matrix $W_*$ corresponding to the network of unweighted directed connections $\{(1,2),(2,1)\}$ and $\{(3,4),(4,5),\dots,(N-1,N),(N,3)\}$.   Note that $\text{diag}(W_*^2)=(1,1,0,\dots,0)$ and this is an admissible adjacency matrix under assumptions (A1)-(A2) and (A4)-(A6).  We first show that $W$ is path-connected to $W_*$. 

Consider the path given by
\begin{eqnarray*}
W(t) &=& t W_* + (1-t){W}
\end{eqnarray*}
which implies that
\begin{eqnarray*}
(W(t)^2)_{11} &=& (1-t)^2 ({W}^2)_{11} + t^2 + (1-t)t({W}_{12} + {W}_{21})\\
(W(t)^2)_{22} &=& (1-t)^2 ({W}^2)_{22} + t^2 + (1-t)t({W}_{12} + {W}_{21}).
\end{eqnarray*}
Since $(W(t)^2)_{11}-(W(t)^2)_{22} = (1-t)^2[({W}^2)_{11}-({W}^2)_{22}]>0$ for $t\in [0,1)$ and $W(1)=W_*$, (A5) is satisfied for any matrix $W(t)$ such that $t\in [0,1]$.  Since all rows in $W_*$ sum to one and $(W_*)_{ii}=0$ for any $i$, it is straightforward to see that $W(t)$ also satisfies (A1) and (A4). Finally, $\sum_{j=1}^N|W_{ij}(t) | \le t \sum_{j=1}^N|(W_*)_{ij} | + (1-t) \sum_{j=1}^N|W_{ij} | \le 1$ for every $i=1,\dots,N$, and $W(t)$ satisfies Assumption (A2).

\bigskip

If $W^*$ is such that $(W^{*2})_{11} \neq (W^{*2})_{22}$, the convex combination of $W^*$ and $W_*$ is also seen to satisfy (A1)-(A2) and (A4)-(A6), and a path between $W$ and $W^*$ can be constructed via $W_*$.   If, on the other hand $(W^{*2})_{11} = (W^{*2})_{22}$, suppose without loss of generality that $(W^{*2})_{11} \neq (W^{*2})_{33}$.   In this case, one can construct a path between $W^*$ and $W_{**}$ where $W_{**}$ represents the network of unweighted directed connections $\{(1,3),(3,1)\}$ and $\{(2,4),(4,5),\dots,(N-1,N),(N,2)\}$.

Like $W(t)$ above, this path can be seen to satisfy assumptions (A1)-(A2) and (A4)-(A6). Now note that a path can also be constructed between $W_*$ and $W_{**}$, as their convex combination also satisfies (A1)-(A2) and (A4)-(A6). For example, note that $\hat W(t) = t W_* + (1-t) W_{**}$ is such that $(\hat W(t)^2)_{11} = t^2 + (1-t)^2$ and $(\hat W(t)^2)_{NN} = 0$, so $(\hat W(t)^2)_{11} - (\hat W(t)^2)_{NN} > 0$ for any $t \in (0,1)$ and both $\hat W(0)$ and $\hat W(1)$ satisfy (A5).  Hence, we can construct a path $W(t)$ between $W$ and $W^*$ through $W_*$ and $W_{**}$.

\bigskip

Furthermore, $\rho(t)=t\rho^*+(1-t){\rho}$, $\beta(t)=(t \rho ^* \beta^*+(1-t){\rho \beta})/(t\rho^*+(1-t){\rho})$, $\gamma(t)=t\gamma^*+(1-t){\gamma}$ are such that
\begin{eqnarray*}
f(t)\;\;\equiv\;\;\rho(t)\beta(t)+\gamma(t)\;\;=\;\;t(\rho^*\beta^*+\gamma^*)+(1-t)({\rho}{\beta}+{\gamma}) > 0,
\end{eqnarray*}
since $\theta^*$ and ${\theta}\in\Theta_{+}$. (Note also that $| \rho(t) | < 1$, so Assumption (A2) is satisfied.)  These facts taken together imply that
\begin{eqnarray*}
\theta(t)\;\;\equiv\;\;\left(W(t)_{12},\dots,W(t)_{N,N-1},\rho(t),\gamma(t),\beta(t)\right)\;\;\in\;\;\Theta_{+}.
\end{eqnarray*}
That is, $\Theta_{+}$ is path-connected and therefore connected.  Since $\Pi(\cdot)$ is continuous on $\Theta_+$, $\Pi(\Theta_+)$ is connected.
\end{proof}

\subsubsection*{Theorem 2}

\begin{proof}
The proof uses  Corollary 1.4 in \citet[p. 46]{AmbrosettiProdi1995},\footnote{Related results can be found in \citet{AmbrosettiProdi1972} and \citet{Marcoetal2014}.} which we reproduce here with our notation for convenience: \emph{Suppose the function $\Pi(\cdot)$ is continuous, proper, and locally invertible with a connected image.  Then, the cardinality of $\Pi^{-1}(\overline \Pi)$ is constant for any $\overline \Pi$ in the image of $\Pi(\cdot)$.}  

The mapping $\Pi(\theta)$ is continuous and proper (by Corollary 1),  and non-singular Jacobian at any point (as per the proof for Theorem 1),
 which guarantees local invertibility. Following Corollary 1.4 in \citet[p. 46]{AmbrosettiProdi1995} reproduced above, we obtain that the cardinality of the pre-image of $\Pi(\theta)$ is finite and constant. Take $\theta\in\Theta_{+}$ such that $\gamma=0$, $(W)_{1,2}=(W)_{2,1}=1$ and $(W)_{i,j}=0$ otherwise, with $\rho\neq0$ and $\beta\neq0$. By Lemma 1, that cardinality is one. 

\end{proof}

\subsubsection*{Corollary 3}

\begin{proof}
Since $\rho \in (0,1)$ and $W_{ij} \ge 0$, $\sum_{k=1}^\infty \rho^{k-1}W^k$ is a non-negative matrix. By (5), the off-diagonal elements of $\Pi(\theta)$ are equal to the off-diagonal elements of $ \left(\rho \beta +\gamma \right)\sum_{k=1}^\infty \rho^{k-1}W^k$, and the sign of those elements identifies the sign of $\rho \beta +\gamma$.  By Theorem 2, the model is identified. 

\end{proof}

\subsubsection*{Corollary 4}

\begin{proof}
Since $W_0$ is non-negative and irreducible, there is a real eigenvalue equal to the spectral radius of $W_0$ corresponding to the unique eigenvector whose entries can be chosen to be strictly positive (i.e., all the entries share the same sign).  A generic eigenvalue of $W_0$,  $\lambda_{0}$, corresponds to an eigenvalue of $\Pi_0$ according to:
$$\lambda_{\Pi_0} = \beta_0 + (\rho_0 \beta_0 + \gamma_0) \frac{\lambda_{0}}{1-\rho_0 \lambda_{0}}$$
If $\lambda_{0} = a_0 + b_0 i$ where $a_0, b_0 \in \mathbb{R}$ and $i = \sqrt{-1}$, then
$$\lambda_{\Pi_0} = \beta_0 + (\rho_0 \beta_0 + \gamma_0) \frac{a_0(1-\rho_0 a_0)-\rho_0 b_0^2}{(1-\rho_0 a_0)^2 + \rho_0^2 b_0^2} + (\rho_0 \beta_0 + \gamma_0) \frac{b_0}{(1-\rho_0 a_0)^2 + \rho_0^2 b_0^2} i.$$
If the eigenvalue $\lambda_0$ is real, $b_0=0$ and the corresponding eigenvalue $\lambda_{\Pi_0}$ is also real.  Differentiating $Re(\lambda_{\Pi_0})$, the real part of $\lambda_{\Pi_0}$, with respect to $Re(\lambda_0) = a_0$, we get
\begin{equation} \label{eq:derivative}
\frac{\partial Re(\lambda_{\Pi_0})}{\partial a_0} = \frac{(1-\rho_0 a_0)^2-\rho_0^2 b_0^2}{[(1-\rho_0 a_0)^2 + \rho_0^2 b_0^2 ]^2} \times (\rho_0 \beta_0 + \gamma_0).
\end{equation}
If the eigenvalue $\lambda_0$ is real, expression (\ref{eq:derivative}) becomes:
$$\frac{\partial Re(\lambda_{\Pi_0})}{\partial a_0}=\frac{\partial \lambda_{\Pi_0}}{\partial a_0} = \frac{1}{(1-\rho_0 a_0)^2} \times (\rho_0 \beta_0 + \gamma_0).$$
The fraction multiplying $\rho_0 \beta_0 + \gamma_0$ is positive.  If $\rho_0 \beta_0 + \gamma_0 < 0$, the real eigenvalues of $\Pi_0$ are decreasing on the real eigenvalues of $W_0$.  Consequently, the eigenvector corresponding to the largest (real) eigenvalue of $W_0$ will be associated with the smallest real eigenvalue of $\Pi_0$.  If, on the other hand, $\rho_0 \beta_0 + \gamma_0 > 0$, the eigenvector corresponding to the largest real eigenvalue of $W_0$ will correspond to the largest real eigenvalue of $\Pi_0$.  Since that eigenvector is the unique eigenvector that can be chosen to have strictly positive entries, the sign of $\rho_0 \beta_0 + \gamma_0$ is identified by the $\lambda_{\Pi_0}$ eigenvalue it is associated with and whether it is the largest or smallest real eigenvalue. By Theorem 2,
the model is identified.

If there is only one real eigenvalue, note that the denominator in the fraction in (\ref{eq:derivative}) is positive.  The minimum value of the numerator subject to $|\lambda_0|^2 = a_0^2 + b_0^2 \le 1$ is given by
$$\min_{a_0, b_0} (1-\rho_0 a_0)^2-\rho^2 b_0^2 \textrm{ s.t. } a_0^2 + b_0^2 \le 1.$$
The Lagrangean for this minimization problem is given by
\begin{eqnarray*}
\mathcal{L}\left(a_0,b_0;\mu\right) & = & (1-\rho_0 a_0)^2-\rho^2 b_0^2 + \mu(a_0^2 + b_0^2 - 1)
\end{eqnarray*}
where $\mu$ is the Lagrange multiplier associated with the constraint $a_0^2 + b_0^2 \le 1$.  The Kuhn-Tucker necessary conditions for the solution $(a_0^*,b_0^*,\mu^*)$ of this problem are given by
\begin{eqnarray*}
 &( \partial a_0:) & \rho_0 (1-\rho_0 a_0^*) - \mu^* a_0^* = 0 \\
 & (\partial b_0:) & (\rho_0^2 - \mu^*) b^*_0 = 0 \\
 & & \mu^* (  a_0^{*2} + b_0^{*2} - 1 ) = 0 \\
 & & a_0^{*2} + b_0^{*2} \le 1 \textrm{ and } \mu^* \ge 0.
 \end{eqnarray*} 
Let $\rho_0 \ne 0$.  (Otherwise, the objective function above is equal to one irrespective of $a_0$ or $b_0$, and the partial derivative is $\rho_0 \beta_0 + \gamma_0$).  If $\mu^* = 0$, $\partial b_0$ implies that $b_0^*=0$.  Then, $\partial a_0$ will have $a^*_0=\rho_0^{-1}$, which violates $a_0^{*2} + b_0^{*2} \le 1$.   

Hence, a solution should have $\mu^*>0$.  In this case, there are two possibilities: $b^*_0=0$ or $b^*_0 \ne 0$.  If $b^*_0 \ne 0$, condition $\partial b_0$ implies that $\mu^* = \rho_0^2$, and $\partial a_0$ then gives $a_0^* = (2 \rho_0)^{-1}$.  Because the constraint is binding, $b_0^{*2} = 1 - (4 \rho_0^2)^{-1}$.  In this case, $a_0^{*2} \le 1$ and $b_0^{*2} \ge 0$ requires that $|\rho_0| \ge 1/2$.  The value of the minimised objective function in this case is $1/2 - \rho_0^2$.  This is positive if $| \rho_0 | < \sqrt{2}/2$.

The other possibility is to have $b_0 = 0$.  Because the constraint is binding, $a_0=1$ and the objective function takes the value $(1-\rho_0)^2 > 0$.  Since $(1-\rho_0)^2 - 1/2 + \rho_0^2 = 2 \rho_0^2 - 2 \rho_0 + 1/2 \ge 0$, this solution is dominated by the previous one when $|\rho_0| \ge 1/2$.  

Consequently, the fraction multiplying $\rho_0 \beta_0 + \gamma_0$ is non-negative, and it can be ascertained that

$$ \textrm{sgn}\left[ \frac{\partial Re(\lambda_{\Pi_0})}{\partial a_0} \right] = \textrm{sgn}[ \rho_0 \beta_0 + \gamma_0 ] $$

\noindent as long as $| \rho_0 | < \sqrt{2}/2$.

If $\rho_0 \beta_0 + \gamma_0 < 0$, the real part of the eigenvalues of $\Pi_0$ is decreasing on the real part of the eigenvalues of $W_0$.  Since the dominant eigenvalue for $W_0$ will be real and thus the one with the largest real part, we only need to focus on the real part of the eigenvalues. 
 Consequently, the eigenvector corresponding to the eigenvalue of $W_0$ with the largest real part will correspond to the eigenvalue of $\Pi_0$ with the smallest real part.  If, on the other hand, $\rho_0 \beta_0 + \gamma_0 > 0$, the eigenvector corresponding to the eigenvalue of $W_0$ with the largest real part will correspond to the eigenvalue of $\Pi_0$ with the largest real part.  Since that eigenvector is the unique eigenvector that can be chosen to have strictly positive entries, the sign of $\rho_0 \beta_0 + \gamma_0$ is identified by the $\lambda_{\Pi_0}$ eigenvalue it is associated with. 

By Theorem 2, the model is identified. 
\end{proof}

\subsubsection*{Proposition 1}

\begin{proof}
From equation 
 (6), we observed that $\Pi_0v_j=\lambda_{\Pi_0,j}v_j$, where $v_j$ is an eigenvector of both $W_0$ and $\Pi_0$ with the corresponding eigenvalue $\lambda_{\Pi_0,j}=\frac{\beta_0+\gamma_0\lambda_{0,j}}{1-\rho_0\lambda_{0,j}}$. Defining $c$ as the row-sum of $\Pi_0$, we also have that 
\begin{eqnarray*}
\tilde{\Pi}_0(I-H)v_j&=&(I-H)\Pi_0 (I-H)v_j\;\;=\;\;(I-H)\Pi_0v_j-(I-H)\Pi_0 Hv_j\\
&=& \lambda_{\Pi_0,j}(I-H)v_j-(I-H) c H v_j \;\;=\;\; \lambda_{\Pi_0,j}(I-H)v_j-(H-H^2)c v_j\\
&=&\lambda_{\Pi_0,j}(I-H)v_j-(H-H)c v_j\;\;=\;\;\lambda_{\Pi_0,j}(I-H)v_j,
\end{eqnarray*}
where the third equality is obtained from $\Pi_0 H = c H$, where $c=\frac{\beta_0+\gamma_0}{1-\rho_0}$, and the fifth equality holds since $H$ is idempotent. So, $\tilde{\Pi}_0$ and $\Pi_0$ have common eigenvalues, with the corresponding eigenvector $\tilde{v}_j=v_j-\bar{v}_j\iota$ for $\tilde{\Pi}_0$, where $\bar{v}_j=\frac{1}{N}\iota'v_j$, $j=1,\dots,N$. 

If $\Pi_0$ is diagonalizable, since $\lambda_{\Pi_0,j}$ and $\tilde{v}_j$ are observed from $\tilde{\Pi}_0$, identification of $\Pi_0$ is equivalent to identification of $\bar{v}_j$.  In order to handle non-diagonalizable matrices, we relate the Jordan canonical forms for $\Pi_0$ and $\tilde \Pi_0$.  (If $\Pi_0$ is diagonalizable, the Jordan canonical form will coincide with its eigendecomposition.) In this case, $\Pi_0=MJM^{-1}$,where $J$ is the Jordan form and $M$ is a matrix comprised of ordinary and generalized eigenvectors.  The generalized eigenvectors can be obtained recursively from ordinary eigenvectors.\footnote{``The search for the Jordan form of $A$ becomes a search for these strings of vectors, each one headed by an eigenvector. For each $i$, either $Ax_i = \lambda_i x_i$ or $Ax_i = \lambda_i x_i + x_{i-1}$.  The vectors $x_i$ go into the columns of $M$, and each string produces a single block in $J$.'' (\citealp{strang2006})} We have established above that if $v_j$ is an ordinary eigenvector for $\Pi_0$, $\tilde{v}_j=v_j-\bar{v}_j\iota$ is an ordinary eigevector for $\tilde \Pi_0$, corresponding to the same eigenvalue. If this eigenvalue has multiplicity larger than one, let $x_j$ be the initial generalized eigenvector for $\Pi_0$, obtained as $\Pi_0 x_j = \lambda_{\Pi_0,j} x_j + v_j$. 
 Then, notice that
 \begin{eqnarray*}
\tilde{\Pi}_0(I-H)x_j&=&(I-H)\Pi_0 (I-H)x_j\;\;=\;\;(I-H)\Pi_0x_j-(I-H)\Pi_0 Hx_j\\
&=& \lambda_{\Pi_0,j}(I-H)x_j + (I-H)v_j-(I-H) c H x_j \;\;=\;\; \lambda_{\Pi_0,j}(I-H)x_j + (I-H)v_j,
\end{eqnarray*}
where we reproduce similar steps as above.  This implies that $\tilde x_j = x_j - \overline x_j$ is a generalized eigenvector for $\tilde \Pi_0$ obtained from the ordinary eigevector $\tilde v_j$. If additional generalized eigevectors are needed for the Jordan block, they can be obtained recursively from $x_j$ and $\tilde x_j$.  Notice also that both $\Pi_0$ and $\tilde \Pi_0$ will have the same Jordan form $J$ since their eigenvalues are the same.  Hence, to establish identification, it suffices to establish identification of $\overline v_j$ and $\overline x_j$.

To establish identification of $\bar{v}_j$, note that $\Pi_0(\tilde{v}_j+\bar{v}_j\iota)=\lambda_{\Pi_0,j}(\tilde{v}_j+\bar{v}_j\iota)$ since $v_j$ is an eigenvector of $\Pi_0$. Consider an alternative constant $\bar{v}_j^*\neq\bar{v}_j$ that satisfies the previous equation. Then,
\begin{eqnarray*}
\Pi_0\iota(\bar{v}_j-\bar{v}_j^*)&=&\lambda_{\Pi_0,j}\iota(\bar{v}_j-\bar{v}_j^*).	
\end{eqnarray*}
Since $\Pi_0\iota=c$, $v_j$ must satisfy $(c-\lambda_{\Pi_0,j})(\bar{v}_j-\bar{v}_j^*)=0$. For $j=2,\dots,N$, $|\lambda_{0,j}|<1$ which implies that $c \neq \lambda_{\Pi_0,j}$. So, $\bar{v}_j=\bar{v}_j^*$ and therefore is identified. For $j=1$, $\lambda_{\Pi_0,1}=c$ with eigenvector $v_1=\iota$ if $W_0$ is non-negative and irreducible, since it corresponds to $\lambda_{0,1}=1$, which is a simple eigenvalue with both algebraic and geometric multiplicity equal to one by the Perron-Frobenius theorem. 

Now consider $x_j$ and assume, without loss, that it is a generalized eigenvector obtained from $v_j$ as above.  In this case, $\Pi_0(\tilde{x_j}+\bar{x}_j\iota)=\lambda_{\Pi_0,j}(\tilde{x}_j+\bar{x}_j\iota) + (\tilde{v}_j+\bar{v}_j\iota)$. Since we have established the identification $\overline v_j$ above, consider an alternative constant $\bar{x}_j^*\neq\bar{x}_j$ that satisfies the previous equation.  Then, similar arguments as those used above allow one to obtain identification for $\overline x_j$ and analogously for possible successive generalized eigenvectors.  This then allows us to reconstruct the Jordan decomposition for $\Pi_0$ from the Jordan decomposition of $\tilde \Pi_0$.
\end{proof}

\subsubsection*{Proposition 2} 

\begin{proof}
Under row-sum normalization and $|\rho_0|<1$, $(I-\rho_0W_0)^{-1}\iota=\iota+\rho_0W_0\iota+\rho_0^2W_0^2\iota+\cdots=\iota+\rho_0\iota+\rho_0^2\iota+\cdots=\iota\frac{1}{1-\rho_0}$, so $\Pi_{01}\equiv(I-\rho_0W)^{-1}$ has constant row-sums. If row-sum normalization fails, $\Pi_{01}$ may not have constant row-sums. Define $h_{ij}$ as the $(ij)$-th element of $\tilde{H}$. The first row of the system $(I-\tilde{H})(I-\rho_0W)^{-1}\iota=(I-\tilde{H})r_{W_{0}}=0$ is $h^*_{11}r_{W_{0},1}-h_{12}r_{W_{0},2}-\cdots-h_{1N}r_{W_{0},N}=0$ where $h^*_{11}=1-h_{11}$ and $r_{{W_0},l}$ is the $l$-th element of $r_{W_{0}}$. If there are $N$ possible $W_0$, $W_0^{(1)},\dots,W_0^{(n)}$, such that $[r_{W_{0}^{(1)}}\;\cdots\;r_{W_{0}^{(N)}}]$ has rank $N$, then $h^*_{11}=h_{12}=\cdots=h_{1N}=0$. Since the same reasoning applies to all rows, $\tilde{H}$ is the trivial transformation $\tilde{H}=I$.
\end{proof}

\section{Extensions}

\subsubsection*{Extension: Multivariate Covariates\label{sec:additionaldiscussion}}

Allowing for multivariate $x_{t}$ of dimension $n\times k$, the
reduced-form model (4) is,
\begin{equation*}
y_{t}=\sum_{k=1}^{K}\Pi _{0,k}x_{k,t}+\nu _{t},
\end{equation*}%
where $\Pi _{0,k}=\left( I-\rho _{0}W_{0}\right) ^{-1}\left(
\beta_{0,k}I+\gamma _{0,k}W_{0}\right) $, $x_{k,t}$ refers to the $k$-th
column of $x_{t}$, and $\beta _{0,k}$ and $\gamma _{0,k}$ select the $k$-th
element of $K$-dimensional $\beta _{0}$ and $\gamma _{0}$, respectively. The
previous identification results then apply element by element to each $\Pi _{0,k}$%
, $k=1,\dots ,K$. In fact, we only then need to maintain $W_{x}=\gamma
_{0}W_{0} $ and $\gamma_0\neq0$ for one covariate. It is therefore possible to allow the
structure of endogenous and exogenous social effects to differ for $K-1$ of
the covariates. With $K$ covariates, equation (3) is,
\begin{equation*}
y_{t}=\rho _{0}W_{0}y_{t}+\sum_{k=1}^{K}\beta
_{0,k}x_{k,t}+\sum_{k=1}^{K}\gamma _{0,k}W_{0,k}x_{k,t}+\epsilon _{t}.
\end{equation*}%
Let $W_{0,k}=W_{0}$ be the case for $k=1$. Then, having identified $\rho
_{0} $ and $W_{0}$ from $\Pi _{0,1}$,%
\begin{equation*}
(I-\rho _{0}W_{0})\Pi _{0,k}=\beta _{0,k}I+\gamma _{0,k}W_{0,k},
\end{equation*}%
for $k=2,\dots ,K$. The parameter $\beta _{0,k}$ then corresponds to the
diagonal elements of $(I-\rho _{0}W_{0})\Pi _{0,k}$ and the off-diagonal
entries correspond to the off-diagonal elements of $\gamma _{0,k}W_{0,k}$.
If Assumption (A4) holds for every $k=1,\dots ,K$, we can identify $\gamma
_{0,k}$ and thus $W_{0,k}$ for every $k=1,\dots ,K$.\footnote{%
If $x$ is a scalar, $\Pi = (I-W_{y})^{-1}(\beta I + W_x)$, where we absorb $%
\rho$ and $\gamma$ into the neighborhood matrices for simplicity. Supposing
both $W_y$ and $W_x$ have zero diagonals (Assumption (A1)) one has $%
2N(N-1)+1 $ structural parameters against $N^2$ elements in $\Pi$. Since $%
2N(N-1)+1 > N^2$ for $N>1$ one would not be able to identify the parameters
of interest without further information. \citet{Blumeetal2015} also study
the case in which the social structure mediating endogenous and exogenous
social effects might differ. When $W_{x}$ is known and there is partial
knowledge of the endogenous social interaction matrix $W_{0}$, they show
that the parameters of the model can be identified (their Theorem 6).
Analogously, when there are enough unconnected nodes in each of the social
interaction matrices represented by $W_{x}$ and $W_{0}$, and the identity of
those nodes is known, identification is also (generically) possible (their
Theorem 7).}

\subsubsection*{Extension: Heterogeneous $\protect\beta_{0}$}

While many applications assume $\beta _{0}$ to be homogeneous across
individual units, we here consider possible avenues allowing for
heterogeneous coefficients. In a slight abuse of notation, consider for this
subsection $\beta _{0}$ in equation system (3) to be 
$\text{diag}(\beta _{01},\dots ,\beta _{0N})_{N\times N}$. Instead of a homogeneous
scalar, $\beta _{0}$ is a diagonal matrix with the individual-specific
coefficients $\beta _{01},\dots ,\beta _{0N}$ along its diagonal. When $\rho
_{0}=0$ as in \citet{Manresa2016}, $\Pi _{0}=\beta _{0}+\gamma _{0}W_{0}$.
In this case, under Assumption (A1), $\beta _{0}$ is identified from the
diagonal elements in $\Pi _{0}$ and $\gamma _{0}W_{0}$ is identified from
its off-diagonal elements.

With multiple covariates, as long as the coefficients are homogeneous for one of
the covariates, one can also identify heterogeneous coefficients on the
remaining covariates, as done in the previous subsection. For example, let
there be $K$ covariates and $\beta _{0,k}=\text{diag}(\beta _{01,k},\dots ,\beta
_{0N,k})$ for $k=1,\dots ,N$. Suppose $\beta _{01,k}=\dots =\beta _{0N,k}$
for one of these covariates, and let $k=1$ without loss of generality. Having
identified $\rho _{0}$ and $W_{0}$ from $\Pi _{0,1}$,%
\begin{equation*}
(I-\rho _{0}W_{0})\Pi _{0,k}=\beta _{0,k}+\gamma _{0,k}W_{0,k},
\end{equation*}%
for $k=2,\dots ,K$. Then, under Assumption (A1), $\beta _{0,k}$, is
identified from the diagonal elements in $(I-\rho _{0}W_{0})\Pi _{0,k}$ and $%
\gamma _{0,k}W_{0,k}$ is identified from its off-diagonal elements for $%
k=2,\dots ,K-1$.

Alternatively, when $\gamma _{0}=0$, one can apply traditional simultaneous
equation methods to attain identification. For example, let $B\equiv \lbrack
(I-\rho _{0}W_{0})^{\prime },~-(\beta _{0}+\gamma _{0}W)^{\prime
}]_{2N\times N}^{\prime }$ and $R_{(N-1)\times 2N}=[0_{(N-1)\times (N+1)} ~
I_{N-1}]$. The restriction that $\gamma _{0}=0$ in the first equation in
equation system (3) can then be expressed as $RB_{\cdot 
,1}=0_{(N-1)\times 1}$, where $B_{\cdot ,1}$ is the first column in $B$. The
rank condition for the identification of the first equation is then given by,%
\begin{equation*}
RB=\left[ 
\begin{array}{cccc}
0 & -\beta _{0,2} & \dots & 0 \\ 
& \dots &  &  \\ 
0 & 0 & \dots & -\beta _{0,N}%
\end{array}%
\right]
\end{equation*}%
having rank equal to $N-1$ (see Theorem 9.2 in \citealp{Wooldridge2002}).
This will be the case if $\beta _{0,2},\dots ,\beta _{0,N}\neq 0$.
Intuitively, this guarantees that individual specific covariates are valid
instrumental variables for their outcomes. Hence, if $\beta _{0,1},\dots
,\beta _{0,N}$ are each different from zero, identification from $\Pi _{0}$
is obtained.

More generally, because there are $N^{2}$ equations corresponding to the
entries in $\Pi _{0}$ and, allowing for heterogeneity in $\beta _{0}$ and
imposing assumptions (A1)-(A6), there are $N^{2}+1$ parameters, further
restrictions (like row-sum normalization) would be necessary. We conjecture
that adequate restrictions would deliver positive identification results,
but we focus on the more conventional setting with homogeneous $\beta _{0}$.

\section{Estimation}

\subsection{Sparsity of $W_0$ and $\Pi_0$}

\label{appsubec:sparsityWPi}

Define $\tilde{M}$ as the number of nonzero elements of $\Pi_0$. We say
that $\Pi_0$ is sparse if $\tilde{M}\ll NT$. Denote the number of connected
pairs in $W_0$ via paths of any length as $\tilde{m}_c$. We equivalently say
that $W_0$ is ``sparsely connected'' if $\tilde{m}_c\ll NT$. We show that
the sparsity of $\Pi_0$ is related to the sparse connectedness of $W_0$.

\setcounter{prop}{2}
\begin{prop}
$\Pi_0$ is sparse if, and only if, the number of unconnected pairs $W_0$ is
large.
\end{prop}

\begin{proof}
For $|\rho_0|<1$, we have that
\begin{eqnarray*}
\Pi_0 &=& \beta_0I + \left(\rho_0\beta_0+\gamma_0\right)\sum_{k=1}^\infty \rho_0^{k-1}W_0^k.
\end{eqnarray*}
Given that $\rho_0\beta_0+\gamma_0\neq0$, it follows directly that $[\Pi_0]_{ij}=0$ if, and only if, there are no paths between $i$ and $j$ in $W_0$. Therefore, the sparsity of $\Pi_0$ translates into a large number of $\left(i,j\right)$ unconnected pairs in $W_0$. 
\end{proof}

On the one hand, sparsity does not imply sparse connectedness. A circular
graph is clearly sparse, but all nodes connect with all other nodes through
a path of length at most $\frac{N}{2}$. On the other hand, sparse
connectedness implies sparsity and therefore is a stronger requirement. To
see this, take any arbitrary network $G$ with $\tilde{m}\left( G\right) $
non-zero elements and $\tilde{m}_{c}\left( G\right) $ connected pairs. Now
consider the operation of \textquotedblleft completing\textquotedblright\ $G$%
: for every connected $\left( i,j\right) $ pair, add a direct link between $%
\left( i,j\right) $ if nonexistent in $G$ and denote the resulting matrix
as $\mathcal{C}\left( G\right) $. It is clear that $\tilde{m}\left( G\right)
\leq \tilde{m}\left( \mathcal{C}\left( G\right) \right) $. Yet, $\tilde{m}%
\left( \mathcal{C}\left( G\right) \right) =\tilde{m}_{c}\left( G\right) $.

\label{appsubsec:OLS}

\subsection{Adaptive Elastic Net\label{appsubsec:ENet}}

In this section, we detail the algorithm, the computational steps, and the
property of the estimator as derived in \citet{CanerZhang2014}. For any
given choice of $p=(p_{1},p_{1}^{\ast },p_{2})$, the algorithm is composed
of three main steps pertaining to the estimation of the Elastic Net (step 5
below), the Adaptive Elastic Net version (step 6), and the Unpenalised
Post-GMM estimator (step 7). Other steps deal with the selection of the
initial conditions and other details of the implementation. The runtime for a given set of
penalization parameters (that is, steps 1-7 below) is expected to be around
2-3 mins for $N=30$ and 15-20 mins for $N=70$. The increase in the
computational time is due to the fact that the number of parameters to
estimate grows at an $N^{2}$ rate.\footnote{%
Benchmarked on a 2020 Macbook Pro with an M1 processor and 16gb RAM.}

\begin{itemize}
\item[1.] For any choice of penalization parameters $p=(p_1,p_1^*,p_2)$, run
steps 2-7 below.

\item[2.] Data is standardised (subtracting the mean and dividing by the
standard error), and individual and time effects are removed.

\begin{itemize}
\item[a.] For the dynamic versions, use time weights $\omega_{t}$, which are
chosen by the practitioner and applied after standardization.
\end{itemize}

\item[3.] Initial conditions are set at $(\rho,\beta,\gamma)=(0.5,\hat{\beta}%
_{ols},0)$, where $\hat{\beta}_{ols}$ is the OLS regression of $y_{it}$ on $%
x_{it}$ after demeaning.

\item[4.] The derivatives of the objective function with respect to all $%
w_{ij}$ are computed, and set to $w_{ij}=0$ for the derivatives that are
smaller than a small number $\eta$. This number can be chosen to be equal to 
$p_1$ following \citet{CanerZhang2014}. To be conservative, we choose $%
\eta=\min(p_1,0.01)$. Those $w_{ij}$ are set at zero for all the algorithm
that follow.

\item[5.] \emph{Elastic Net.} We devise an algorithm in the following manner.

\begin{itemize}
\item[a.] For any given $(\rho,\beta,\gamma)$, we write, 
\begin{eqnarray*}
\underbrace{y_t-x_t\beta}_{\equiv \tilde{y}_{it}}&=&W\underbrace{(\rho y_t +
x_t \gamma)}_{\equiv \tilde{x}_{t}}+\epsilon_{t},
\end{eqnarray*}
which is endogenous due to the mechanical dependence of $\tilde{x}_{t}$ on $%
y_{t}$, so it is instrumented with $x_t$. This expression is implemented
through a fast Least-Angle Regression algorithm (LARS) from %
\citet{Caneretal2018} (Section 4.2) that extends the original LARS algorithm
to the GMM case. We obtain $\hat{W}$ for any given $(\rho,\beta,\gamma)$. At
this stage, only impose that $|w_{ij}|\leq 1$ and, in particular, do not
impose row-sum normalization.

\item[b.] Use the L-BFGS-B algorithm to minimize the objective function and
obtain $(\tilde{\rho},\tilde{\beta},\tilde{\gamma})$, and then compute $%
\tilde{W}$ from step 5a.
\end{itemize}

\item[6.] \emph{Adaptive Elastic Net.} Use $\tilde{W}$ to compute the
adaptive weights. Following \citet{CanerZhang2014}, we select the adaptive
weight penality as $w_{ij}^{-\kappa}$ if $w_{ij}>0.05$ where $\kappa=2.5$.
If $w_{ij}<0.05$, we set $0.05^{-\kappa}$. This ensures that the second-stage
estimates can be nonzero even if the first-stage estimates were zero or small.

\begin{itemize}
\item[a.] We then run a version of the algorithm described in step 5a, with
the difference that we impose row-sum normalization. Due to this, we can no
longer use a restriction-free LARS algorithm. We instead use CVXR, a convex
optimization routine package in \texttt{R}.

\item[b.] Similar to step 5b, we use the L-BFGS-B package to solve for $(\hat{%
\rho},\hat{\beta},\hat{\gamma})$ and compute the associated $\widehat{W}$.
\end{itemize}

\item[7.] \emph{Unpenalised Post-GMM.} Re-estimate $(\rho,\beta,\gamma)$ and 
$W$ on the support estimated in the previous step. In other words, the zero
elements of $\widehat{W}$ are set to zero, i.e., the GMM objective function
is estimated under the restriction that $\{W_{ij}=0:\widehat{W}_{ij}=0\}$
and setting the penalization to zero. Asymptotic standard errors are also
computed at this step.

\item[8.] Re-compute steps 1-6 on a grid for $(p_1,p_1^*,p_2)$ and choose
the penalization parameter by BIC.
\end{itemize}

\noindent We implement the following modifications and adjustments of the
algorithm for the empirical analysis presented in Section 4. Steps 4$^\prime$ and 8$^\prime$
provide more stability to the period-by-period estimates by avoiding
small-sample biases in the estimation of the derivative and of the
information criteria.

\begin{itemize}
\item[2a$^\prime$.] We use uniform weights $\omega_t=1$ and a Gaussian kernel with a
center varying period by period from $t=1963$ to $2015$, i.e., ranging from
the beginning to the end of the sample. The variance of the kernel is set
such that 75\% of the weight is given to the first half of the data
(pre-1988), when the center of the kernel is at the beginning of the sample
at $t=1963$.

\item[4$^\prime$.] We compute the derivative under uniform weights, store it, and
input the derivative for each dynamic-network version. This is computationally
efficient (given that the numerical approximation of the derivative is
intensive) and ensures that the derivative is computed from the full sample.

\item[8$^\prime$.] The grid for the dynamic-network versions is set around the
penalization parameters chosen under the uniform kernel. More specifically, we
run the procedure for the penalization parameters for one point in all
directions in the grid and select the dynamic-network penalization
parameters by BIC.
\end{itemize}

\subsection{OLS\label{appsubsec:OLSb}}

For the purpose of estimation, it is convenient to write the model in the
stacked form. Let $x=\left[ x_{1},\dots ,x_{T}\right] ^{\prime }$ be the $%
T\times N$ matrix of explanatory variables, $y_{i}=\left[ y_{i1},\dots
,y_{iT}\right] ^{\prime }$ be the $T\times 1$ vector of response variables
for individual $i$, and $\pi _{i}^{0}=\left[ \pi _{i1}^{0},\dots ,\pi
_{iN}^{0}\right] ^{\prime }$, where $\pi _{ij}^{0}$ is a short notation for
the $\left( i,j\right) $-th element of $\Pi _{0}$. The concise model is then%
\begin{equation}
y_{i}=x\pi _{i}^{0}+v_{i}  \label{eq:shortModel}
\end{equation}%
for each $i=1,\dots ,N$, where also $v_{i}=\left[ v_{i1},\dots ,v_{iT}\right]
^{\prime }$. Model \eqref{eq:shortModel} can then be estimated
equation by equation. Denote $\pi ^{0}=[{\pi _{1}^{0}}^{\prime },\dots ,{\pi
_{N}^{0}}^{\prime }]^{\prime }$. Stacking the full set of $N$ equations,%
\begin{equation}
y=X\pi ^{0}+v  \label{eq:shortModelFull}
\end{equation}%
where $y=\left[ y_{1},\dots ,y_{N}\right] $, $X=I_{N}\otimes x$, $\pi ^{0}=%
\text{vec}\left( \Pi _{0}^{\prime }\right) $, and $v=\left[ v_{1},\dots
,v_{N}\right] $. If the number of individuals in the network $N$ is fixed
and much smaller than the number of data points available, $N^{2}\ll NT$, equation %
\eqref{eq:shortModelFull} can be estimated via ordinary least squares (OLS).
Under suitable regularity conditions, the OLS estimator $\hat{\pi}=\left(
X^{\prime }X\right) ^{-1}X^{\prime }y$ is asymptotically distributed,%
\begin{equation*}
\sqrt{NT}(\hat{\pi}-\pi ^{0})\overset{d}{\longrightarrow }\mathcal{N}\left(
0,Q^{-1}\Sigma Q^{-1}\right)
\end{equation*}%
where $Q_{T}\equiv \frac{1}{NT}{X^{\prime }X}$, $Q\equiv p\lim_{T\rightarrow
\infty }Q_{T}$, $\Sigma _{T}\equiv \frac{1}{NT}X^{\prime }vv^{\prime }X$, and 
$\Sigma \equiv p\lim_{T\rightarrow \infty }\Sigma _{T}$. The proof is
standard and omitted here. As noted above, in typical applications, it is
customary to row-sum normalize matrix $W$. If no individual is isolated, one
obtains that, by equation (5),
\begin{eqnarray}
\Pi _{0}\iota _{N} &=&\beta _{0}\iota +(\rho _{0}\beta _{0}+\gamma
_{0})\sum_{k=1}^{\infty }\rho _{0}^{k-1}W_{0}^{k}\iota  \notag \\
&=&\frac{\beta _{0}+\gamma _{0}}{1-\rho _{0}}\iota  \label{eq:PiRowSum}
\end{eqnarray}%
where $\iota _{N}$ is the $N$-length vector of ones. The last equality
follows from the observation that, under row-normalization of $W_{0}$, $%
W^{k}\iota =W\iota =\iota $, $k>0$. Equation \eqref{eq:PiRowSum} implies
that $\Pi _{0}$ has constant row-sums, which implies that row-sum
normalization is, in principle, testable. This suggests a simple Wald
statistic applied to the estimates of $\pi ^{0}$. Under the null hypothesis,%
\begin{equation*}
\sqrt{NT}R\hat{\pi}\overset{d}{\longrightarrow }\mathcal{N}\left(
0,RQ^{-1}\Sigma Q^{-1}R^{\prime }\right)
\end{equation*}%
where $R=\left[ I_{N-1}\otimes \iota _{N}^{\prime };-\iota _{N-1}\otimes
\iota _{N}^{\prime }\right] $. The Wald statistic is $W=NT\left( R\hat{\pi}%
\right) ^{\prime }\left( Q^{-1}\Sigma Q^{-1}\right) ^{-1}\left( R\hat{\pi}%
\right) \sim \chi _{N-1}^{2}$, which is a convenient expression for testing
row-sum normalization of $W_{0}$. We also note that the asymptotic
distribution of $\hat{\theta}$ can be immediately obtained by the Delta
Method,%
\begin{equation*}
\sqrt{T}(\hat{\theta}-\theta _{0})\overset{d}{\longrightarrow }\mathcal{N}%
\left( 0,\nabla _{\theta }^{\prime -1}Q^{-1}\Sigma Q^{-1}\nabla _{\theta
}\right)
\end{equation*}%
where $\nabla _{\theta }$ is the gradient of $\hat{\theta}$ with respect to $%
\hat{\pi}$. 
\textcolor{black}{We note that the derivation of  the Wald
statistic for testing the row-sum normalization and the asymptotic
distribution of $\hat{\theta}$ does not depend on the OLS implementation,
and can be easily adjusted for any estimator for which the asymptotic
distribution is known.}

\section{Simulations\label{sec:simulations copy(1)}}

\subsection{Set-up}

The simulations are based on two stylized random network structures and two
real-world networks. These networks vary in their size, complexity, and
aggregate and node-level features. All four networks are also sparse. Networks (i) and (ii) are stylized, while (iii) and (iv) are based on real data:

\begin{itemize}
\item[(\emph{i})] Erd\"{o}s-Renyi network: We randomly pick exactly one element
in each row of $W_{0}$ and set that element to $1$. This is a random graph
with in-degree equal to $1$ for every individual (\citealp{ErdosRenyi1960}).
Such a network could be observed in practice if connections were formed
independently of one another. With $N=30$, the resulting density of links is 
$3.45$\%.

\item[(\emph{ii})] Political party network: There are two parties, each with
a party leader. The leader directly affects the behavior of half the party
members. We assume that one party has twice members as the
other. More specifically, we assume individuals $i=1,\dots ,\frac{N}{3}$ are
affiliated with Party A and led by individual $1$; individuals $i=\frac{N}{%
3}+1,\dots ,N$ are affiliated with Party B and led by individual $\frac{N}{%
3}+1$. This difference in party size allows us to evaluate our ability to
recover and identify central leaders, even in the smaller party. To test the
procedure further, we add one random link per row to represent ties that are
not determined by links to the party leader. We simulate this network for
various choices of $N$. If $N$ is not a multiple of three, we round $\frac{N%
}{3}$ to the nearest integer. With $N=30$, this network has a density of $%
5.17$\%.

\item[(\emph{iii})] Coleman's (1964)\nocite{Coleman1964} high school
friendship network survey: In $1957-8$, students in a small high school in
Illinois were asked to name \textquotedblleft fellows that they go around
with most often.\textquotedblright\ A link was considered if the student
nominated a peer in either survey wave. The full network has $N=73$ nodes,
of which $70$ are non-isolated and so have at least one link to another
student. On average, students named just over five friendship peers. This
network has a density of $7.58$\%. Furthermore, the in-degree distribution shows
that most individuals received a small number of links, while a small number
received many peer nominations.

\item[(\emph{iv})] Banerjee \emph{et al.}'s (2013)\nocite{banerjeeetal2013}
village network survey: These authors conducted a census of households in $%
75 $ villages in rural Karnataka, India, and survey questions included
several about relationships with other households in the village. To begin
with, we use social ties based on family relations (later examining
insurance networks). We focus on village $10$, which is comprises $N=77$
households and so similar in size to network \emph{(iii)}. In this village,
there are $65 $ non-isolated households, with at least one family link to
another household. This network has a density of $5.07$\%.
\end{itemize}

We simulate the real-world networks (\emph{iii}) and (\emph{iv}) using the
non-isolated nodes in each (so $N=70$ and $65$ respectively). We proceed as
follows for each network. First, for each node, we randomly assign one of
its links to have three times the strength of other links. As the
underlying data generating process is assumed to allow for common time
effects ($\alpha _{t}$), we then set the weight of all the node's other links to be equal and such that row-sum normalization (A4$^\prime$) is complied with.
For example, if in a given row of $W_{0}$ there are two links, one will be
randomly selected to be set to $.75$, and the other will be set to $.25$. If there
are three links, one is set to $.6$ and the other two are set to each have
weight $.2$ to maintain row-sum normalization, and so on. For the
Erd\"{o}s-Renyi network, there are thus only strong ties, as each node has only
one link to another node.

As we consider larger networks, we typically expect them to have more
nonzero entries in each row of $W_{0}$, but row-sum normalization means
that each weaker link will be of lower value. This works \emph{against} the
detection of weaker links using estimation methods involving penalization
because they impose that small-parameter estimates shrink to zero.\footnote{%
\citet{CanerZhang2014} state that \textquotedblleft local to zero
coefficients should be larger than \textcolor{black}{$T^{-\frac{1}{2}}$} to
be differentiated from zero.\textquotedblright} Finally, to aid exposition,
we set a threshold value for link strength to distinguish ``strong'' and
``weak'' links. A strong (weak) link is defined as one for which $W_{0,ij}>$ $%
(\leq )$ $.3$.

Panel data for each of the four simulations is generated as%
\begin{equation*}
y_{t}=(I-\rho _{0}W_{0})^{-1}(x_{t}\beta +W_{0}x_{t}\gamma +\alpha _{t}\iota
+\alpha ^{\ast }+\epsilon _{t}),
\end{equation*}%
where $\alpha _{t}$ is a (scalar) time effect and $\alpha ^{\ast }$ is a $%
N\times 1$ vector of fixed effects, drawn respectively from $N(1,1)$ and $%
N(\iota ,I_{N\times N})$ distributions. We consider $T=%
\{5,10,15,25,50,75,100,125,150\}$. The true parameters are set to $\rho
_{0}=.3$, $\beta _{0}=.4$, and $\gamma _{0}=.5$ (thus satisfying Assumption
A3). The exogenous variable ($x_{t}$) and error term ($\epsilon _{t}$) are
simulated as standard Gaussian, both generated from $N(0_{N},I_{N\times N})$
distributions. 
\textcolor{black}{This is similar to the variance terms set in
other papers, e.g., \citet{Lee2004}.} We later conduct a series of
robustness checks to evaluate the sensitivity of the simulations to
alternative parameter choices and the presence of common and
individual-level shocks. For each combination of parameters, we conduct $%
1,000$ simulation runs.

\subsection{Robustness}

\textcolor{bluez}{Table A2} presents results for the recovered stylized
networks under varying network sizes, $N=\{15,30,50\}$. Differences between
the true and estimated networks are fairly constant as $N$ increases: even
for small $N=15$, a large proportion of zeros and non-zeros are correctly
estimated. In all cases, biases in $\hat{\rho}$ and $\hat{\gamma}$ decrease
with larger $T$.

\textcolor{bluez}{Table A3} conducts robustness checks on the sensitivity of
the estimates to parameter choices for the stylized networks. We consider the 
true parameters $\rho _{0}=\{.1,.3,.7,.9\}$, $\gamma _{0}=\{.3,.7\}$, and $\beta
_{0}=\{.2,.6\}$. We also introduce a common shock in the disturbance
variance-covariance matrix by varying $q$ in%
\begin{equation*}
\epsilon _{t}\sim N\left( 0,\left[ 
\begin{array}{cccc}
1 & q & \cdots & q \\ 
q & 1 & \cdots & q \\ 
\vdots & \vdots & \ddots & \vdots \\ 
q & q & \cdots & 1%
\end{array}%
\right] \right)
\end{equation*}%
where we consider $q=\{.3,.5,.8,1\}$. We find the procedure to be robust to
the true values of $\rho _{0}$, $\beta _{0}$, $\gamma _{0}$, and $q$. The
method performs well in all scenarios.

The next set of robustness checks demonstrates the gains from using the
Adaptive Elastic Net GMM\ estimator over alternative estimators. %
\textcolor{bluez}{Table A4} shows simulation results using Adaptive Lasso
estimates of the interaction matrix $\Pi _{0}$, so estimating and penalizing
the reduced-form. The Adaptive Lasso estimator performs worse,
and increased sample sizes are necessary to achieve similar performance
compared to the Adaptive Elastic Net GMM. \textcolor{bluez}{Table A5} then
shows the performance of the procedure based on OLS estimates of $\Pi _{0}$.
Given OLS requires $m\ll NT$, we use a time dimension ten times larger, $%
T=\{500,1000,1500\}$, and still find a deterioration in performance compared
to the Adaptive Elastic Net GMM estimator.\footnote{%
As opposed to the penalized estimates, all OLS estimates are different from
zero. We compute the \lq\lq\% True Zeroes\rq\rq\ as the proportion of true
zero elements in the social interaction matrix that are estimated as smaller
than .05.}

Taken together, these robustness checks suggest the Adaptive Elastic Net
GMM\ estimator is preferred over Adaptive Lasso and OLS\ estimators. As discussed in the text, the procedure recovers true strong links. In finite samples, weak links can be detected as zeros due to the shrinkage estimator employed. In turn, row-sum normalization may imply that the strength of strong edges is over-estimated. We also showed that \emph{a
fortiori}, the procedure can recover network- and node-level statistics. It does so in networks that
vary in size and complexity, and as the underlying social interactions model
varies in the strength of endogenous and exogenous social effects, and the
structure of shocks.



\section*{Supplementary material}

## Table 1: Geographic Neighbors

**Dependent variable: Change in per capita income and corporate taxes**
Coefficient estimates, standard errors in parentheses

| | Besley and Case (1995) Sample | | Extended Sample | |
|---|---|---|---|---|
| | (1) OLS | (2) 2SLS | (3) OLS | (4) 2SLS |
| Geographic Neighbors' Tax Change (t - [t-2]) | .375*** | .868*** | .271*** | .642*** |
| | (.120) | (.273) | (.075) | (.152) |
| Period | 1962-1988 | 1962-1988 | 1962-2015 | 1962-2015 |
| First Stage (F-stat) | | 6.267 | | 27.320 |
| Controls | Yes | Yes | Yes | Yes |
| State and Year Fixed Effects | Yes | Yes | Yes | Yes |
| Observations | 1,296 | 1,248 | 2,592 | 2,544 |

**Notes:** *** denotes significance at 1%, ** at 5%, and * at 10%. In all specifications, a pair of states are considered neighbors if they share a geographic border. The sample covers 48 mainland US states. In Columns 1 and 2 the sample runs from 1962 to 1988 (as in Besley and Case (1995)). In Columns 3 and 4 the sample is extended to run from 1962 to 2015. The dependent variable is the change in state i's total taxes per capita in year t. OLS regressions estimates are shown in Columns 1 and 3. Columns 2 and 4 show 2SLS regressions where each geographic neighbors' tax change is instrumented by lagged neighbor's state income per capita and unemployment rate. At the foot of Columns 2 and 4 we report the p-value on the F-statistic from the first stage of the null hypothesis that instruments are jointly equal to zero. All regressions control for state i's income per capita in 1982 US dollars, state i's unemployment rate, the proportion of young (aged 5-17) and elderly (aged 65+) in state i's population, and the state governor's age. All specifications include state and time fixed effects. With the exception of governor's age, all variables are differenced between period t and period t-2. Robust standard errors are reported in parentheses.

# Table 2: Economic Neighbors

**Dependent variable: Change in per capita income and corporate taxes**
Coefficient estimates, standard errors in parentheses

|  | No Exogenous Social Effects | | | | Exogenous Social Effects | | |
|---|---|---|---|---|---|---|---|
|  | (1) Initial | (2) OLS | (3) 2SLS: IVs are Characteristics of Neighbors | (4) Initial | (5) OLS | (6) 2SLS: IVs are Characteristics of Neighbors | (7) 2SLS: IVs are Characteristics of Neighbors-of Neighbors |
| Economic Neighbors' Tax Change (t - [t-2]) | .886 | .378*** | .641*** | .645 | .145** | .332* | .608*** |
|  |  | (.061) | (.060) |  | (.072) | (.199) | (.220) |
| Period | 1962-2015 | | | | 1962-2015 | | |
| First Stage (F-stat) |  |  | 19.353 |  |  | 9.571 | 10.480 |
| Controls | Yes | Yes | Yes | Yes | Yes | Yes | Yes |
| State and Year Fixed Effects | Yes | Yes | Yes | Yes | Yes | Yes | Yes |
| Observations | 2,952 | 2,952 | 2,544 | 2,952 | 2,952 | 2,544 | 2,592 |

**Notes:** *** denotes significance at 1%, ** at 5%, and * at 10%. The sample covers 48 mainland US states running from 1962 to 2015. The dependent variable is the change in state i's total taxes per capita in year t. We allow for exogenous social effects in Columns 4 to 7. In subsequent OLS and IV regressions, the economic neighbors' effect is calculated as the weighted average of economic neighbors' variables. OLS regressions estimates are shown in Column 2, 3 and 5. Column 3 and 6 show the 2SLS regression where economic neighbors' tax change is instrumented by lagged neighbor's state income per capita and unemployment rate. Column 7 shows a 2SLS regression where each geographic neighbors' tax change is instrumented by lagged neighbor-of-neighbor's state income per capita and unemployment rate. At the foot of Columns 3, 6 and 7 we report the p-value on the F-statistic from the first stage of the null hypothesis that instruments are jointly equal to zero. All regressions control for state i's income per capita in 1982 US dollars, state i's unemployment rate, the proportion of young (aged 5-17) and elderly (aged 65+) in state i's population, and the state governor's age. All specifications include state and time fixed effects. With the exception of governor's age, all variables are differenced between period t and period t-2. Robust standard errors are reported in parentheses.

# Table 3: Geographic Versus Economic Neighbor Networks

|  | Geographic Neighbor Network | Economic Neighbor Network |
|---|---|---|
| **Number of Edges** | 214 | 144 |
| **Edges in Both Networks** | 79 | 79 |
| **Edges in W-geo only** | 135 | |
| **Edges in W-econ only** | | 65 |
| **Clustering** | .1936 | .0259 |
| **Reciprocated Edges** | 100% | 29.17% |
| **Degree Distribution Across Nodes (states)** | | |
| out-degree | 4.458 (1.597) | 3.000 (1.185) |
| in-degree | | 3.000 (2.073) |

**Notes:** This compares statistics derived from the geographic network of US states to those from the estimated economic network among US states. The number of edges, edges in both networks, edges in W-geo only, edges in W-econ only counts the number of edges in those categories. Reciprocated edges is the frequency of in-edges that are reciprocated by out-edges (by construction, this is 100% for geographic networks). The clustering coefficient is the frequency of the number of fully connected triplets over the total number of triplets. The degree distribution across nodes counts the average number of connections (standard deviation in parentheses): we show this separately for in-degree and out-degree (by construction, these are identical for geographic networks).

# Table 4: Predicting Links to Economic Neighbors

**Columns 1-7: Linear Probability Model; Column 8: Tobit**
**Dependent variable (Cols 1-7): =1 if Economic Link Between States Identified**
**Dependent variable (Col 8): =Weighted Link Between States**
**Coefficient estimates, standard errors in parentheses**

| | Geography | | | Economic and Demographic Homophyly | Labor Mobility | Political Homophyly | Tax Havens | Tobit, Partial Avg Effects |
|---|---|---|---|---|---|---|---|---|
| | (1) | (2) | (3) | (4) | (5) | (6) | (7) | (8) |
| Geographic Neighbor | .699*** | | .701*** | .701*** | .698*** | .698*** | .697*** | .068*** |
| | (.030) | | (.032) | (.030) | (.031) | (.031) | (.031) | (.006) |
| Distance | | -.453*** | -.008 | | | | | |
| | | (.033) | (.024) | | | | | |
| Distance sq. | | .0949*** | .003 | | | | | |
| | | (.007) | (.006) | | | | | |
| GDP Homophyly | | | | 2.409** | 2.369* | 2.296* | 1.046 | .322 |
| | | | | (1.183) | (1.186) | (1.193) | (1.150) | (.302) |
| Demographic Homophyly | | | | .222 | .235 | .241 | .256 | .077 |
| | | | | (.226) | (.226) | (.228) | (.225) | (.067) |
| Net Migration | | | | | .044* | .044* | -0.032 | 0.001 |
| | | | | | (.025) | (.025) | (.025) | (.002) |
| Political Homophyly | | | | | | -.057 | -.083** | -.025* |
| | | | | | | (.042) | (.042) | (.014) |
| Tax Haven Sender | | | | | | | .107*** | .021*** |
| | | | | | | | (.024) | (.005) |
| Adjusted R-squared | 0.427 | 0.152 | 0.427 | 0.428 | 0.429 | 0.429 | 0.440 | - |
| Observations | 2,256 | 2,256 | 2,256 | 2,256 | 2,256 | 2,256 | 2,256 | 2,256 |

**Notes:** *** denotes significance at 1%, ** at 5%, and * at 10%. The specifications in Columns 1-7 are cross-sectional linear probabilities models where the dependent variable is equal to 1 if two states are linked, and zero otherwise. In Column 8 the dependent variable is the weighted link between states. Column 7 reports the partial average effects from a Tobit model. A pair of states is considered a first-degree geographic neighbor if they share a border. Distance and distance squared are calculated from the centroids of states' capital cities. GDP homophyly is the absolute difference of states' GDP per capita. Demographic homophyly is the absolute difference of share of young (aged 5-17) plus the absolute difference of the share of elderly in states' population (aged 65+). Net migration is based on individuals' tax returns (Source: Internal Revenue Service, https://www.irs.gov/statistics/soi-tax-stats-migration-data). Political homophyly is equal to one if a pair of states have governors of same party at given year. Nevada, Delaware, Montana, South Dakota, Wyoming and New York are considered tax haven states. Time averages are taken for all explanatory variables. Robust standard errors in parentheses.

**Table 5: Gubernatorial Term Limits**

Dependent variable: Change in per capita income and corporate taxes
Coefficient estimates, standard errors in parentheses
IVs: Characteristics of Neighbors-of-Neighbors

|  | All Governors | | Governor Cannot Run for Re-election | | Governor Can Run for Re-election | |
|---|---|---|---|---|---|---|
|  | (1) OLS | (2) 2SLS | (3) OLS | (4) 2SLS | (5) OLS | (6) 2SLS |
| Economic Neighbors' tax change (t - [t-2]) | .145** | .608*** | .016 | .937* | .182** | .543** |
|  | (.072) | (.220) | (.105) | (.534) | (.084) | (.237) |
| Period | 1962-2015 | | 1962-2015 | | 1962-2015 | |
| First Stage (F-stat) |  | 10.480 |  | 2.835 |  | 10.120 |
| Controls | Yes | Yes | Yes | Yes | Yes | Yes |
| State and Year Fixed Effects | Yes | Yes | Yes | Yes | Yes | Yes |
| Observations | 2,592 | 2,592 | 640 | 640 | 1,917 | 1,917 |

**Notes:** *** denotes significance at 1%, ** at 5%, and * at 10%. The sample in Columns 1 and 2 covers 48 mainland US states running from 1962 to 2015. In Columns 3 and 4 we use the subsample of state-years in which the governor that faced term limits in the subsequent gubernatorial election. In Columns 5 and 6 we use the subsample of state-years in which the governor did not face term limits in the subsequent gubernatorial election, and so could run for reelection. The dependent variable is the change in state i's total taxes per capita in year t. We first estimate our procedure which outputs parameters and the network of economic neighbors. We penalize geographic neighbors throughout and also allow for exogenous social effects. OLS regressions estimates are shown in Columns 1, 3 and 5. Columns 2, 4 and 6 show a 2SLS regression where each geographic neighbors' tax change is instrumented by lagged neighbor-of-neighbor's state income per capita and unemployment rate. At the foot of Columns 2, 4 and 6 we report the p-value on the F-statistic from the first stage of the null hypothesis that instruments are jointly equal to zero. All regressions control for state i's income per capita in 1982 US dollars, state i's population, state i's unemployment rate, the proportion of young (aged 5-17) and elderly (aged 65+) in state i's population, and the state governor's age. All specifications include state and time fixed effects. With the exception of governor's age, all variables are differenced between period t and period t-2. Robust standard errors are reported in parentheses.

**Figure 1A: Network Graph of US States, Geographic Neighbors**

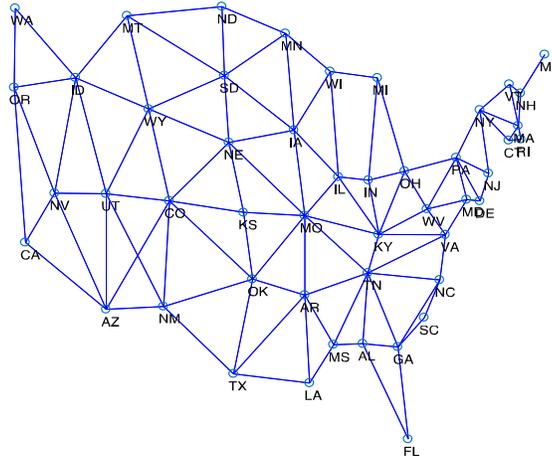

**Notes:** Figure 3A represents the continental US states (N=48). An edge is drawn between a pair of states if they share a geographic border. State abbreviations are as used by US Post Office (http://about.usps.com/who-we-are/postal-history/state-abbreviations.pdf).

**Figure 1B: Network Graph of US States, Identified Economic Neighbors**

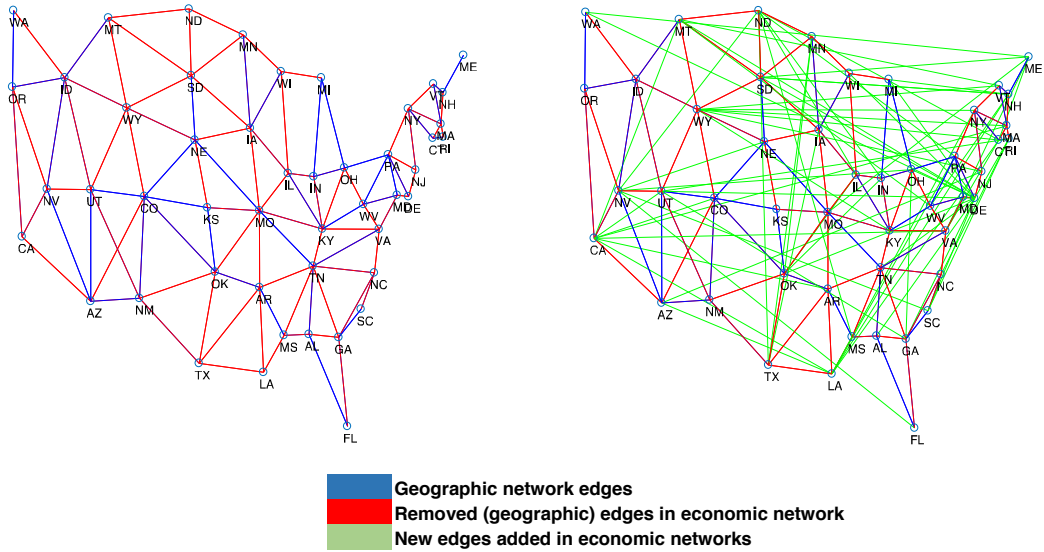

- Geographic network edges
- Removed (geographic) edges in economic network
- New edges added in economic networks

**Notes:** Figure 3B represents the continental United States (N=48). The economic network is derived from our preferred specification, where we penalize geographic neighbors to states, and allow for exogenous social effects. A blue edge is drawn between a pair of states if they are geographic neighbors and were estimated as connected. A red edge is drawn between a pair of states if they are geographic neighbors but were not estimated as connected. A green edge is drawn between a pair of states if they are not geographic neighbors and were estimated connected. The left hand side graph just shows read and blue edges. The right hand side shows all three types of edges. State abbreviations are as used by US Post Office (http://about.usps.com/who-we-are/postal-history/state-abbreviations.pdf).

# Figure 2: Strength of Ties and Reciprocity

**Panel A: Histogram of Strength of Ties, Conditional on $W_{0,ij}>0$**

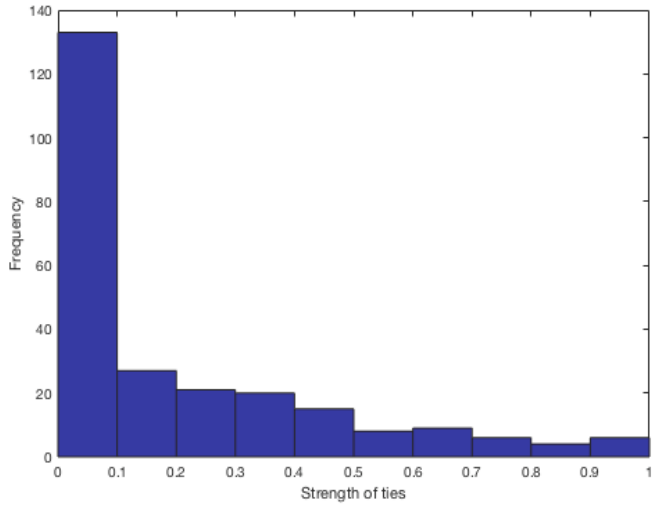

**Panel B: In-network for California**     **Panel C: Out-network for California**

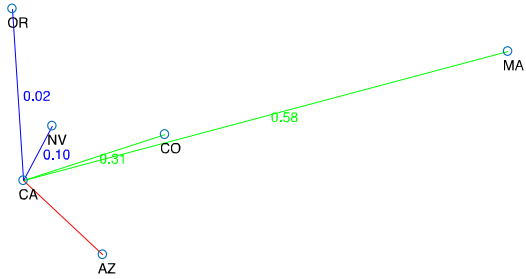     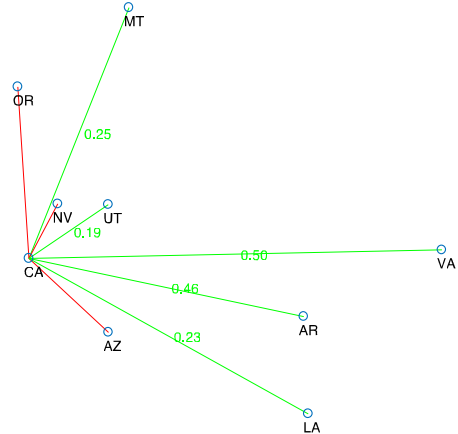

- **Geographic network edges** (blue)
- **Removed (geographic) edges in economic network** (red)
- **New edges added in economic networks** (green)

**Notes:** Panel A is the histogram of ties in the economic network, conditional on non-zero ties. Panels B and C show the in-network and out-network of California as derived from our preferred specification, where we penalize geographic neighbors to states, and allow for exogenous social effects. The in-network are the states that determine tax setting in California. The out-network is the states in which taxes are set in direct response to those in California. A blue edge is drawn between a pair of states if they are geographic neighbors and were estimated as connected. A red edge is drawn between a pair of states if they are geographic neighbors but were not estimated as connected. A green edge is drawn between a pair of states if they are not geographic neighbors and were estimated connected. State abbreviations are as used by US Post Office (http://about.usps.com/who-we-are/postal-history/state-abbreviations.pdf).

# Figure 3: In- and Out-degree Distribution

**Panel A: In-degree distribution**

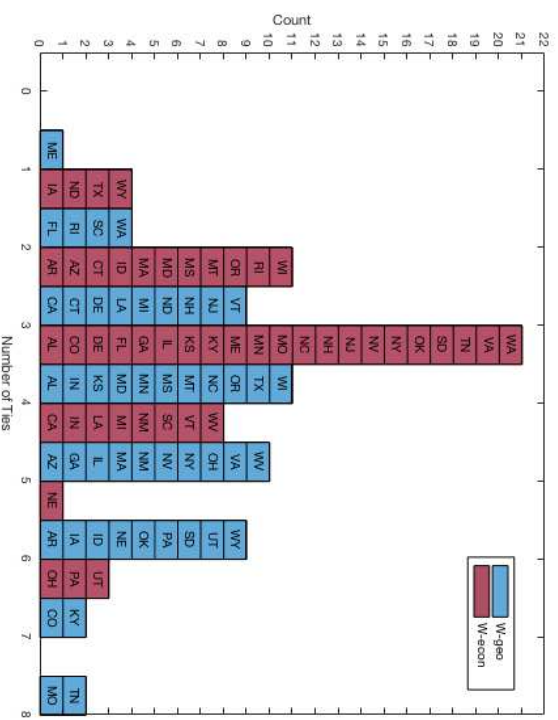

**Panel B: Out-degree distribution**

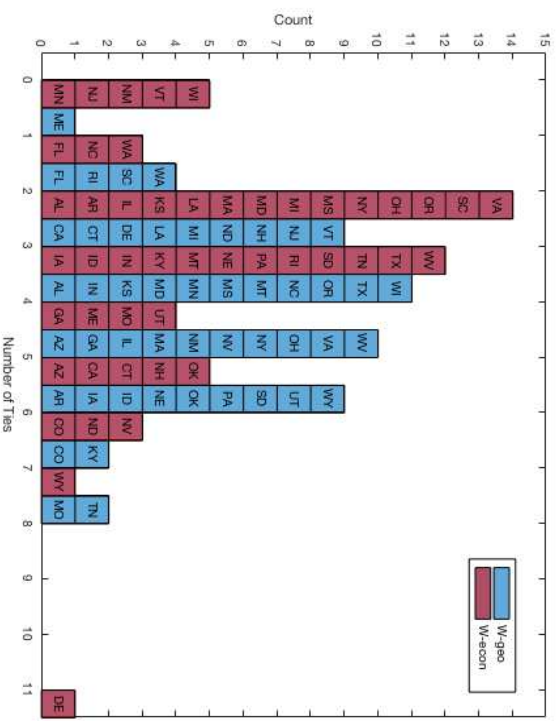

**Notes:** In-degree distribution (Panel A) and out-degree distribution (Panel B). Distribution calculated from geographic neighbors' network (W-geo) in blue. Distribution calculated from economic neighbor's network in (W-econ) in red. State abbreviations are as used by US Post Office (http://about.usps.com/who-we-are/postal-history/state-abbreviations.pdf).

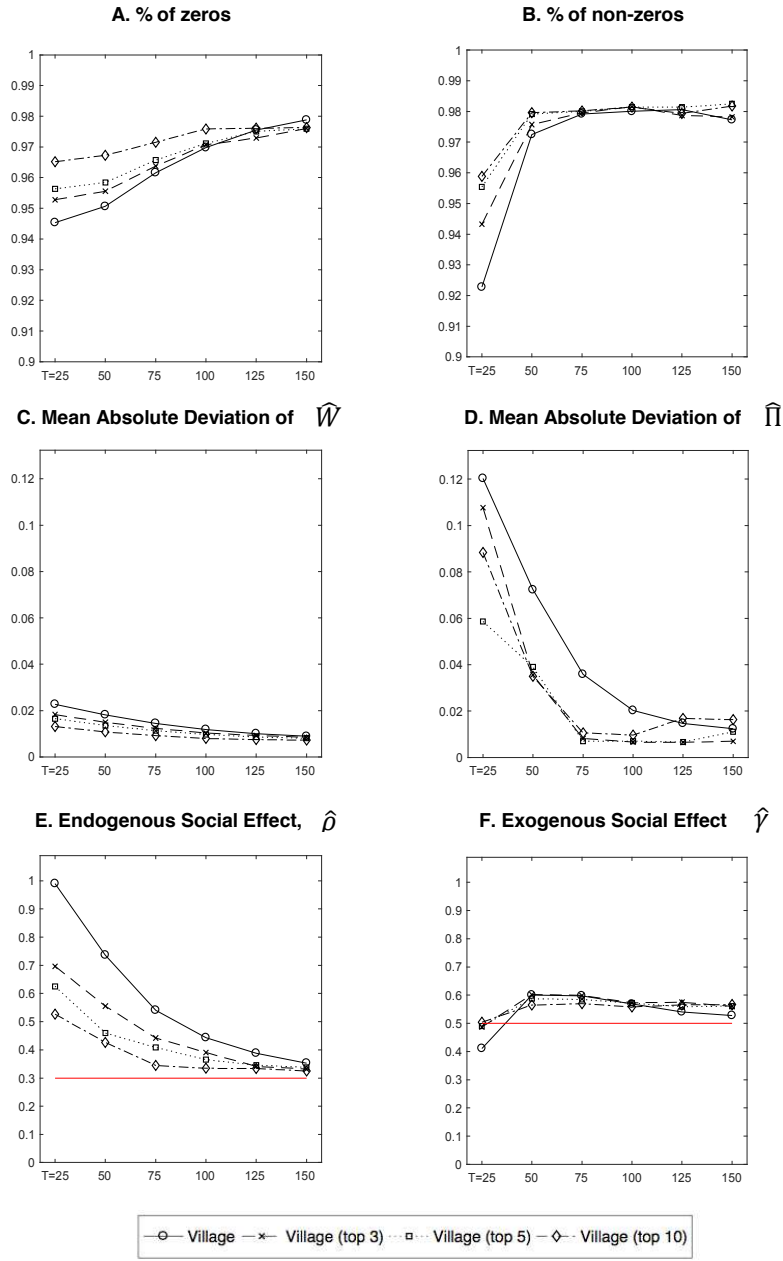

**Figure 4: Simulation Results, Adaptive Elastic Net GMM**
**Partial Knowledge of $W_0$**

**Notes:** These simulation results are based on the Banerjee et al. (2013) village network, using the Adaptive Elastic Net GMM algorithm, with penalization parameters chosen by BIC, under various assumptions about knowledge of the true network and time periods T=25, 50, 100, 125 and 150. The "Village" case refers to the simulation implemented without knowledge about the true network. "Village (top 3)" refers to the case where all connections of the three households with highest out-degrees are assumed to be known. "Village (top 5)" and "Village (top 103)" are analogously defined. In all cases, 1000 Monte Carlo iterations were performed. The true parameters are rho-0=.3, beta-0=.4 and gamma-0=.5. In Panel A, the % of zeroes refers to the proportion of true zero elements in the social interaction matrix that are estimated as smaller than .05. In Panel B, the % of non-zeros refers to the proportion of true elements greater than .3 in the social interaction matrix that are estimated as non-zeros. In Panels C and D, the Mean Absolute Deviations are the mean absolute error of the estimated network compared to the true network for the social interaction matrix W and the reduced form matrix respectively. In Panels E and F, the true parameter values are marked in the horizontal red lines. The recovered parameter are the estimated parameters averaged across iterations. All specifications include time and node fixed effects.

## Figure 5: Network Graph of US States, Identified Economic Neighbors by Subsamples

**Panel A: Econ Network, 1962-1988 subsample**     **Panel B: Econ Network, 1989-2015 subsample**

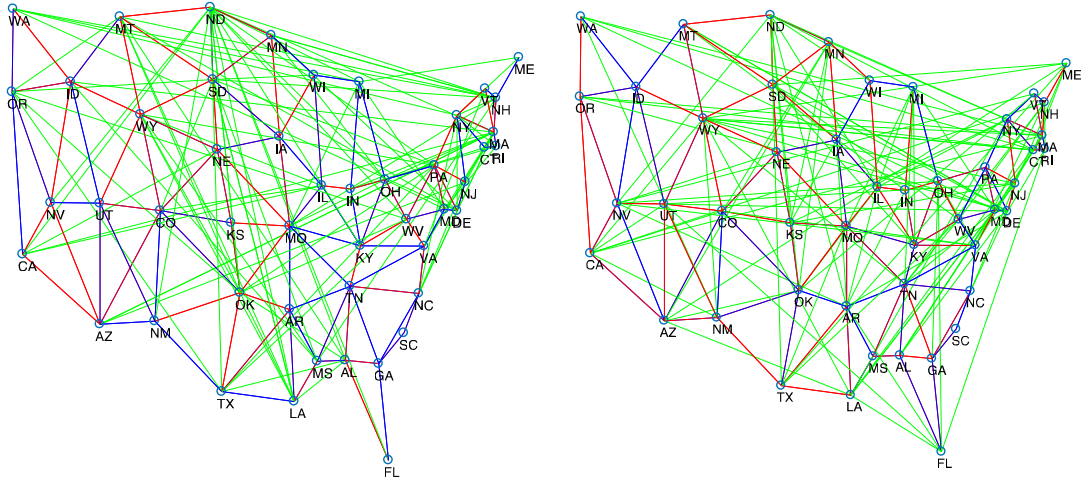

**Notes:** Panels A and B represents the continental United States (N=48). Panel A is estimated withe 1962-1988 subsample. Panel B is estimated with the 1989-2015 subsample. The economic network is derived from our preferred specification, where we penalize geographic neighbors to states, and allow for exogenous social effects. A blue edge is drawn between a pair of states if they are geographic neighbors and were estimated as connected. A red edge is drawn between a pair of states if they are geographic neighbors but were not estimated as connected. A green edge is drawn between a pair of states if they are not geographic neighbors and were estimated connected. The left hand side graph just shows read and blue edges. The right hand side shows all three types of edges. State abbreviations are as used by US Post Office (http://about.usps.com/who-we-are/postal-history/state-abbreviations.pdf).

### Panel C: Geographic Versus Economic Neighbor Networks

| | Geographic Neighbor Network | Economic Neighbor Network (Full Sample) | Economic Neighbor Network (1962-1988) | Economic Neighbor Network (1989-2015) |
|---|---|---|---|---|
| **Number of Edges** | 214 | 144 | 197 | 185 |
| **Edges in both W-geo and W-econ** | | 79 | 108 | 86 |
| **Edges in W-econ only** | | 65 | 89 | 99 |
| **Edges in W-geo only** | | 135 | 106 | 128 |
| **Clustering** | .1936 | .0259 | .0389 | .0394 |
| **Reciprocated Edges** | 100% | 29.17% | 35.53% | 23.78% |
| **Degree Distribution Across Nodes (states)** | | | | |
| out-degree | 4.458 (1.597) | 3.000 (1.185) | 4.104 (1.704) | 3.854 (1.473) |
| in-degree | | 3.000 (2.073) | 4.104 (2.707) | 3.854 (2.518) |

**Notes:** This compares statistics derived from the geographic network of US states to those from the estimated economic network among US states, for the three samples ("Full Sample", 1962-2015; 1962-1988; 1989-2015). The number of edges, edges in both networks, edges in W-geo only, edges in W-econ only counts the number of edges in those categories. Numbers are relative to the W-geo network in the first column. Reciprocated edges is the frequency of in-edges that are reciprocated by out-edges (by construction, this is 100% for geographic networks). The clustering coefficient is the frequency of the number of fully connected triplets over the total number of triplets. The degree distribution across nodes counts the average number of connections (standard deviation in parentheses): we show this separately for in-degree and out-degree (by construction, these are identical for geographic networks).

## Table A1: True and Recovered Networks

| | Erdos-Renyi | Political Party | High school<br>*Coleman (1964)* | Village<br>*Banerjee et al. (2013)* |
|---|---|---|---|---|
| | | **A. True Networks** | | |
| **Number of nodes** | **30** | **30** | **70** | **65** |
| **(a) Network-wide statistics** | | | | |
| *Number of edges* | 30 | 45 | 366 | 240 |
| *Number of strong edges* | 30 | 30 | 70 | 65 |
| *Number of weak edges* | 0 | 15 | 296 | 175 |
| *Number of reciprocated edges* | 2 | 2 | 184 | 240 |
| *Clustering coefficient* | - | .000 | .120 | .141 |
| *Number of components* | 12 | 11 | 3 | 3 |
| *Size of maximal component* | 10 | 16 | 68 | 51 |
| *Standard deviation of the diagonal of squared W* | .254 | .254 | .167 | .239 |
| **(b) Node-level statistics** | | | | |
| *In-degree distribution* | 1.00 (0.00) | 1.50 (.509) | 5.23 (2.04) | 3.69 (2.35) |
| *Out-degree distribution* | 1.00 (1.05) | 1.50 (2.49) | 5.23 (3.64) | 3.69 (2.35) |
| *Nodes with highest out-degree* | { 7, 11, 26 } | { 1, 11, 28 } | { 21, 22, 69 } | { 16, 35, 57 } |
| | | **B. Recovered Networks** | | |
| **(a) Network-wide statistics** | | | | |
| *Number of edges* | 30 | 38 | 210 | 194 |
| *Number of strong edges* | 30 | 30 | 70 | 68 |
| *Number of weak edges* | 0 | 8 | 140 | 126 |
| *Number of reciprocated edges* | 2 | 2 | 184 | 170 |
| *Clustering coefficient* | - | .000 | .162 | .134 |
| *Number of components* | 12 | 11 | 1 | 4 |
| *Size of maximal component* | 10 | 14 | 70 | 48 |
| **(b) Node-level statistics** | | | | |
| *In-degree distribution* | 1.00 (0.00) | 1.27 (.450) | 3.00 (1.18) | 2.99 (1.29) |
| *Out-degree distribution* | 1.00 (1.05) | 1.27 (1.76) | 3.00 (1.02) | 2.99 (1.15) |
| *Nodes with highest out-degree* | { 7, 11, 26 } | { 1, 11, 28 } | { 21, 48, 69 } | { 16, 35, 57 } |

**Notes:** Panel A refers to the true networks. Panel B refers to the recovered networks. In each Panel, the summary statistics are divided into network-wide and node-level statistics. Strong edges are defined as those with strength greater than or equal to .3. For the in-degree and out-degree distribution, the mean is shown and the standard deviation is in parentheses. The nodes with the highest out-degree are those with the greatest influence on others, and are calculated as the column-sum of the social interaction matrix. The recovered networks statistic are calculated over the average network across simulations with T=100.

# Table A2: Simulation Results, Adaptive Elastic Net GMM, Alternative Network Sizes

| | A. Erdos-Renyi | | | | | | | | | B. Political party | | | | | | | | |
|---|---|---|---|---|---|---|---|---|---|---|---|---|---|---|---|---|---|---|
| | N = 15 | | | N = 30 | | | N = 50 | | | N = 15 | | | N = 30 | | | N = 50 | | |
| | T=50 | 100 | 150 | T=50 | 100 | 150 | T=50 | 100 | 150 | T=50 | 100 | 150 | T=50 | 100 | 150 | T=50 | 100 | 150 |
| % True Zeroes | .954 | .976 | .973 | .957 | .992 | .997 | .947 | .965 | .979 | .950 | .977 | .979 | .943 | .978 | .989 | .943 | .961 | .976 |
|  | (.033) | (.031) | (.031) | (.009) | (.004) | (.002) | (.007) | (.006) | (.005) | (.031) | (.025) | (.024) | (.009) | (.006) | (.005) | (.007) | (.006) | (.006) |
| % True Non-Zeroes | .899 | .919 | .924 | .958 | .960 | .962 | .977 | .978 | .977 | .914 | .925 | .932 | .930 | .954 | .960 | .970 | .977 | .977 |
|  | (.051) | (.032) | (.026) | (.017) | (.016) | (.011) | (.007) | (.007) | (.008) | (.038) | (.022) | (.009) | (.037) | (.023) | (.016) | (.015) | (.009) | (.007) |
| $MAD(\widehat{W})$ | .021 | .007 | .004 | .017 | .003 | .001 | .020 | .011 | .006 | .032 | .019 | .017 | .028 | .014 | .010 | .023 | .014 | .009 |
|  | (.014) | (.011) | (.008) | (.004) | (.001) | (.001) | (.002) | (.001) | (.001) | (.011) | (.007) | (.004) | (.004) | (.002) | (.002) | (.002) | (.001) | (.001) |
| $MAD(\widehat{\Pi})$ | .025 | .012 | .009 | .020 | .004 | .003 | .054 | .016 | .007 | .034 | .021 | .019 | .028 | .014 | .012 | .065 | .019 | .009 |
|  | (.012) | (.008) | (.006) | (.004) | (.002) | (.001) | (.018) | (.004) | (.002) | (.009) | (.006) | (.003) | (.004) | (.002) | (.001) | (.028) | (.004) | (.001) |
| $\widehat{\rho}$ | .262 | .270 | .276 | .286 | .286 | .283 | .667 | .398 | .270 | .235 | .245 | .241 | .209 | .228 | .223 | .700 | .383 | .242 |
|  | (.069) | (.044) | (.038) | (.079) | (.026) | (.022) | (.079) | (.050) | (.050) | (.078) | (.049) | (.038) | (.084) | (.035) | (.029) | (.100) | (.069) | (.051) |
| $\widehat{\beta}$ | .403 | .399 | .400 | .405 | .400 | .400 | .380 | .399 | .401 | .405 | .400 | .397 | .404 | .399 | .398 | .380 | .400 | .401 |
|  | (.039) | (.028) | (.022) | (.028) | (.018) | (.015) | (.025) | (.015) | (.012) | (.040) | (.028) | (.022) | (.029) | (.019) | (.015) | (.024) | (.015) | (.012) |
| $\widehat{\gamma}$ | .577 | .521 | .507 | .670 | .516 | .498 | .748 | .686 | .595 | .550 | .481 | .459 | .669 | .501 | .463 | .686 | .666 | .566 |
|  | (.094) | (.059) | (.046) | (.054) | (.022) | (.018) | (.118) | (.056) | (.028) | (.093) | (.060) | (.050) | (.060) | (.027) | (.021) | (.128) | (.063) | (.032) |

**Notes:** These simulation results are based on the Adaptive Elastic Net GMM algorithm, with penalization parameters chosen by BIC, under various true networks, network sizes and time periods T=50, 100 and 150. In all cases, 1000 Monte Carlo iterations were performed. The true parameters are rho-0=.3, beta-0=.4 and gamma-0=.5. The % of true zeroes refers to the proportion of true zero elements in the social interaction matrix that are estimated as smaller than .05. The % of true non-zeroes refers to the proportion of true elements greater than .3 in the social interaction matrix that are estimated as non-zeros. The Mean Absolute Deviations are the mean absolute error of the estimated network compared to the true network for the social interaction matrix W and the reduced form matrix respectively. The recovered parameter are the estimated parameters averaged across iterations. All specifications include time and node fixed effects. Standard errors across iterations are in parentheses.

## Table A3: True and Recovered Village Networks

| | Village<br>Family<br>*Banerjee et al. (2013)* | Village<br>Savings and Insurance<br>*Banerjee et al. (2013)* |
|---|---|---|
| | **A. True Networks** | |
| **Number of nodes** | **65** | **65** |
| **(a) Network-wide statistics** | | |
| *Number of edges* | 240 | 343 |
| *Number of strong edges* | 65 | 47 |
| *Number of weak edges* | 175 | 296 |
| *Number of reciprocated edges* | 240 | 340 |
| *Clustering coefficient* | .141 | .073 |
| *Number of components* | 3 | 6 |
| *Size of maximal component* | 51 | 62 |
| *Standard deviation of the diagonal of squared W* | .239 | .159 |
| **(b) Node-level statistics** | | |
| *In-degree distribution* | 3.69 (2.35) | 4.90 (3.42) |
| *Out-degree distribution* | 3.69 (2.35) | 4.90 (3.43) |
| *Nodes with highest out-degree* | { 16, 35, 57 } | { 16, 35, 55 } |
| | **B. Recovered Networks** | |
| **(a) Network-wide statistics** | | |
| *Number of edges* | 194 | 269 |
| *Number of strong edges* | 68 | 65 |
| *Number of weak edges* | 126 | 204 |
| *Number of reciprocated edges* | 170 | 250 |
| *Clustering coefficient* | .134 | .058 |
| *Number of components* | 4 | 4 |
| *Size of maximal component* | 48 | 62 |
| **(b) Node-level statistics** | | |
| *In-degree distribution* | 2.99 (1.29) | 3.84 (1.90) |
| *Out-degree distribution* | 2.99 (1.15) | 3.84 (1.98) |
| *Nodes with highest out-degree* | { 16, 35, 57 } | { 16, 35, 55 } |

**Notes:** Panel A refers to the true networks. Panel B refers to the recovered networks. In each Panel, the summary statistics are divided into network-wide and node-level statistics. Strong edges are defined as those with strength greater than or equal to .3. For the in-degree and out-degree distribution, the mean is shown and the standard deviation is in parentheses. The nodes with the highest out-degree are those with the greatest influence on others, and are calculated as the column-sum of the social interaction matrix. The recovered networks statistic are calculated over the average network across simulations with T=100.

# Table A4: Simulation Results, Adaptive Elastic Net GMM, Alternative Parameters

## A. Erdos-Renyi

|  | $\rho_0$ | | | | $\beta_0$ | | | $\gamma_0$ | | | q | | | |
|---|---|---|---|---|---|---|---|---|---|---|---|---|---|---|
|  | .1 | .5 | .7 | .9 | .0 | .8 | .3 | .7 | .3 | .5 | .8 | 1.0 |
| % True Zeroes | .986 | .994 | .991 | .979 | .986 | .987 | .974 | .997 | .996 | .997 | .997 | .997 |
|  | (.005) | (.004) | (.004) | (.005) | (.005) | (.007) | (.006) | (.002) | (.002) | (.003) | (.002) | (.002) |
| % True Non-Zeroes | .951 | .963 | .963 | .956 | .806 | .967 | .961 | .961 | .963 | .962 | .961 | .959 |
|  | (.028) | (.011) | (.011) | (.017) | (.096) | (.005) | (.015) | (.016) | (.012) | (.013) | (.013) | (.018) |
| $MAD(\hat{W})$ | .005 | .002 | .003 | .007 | .014 | .004 | .011 | .001 | .001 | .001 | .000 | .000 |
|  | (.002) | (.001) | (.001) | (.002) | (.005) | (.003) | (.003) | (.001) | (.001) | (.001) | (.001) | (.001) |
| $MAD(\hat{\Pi})$ | .005 | .006 | .013 | .077 | .013 | .007 | .009 | .004 | .003 | .003 | .002 | .002 |
|  | (.002) | (.002) | (.005) | (.051) | (.004) | (.002) | (.002) | (.002) | (.001) | (.001) | (.001) | (.001) |
| $\hat{\rho}$ | .081 | .487 | .709 | .917 | .318 | .244 | .287 | .286 | .285 | .285 | .288 | .292 |
|  | (.034) | (.027) | (.025) | (.012) | (.047) | (.076) | (.036) | (.023) | (.023) | (.020) | (.017) | (.009) |
| $\hat{\beta}$ | .401 | .399 | .397 | .364 | .801 | .801 | .402 | .399 | .401 | .400 | .400 | .400 |
|  | (.018) | (.019) | (.019) | (.023) | (.025) | (.018) | (.019) | (.019) | (.016) | (.012) | (.009) | (.001) |
| $\hat{\gamma}$ | .532 | .512 | .517 | .509 | .432 | .552 | .376 | .698 | .503 | .495 | .493 | .492 |
|  | (.023) | (.022) | (.024) | (.028) | (.070) | (.033) | (.027) | (.021) | (.018) | (.015) | (.012) | (.007) |

## B. Political party

|  | $\rho_0$ | | | | $\beta_0$ | | | $\gamma_0$ | | | q | | | |
|---|---|---|---|---|---|---|---|---|---|---|---|---|---|---|
|  | .1 | .5 | .7 | .9 | .0 | .8 | .3 | .7 | .3 | .5 | .8 | 1.0 |
| % True Zeroes | .971 | .983 | .982 | .966 | .978 | .977 | .959 | .986 | .985 | .992 | .996 | .997 |
|  | (.008) | (.005) | (.006) | (.007) | (.008) | (.007) | (.007) | (.005) | (.004) | (.004) | (.004) | (.003) |
| % True Non-Zeroes | .772 | .813 | .836 | .864 | .469 | .856 | .741 | .834 | .803 | .808 | .816 | .832 |
|  | (.045) | (.036) | (.034) | (.031) | (.147) | (.044) | (.036) | (.037) | (.035) | (.034) | (.027) |
| $MAD(\hat{W})$ | .017 | .012 | .012 | .018 | .029 | .012 | .023 | .010 | .011 | .009 | .008 | .007 |
|  | (.003) | (.002) | (.002) | (.003) | (.007) | (.003) | (.002) | (.002) | (.001) | (.001) | (.001) | (.001) |
| $MAD(\hat{\Pi})$ | .011 | .021 | .043 | .208 | .021 | .017 | .016 | .016 | .013 | .011 | .010 | .009 |
|  | (.002) | (.002) | (.004) | (.026) | (.004) | (.003) | (.002) | (.002) | (.002) | (.001) | (.001) | (.001) |
| $\hat{\rho}$ | .046 | .403 | .598 | .847 | .298 | .171 | .221 | .225 | .224 | .222 | .218 | .218 |
|  | (.038) | (.031) | (.040) | (.040) | (.071) | (.067) | (.048) | (.030) | (.030) | (.028) | (.025) | (.018) |
| $\hat{\beta}$ | .400 | .399 | .398 | .380 | -.015 | .801 | .402 | .400 | .399 | .398 | .398 | .398 |
|  | (.019) | (.020) | (.020) | (.026) | (.029) | (.020) | (.020) | (.019) | (.016) | (.014) | (.009) | (.003) |
| $\hat{\gamma}$ | .519 | .489 | .488 | .469 | .305 | .548 | .391 | .638 | .474 | .454 | .440 | .431 |
|  | (.030) | (.027) | (.031) | (.050) | (.117) | (.038) | (.036) | (.025) | (.023) | (.019) | (.015) | (.010) |

**Notes**: These simulation results are based on the Adaptive Elastic Net GMM algorithm, with penalization parameters chosen by BIC, under various true networks, network sizes, time periods T=100 and parameter values. In all cases, 1000 Monte Carlo iterations were performed. The % of true zeroes refers to the proportion of true zero elements in the social interaction matrix that are estimated as smaller than .05. The % of true non-zeroes refers to the proportion of true elements greater than .3 in the social interaction matrix that are estimated as non-zeros. The Mean Absolute Deviations are the mean absolute error of the estimated network compared to the true network for the social interaction matrix W and the reduced form matrix respectively. The recovered parameter are the estimated parameters averaged across iterations. All specifications include time and node fixed effects. Standard errors across iterations are in parentheses.

# Table A5: Simulation Results, Adaptive Lasso

| | A. Erdos-Renyi | | | | | | | | | B. Political party | | | | | | | | |
|---|---|---|---|---|---|---|---|---|---|---|---|---|---|---|---|---|---|---|
| | N = 15 | | | N = 30 | | | N = 50 | | | N = 15 | | | N = 30 | | | N = 50 | | |
| | T=50 | 100 | 150 | T=50 | 100 | 150 | T=50 | 100 | 150 | T=50 | 100 | 150 | T=50 | 100 | 150 | T=50 | 100 | 150 |
| % True Zeroes | .728 | .798 | .842 | .737 | .848 | .878 | .845 | .846 | .913 | .730 | .799 | .840 | .737 | .849 | .880 | .846 | .846 | .921 |
|  | (.023) | (.024) | (.022) | (.012) | (.008) | (.008) | (.005) | (.006) | (.027) | (.024) | (.025) | (.023) | (.013) | (.009) | (.009) | (.006) | (.006) | (.008) |
| % True Non-Zeroes | .996 | 1.000 | 1.000 | .497 | 1.000 | 1.000 | .494 | .496 | .930 | .992 | 1.000 | 1.000 | .507 | 1.000 | 1.000 | .498 | .501 | .996 |
|  | (.039) | (.000) | (.000) | (.096) | (.001) | (.000) | (.071) | (.074) | (.177) | (.023) | (.002) | (.000) | (.105) | (.004) | (.001) | (.074) | (.076) | (.042) |
| $MAD(\widehat{W})$ | .071 | .050 | .039 | .064 | .039 | .033 | .039 | .039 | .027 | .070 | .050 | .039 | .063 | .038 | .032 | .039 | .038 | .025 |
|  | (.005) | (.005) | (.004) | (.002) | (.002) | (.002) | (.000) | (.002) | (.005) | (.005) | (.004) | (.004) | (.002) | (.002) | (.002) | (.000) | (.002) | (.001) |
| $MAD(\widehat{\Pi})$ | .092 | .056 | .042 | .084 | .049 | .036 | .071 | .045 | .032 | .093 | .057 | .044 | .084 | .050 | .036 | .071 | .045 | .033 |
|  | (.008) | (.005) | (.004) | (.005) | (.003) | (.002) | (.004) | (.003) | (.002) | (.008) | (.005) | (.003) | (.005) | (.003) | (.002) | (.004) | (.002) | (.002) |
| $\hat{\rho}$ | .962 | .815 | .535 | .998 | .988 | .965 | 1.000 | .993 | .995 | .970 | .783 | .512 | .998 | .993 | .979 | 1.000 | .994 | .995 |
|  | (.112) | (.186) | (.231) | (.040) | (.066) | (.095) | (.000) | (.075) | (.054) | (.106) | (.199) | (.236) | (.041) | (.050) | (.077) | (.000) | (.069) | (.054) |
| $\hat{\beta}$ | .131 | .285 | .330 | .000 | .158 | .254 | .000 | .000 | .121 | .144 | .292 | .336 | .000 | .177 | .259 | .000 | .000 | .167 |
|  | (.079) | (.049) | (.038) | (.000) | (.053) | (.030) | (.000) | (.000) | (.066) | (.081) | (.049) | (.039) | (.000) | (.047) | (.030) | (.000) | (.000) | (.044) |
| $\hat{\gamma}$ | .996 | .998 | .968 | .000 | 1.000 | .999 | .000 | .000 | .863 | 1.000 | .995 | .942 | .000 | 1.000 | .999 | .000 | .000 | .992 |
|  | (.063) | (.015) | (.051) | (.000) | (.000) | (.007) | (.000) | (.000) | (.344) | (.000) | (.023) | (.066) | (.000) | (.000) | (.014) | (.000) | (.000) | (.089) |

**Notes**: These simulation results are based on the Adaptive Lasso algorithm, with penalization parameters chosen by BIC, under various true networks, network sizes and time periods T=50, 100 and 150. In all cases, 1000 Monte Carlo iterations were performed. The true parameters are rho-0=.3, beta-0=.4 and gamma-0=.5. The % of true zeroes refers to the proportion of true zero elements in the social interaction matrix that are estimated as smaller than .05. The % of true non-zeroes refers to the proportion of true elements greater than .3 in the social interaction matrix W that are estimated as non-zeros. The Mean Absolute Deviations are the mean absolute error of the estimated network compared to the true network for the social interaction matrix W and the reduced form matrix respectively. The recovered parameter are the estimated parameters averaged across iterations. All specifications include time and node fixed effects. Standard errors across iterations are in parentheses.

# Table A6: Simulation Results, OLS

| | A. Erdos-Renyi | | | | | | | | | B. Political party | | | | | | | | |
|---|---|---|---|---|---|---|---|---|---|---|---|---|---|---|---|---|---|---|
| | N = 15 | | | N = 30 | | | N = 50 | | | N = 15 | | | N = 30 | | | N = 50 | | |
| | T=500 | 1000 | 1500 | T=500 | 1000 | 1500 | T=500 | 1000 | 1500 | T=500 | 1000 | 1500 | T=500 | 1000 | 1500 | T=500 | 1000 | 1500 |
| % True Zeroes | .825 | .878 | .911 | .884 | .928 | .958 | .936 | .966 | .979 | .824 | .882 | .916 | .886 | .932 | .959 | .940 | .967 | .979 |
| | (.018) | (.020) | (.021) | (.007) | (.007) | (.006) | (.019) | (.003) | (.002) | (.020) | (.021) | (.020) | (.007) | (.007) | (.005) | (.003) | (.003) | (.002) |
| % True Non-Zeroes | 1.000 | 1.000 | 1.000 | 1.000 | 1.000 | 1.000 | .977 | 1.000 | 1.000 | 1.000 | 1.000 | 1.000 | 1.000 | 1.000 | 1.000 | 1.000 | 1.000 | 1.000 |
| | (.000) | (.000) | (.000) | (.000) | (.000) | (.000) | (.107) | (.000) | (.000) | (.000) | (.000) | (.000) | (.000) | (.000) | (.000) | (.000) | (.000) | (.000) |
| $MAD(\widehat{W})$ | .039 | .032 | .028 | .031 | .025 | .022 | .025 | .021 | .019 | .037 | .030 | .025 | .030 | .024 | .021 | .024 | .020 | .018 |
| | (.003) | (.004) | (.004) | (.001) | (.001) | (.001) | (.000) | (.000) | (.000) | (.003) | (.003) | (.003) | (.001) | (.001) | (.001) | (.000) | (.000) | (.000) |
| $MAD(\widehat{\Pi})$ | .038 | .027 | .022 | .039 | .027 | .022 | .040 | .027 | .022 | .038 | .027 | .022 | .038 | .027 | .022 | .039 | .027 | .022 |
| | (.002) | (.001) | (.001) | (.001) | (.001) | (.001) | (.001) | (.000) | (.000) | (.002) | (.001) | (.001) | (.001) | (.001) | (.001) | (.001) | (.000) | (.000) |
| $\widehat{\rho}$ | .488 | .381 | .362 | 1.000 | .981 | .617 | 1.000 | 1.000 | 1.000 | .465 | .396 | .382 | 1.000 | .953 | .590 | 1.000 | 1.000 | 1.000 |
| | (.114) | (.117) | (.108) | (.000) | (.051) | (.064) | (.000) | (.000) | (.000) | (.108) | (.104) | (.074) | (.000) | (.075) | (.062) | (.000) | (.000) | (.000) |
| $\widehat{\beta}$ | .394 | .398 | .399 | .343 | .372 | .400 | .256 | .350 | .353 | .397 | .401 | .399 | .346 | .382 | .400 | .280 | .350 | .350 |
| | (.017) | (.009) | (.006) | (.017) | (.025) | (.003) | (.064) | (.002) | (.013) | (.022) | (.014) | (.005) | (.013) | (.024) | (.000) | (.029) | (.002) | (.000) |
| $\widehat{\gamma}$ | .964 | .871 | .809 | 1.000 | 1.000 | .998 | .956 | 1.000 | 1.000 | .955 | .841 | .769 | 1.000 | 1.000 | .995 | 1.000 | 1.000 | 1.000 |
| | (.046) | (.079) | (.083) | (.000) | (.000) | (.009) | (.205) | (.000) | (.000) | (.053) | (.068) | (.046) | (.000) | (.000) | (.015) | (.000) | (.000) | (.000) |

**Notes**: These simulation results are based on OLS estimates, under various true networks, network sizes and time periods T=500, 1000 and 1500. In all cases, 1000 Monte Carlo iterations were performed. The true parameters are rho-0=.3, beta-0=.4 and gamma-0=.5. The % of true zeroes refers to the proportion of true zero elements in the social interaction matrix that are estimated as smaller than .05. The % of true non-zeroes refers to the proportion of true elements greater than .3 in the social interaction matrix that are estimated as non-zeros. The Mean Absolute Deviations are the mean absolute error of the estimated network compared to the true network for the social interaction matrix W and the reduced form matrix respectively. The recovered parameter are the estimated parameters averaged across iterations. All specifications include time and node fixed effects. Standard errors across iterations are in parentheses.

# Table A7: Summary Statistics, Tax Competition Application

| | Obs | Mean | SD | Min | q25 | Median | q75 | Max |
|---|---|---|---|---|---|---|---|---|
| **A. Besley and Case sample (1962-1988)** | | | | | | | | |
| State total tax per capita | 1296 | .371 | .266 | .036 | .145 | .300 | .530 | 1.345 |
| State income per capita | 1296 | 9.951 | 2.130 | 4.105 | 8.585 | 9.919 | 11.375 | 18.808 |
| Unemployment rate | 1296 | 5.885 | 2.242 | 1.800 | 4.200 | 5.500 | 7.000 | 18.000 |
| Proportion of young | 1296 | .234 | .033 | .160 | .210 | .240 | .260 | .310 |
| Proportion of elderly | 1296 | .106 | .020 | .040 | .090 | .110 | .120 | .190 |
| State governor's age | 1296 | 51.088 | 7.441 | 33.000 | 45.000 | 50.000 | 56.000 | 73.000 |
| Governor term limit dummy | 1296 | .258 | .438 | .000 | .000 | .000 | 1.000 | 1.000 |
| **B. Extended sample (1962-2014)** | | | | | | | | |
| State total tax per capita | 2688 | 0.983 | 0.803 | 0.036 | 0.037 | 0.813 | 1.557 | 4.298 |
| State income per capita | 2736 | 13.268 | 4.016 | 4.147 | 10.348 | 12.960 | 15.879 | 27.974 |
| Unemployment rate | 2688 | 5.764 | 2.026 | 1.800 | 4.300 | 5.400 | 6.800 | 17.800 |
| Proportion of young | 2688 | 0.236 | 0.033 | 0.170 | 0.210 | 0.230 | 0.260 | 0.340 |
| Proportion of elderly | 2688 | 0.117 | 0.023 | 0.050 | 0.100 | 0.120 | 0.130 | 0.190 |
| State governor's age | 2736 | 53.557 | 8.134 | 33.000 | 47.000 | 53.000 | 59.000 | 78.000 |
| Governor term limit dummy | 2638 | 0.249 | 0.433 | 0.000 | 0.000 | 0.000 | 0.000 | 1.000 |

**Notes**: Summary statistics of variables (in levels) used in subsequent regressions. Besley and Case sample runs from 1962 to 1988 and extended sample until 2014. State total tax per capita is the sum of sales, income and corporation tax in thousands of 1982 US dollars. State income per capita in thousands of 1982 US dollars. Proportion of young is the proportion of the population aged 65 or older. State governor's age in years. Governor term limit dummy is equal to 1 if governor faces term limits in the current mandate. Data sources: State total tax per capita, Census of Governments (1972, 1977, 1982, 1987, 1992-2016) and Annual Survey of Goverment Finances (all other years); Steta income per capita, Bureau of Economic Analysis; Unemployment rate, Bureau of Labor Statistics; Proportion of young (aged 5-17) and elderly (aged 65+), Census Population & Housing Data; State governor's age and political variables manually sourced from individual governor's webpages on Wikipedia.

## Table A8: Exogenous Social Effects

**Dependent variable: Change in per capital income and corporate taxes**
Coefficient estimates, standard errors in parentheses

| | (1) Initial | (2) OLS | (3) 2SLS: IVs are Characteristics of Neighbors | (4) 2SLS: IVs are Characteristics of Neighbors-of Neighbors |
|---|---|---|---|---|
| Economic Neighbors' tax change (t - [t-2]) | .645 | .145** | .332* | .608*** |
| | | (.072) | (.199) | (.220) |
| Economic Neighbors' income per capita | .090 | .098*** | .091*** | .080*** |
| | | (.011) | (.012) | (.014) |
| Economic Neighbors' unemployment rate | 37.200 | 9.899*** | 11.780*** | 13.714*** |
| | | (3.443) | (2.856) | (3.022) |
| Economic Neighbors' population aged 5-17 | 1378.1 | 376.2 | 478.5 | 596.6 |
| | | (399.0) | (414.2) | (401.7) |
| Economic Neighbors' population aged 65+ | -4304.5 | -842.8 | -769.7* | -641.3 |
| | | (504.3) | (450.2) | (468.6) |
| Economic Neighbors' governor age | -2.158 | -0.311 | -0.293 | -0.263 |
| | | (.281) | (.285) | (.294) |
| Period | | | 1962-2015 | |
| First Stage (F-stat) | | | 9.571 | 10.480 |
| Controls | Yes | Yes | Yes | Yes |
| State and Year Fixed Effects | Yes | Yes | Yes | Yes |
| Observations | 2,952 | 2,952 | 2,544 | 2,592 |

**Notes:** *** denotes significance at 1%, ** at 5%, and * at 10%. The sample covers 48 mainland US states running from 1962 to 2015. The dependent variable is the change in state i's total taxes per capita in year t. In all Columns, we penalize geographic neighbors in all Columns and allow for exogenous social effects. In OLS and IV regressions, the economic neighbors effect is calculated as the weighted average of economic neighbors' variables. OLS regressions estimates are shown in Column 2. Column 3 shows the 2SLS regression where each geographic neighbors' tax change is instrumented by lagged neighbor's state income per capita and unemployment rate. Column 4 shows a 2SLS regression where each geographic neighbors' tax change is instrumented by lagged neighbor-of-neighbor's state income per capita and unemployment rate. At the foot of Columns 3 and 4 we report the p-value on the F-statistic from the first stage of the null hypothesis that instruments are jointly equal to zero. All regressions control for state i's income per capita in 1982 US dollars, state i's unemployment rate, the proportion of young (aged 5-17) and elderly (aged 65+) in state i's population, and the state governor's age. All specifications include state and time fixed effects. With the exception of governor's age, all variables are differenced between period t and period t-2. Robust standard errors are reported in parentheses.

## Table A9: General Equilibrium Impacts of California Tax Rise

| | Geographic Neighbor Network | Economic Neighbor Network | Ratio |
|---|---|---|---|
| Average tax increase | 0.0038 | 0.0066 | 1.74 |
| Variance tax increase | 0.0160 | 0.0153 | 0.96 |
| Tax dispersion | 0.0053 | 0.0141 | 2.66 |
| States with tax increase | 48 | 48 | 1.00 |
| States with tax increase > 0.05% | 11 | 44 | 4.00 |
| States with tax increase > 0.5% | 5 | 11 | 2.20 |
| States with tax increase > 1% | 4 | 8 | 2.00 |
| States with tax increase > 2.5% | 1 | 3 | 3.00 |
| States with tax increase > 5% | 1 | 1 | 1.00 |

**Notes:** This shows the equilibrium impulse responses in taxes set in each state as a result of California increasing its tax change by 10%. The rho coefficient is derived from our preferred specification to estimate the economic network, where we penalize geographic neighbors to states, and allow for exogenous social effects (based on a sample of 48 mainland US states running from 1962 to 2015). We compare these derived tax changes under the identified economic network structure, relative to that assumed under a geographic neighbors structure. The final Column shows the ratio of the same statistic derived under each network.

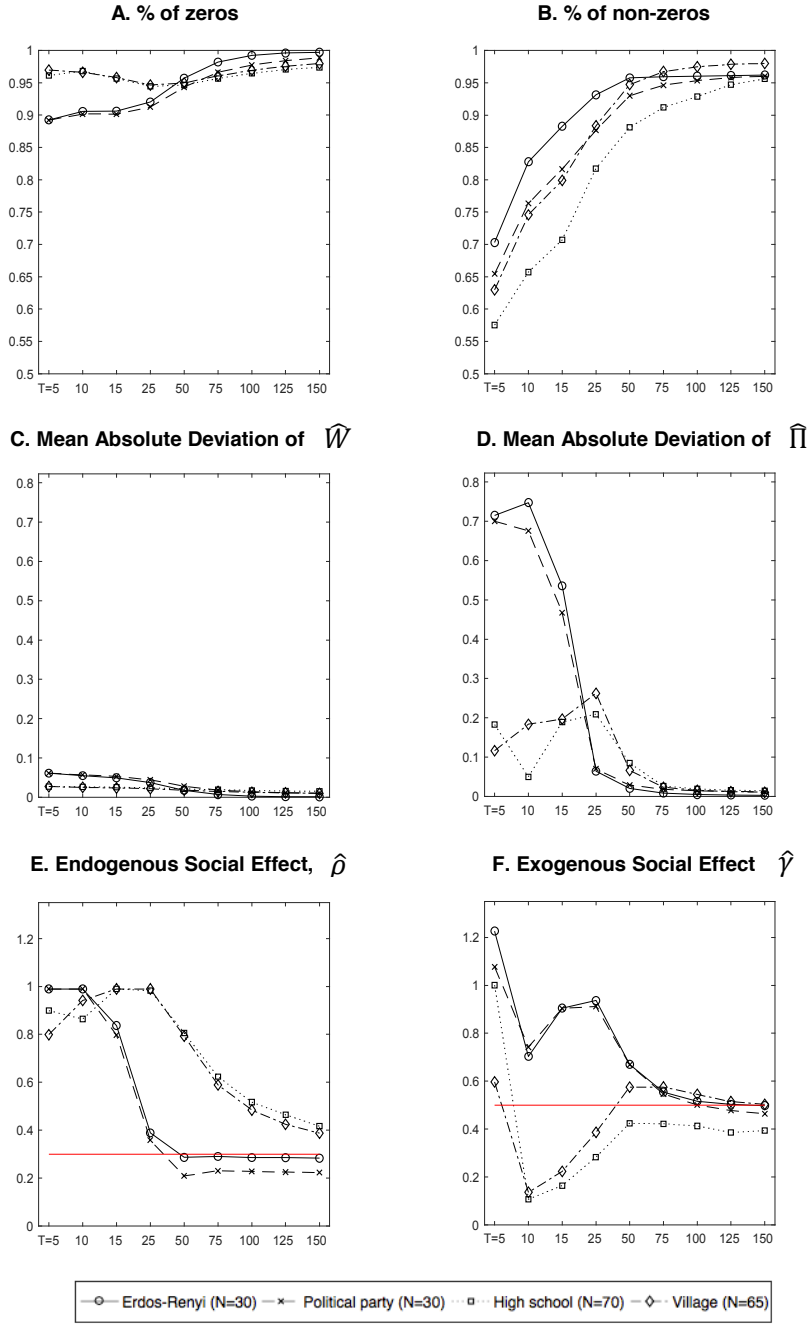

**Figure A1: Simulation Results, Adaptive Elastic Net GMM**

**Notes:** These simulation results are based on the Adaptive Elastic Net GMM algorithm, with penalization parameters chosen by BIC, under various true networks and time periods T=5, 10, 15, 25, 50, 75, 100, 125 and 150. In all cases, 1000 Monte Carlo iterations were performed. The true parameters are rho-0=.3, beta-0=.4 and gamma-0=.5. In Panel A, the % of zeroes refers to the proportion of true zero elements in the social interaction matrix that are estimated as smaller than .05. In Panel B, the % of non-zeros refers to the proportion of true elements greater than .3 in the social interaction matrix that are estimated as non-zeros. In Panels C and D, the Mean Absolute Deviations are the mean absolute error of the estimated network compared to the true network for the social interaction matrix W and the reduced form matrix respectively. In Panels E and F, the true parameter values are marked in the horizontal red lines. The recovered parameter are the estimated parameters averaged across iterations. All specifications include time and node fixed effects.

# Figure A2: Simulated and True Networks

**A. Erdos-Renyi**

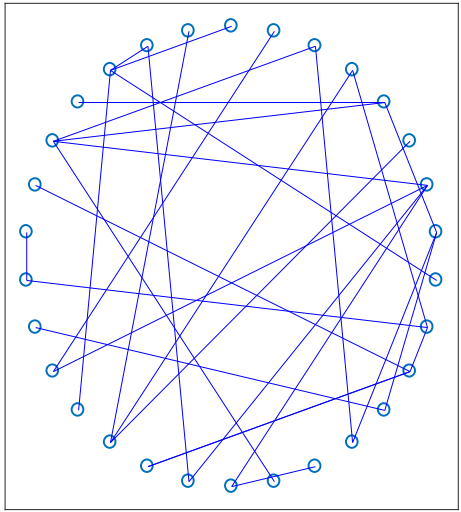

**B. Political Party**

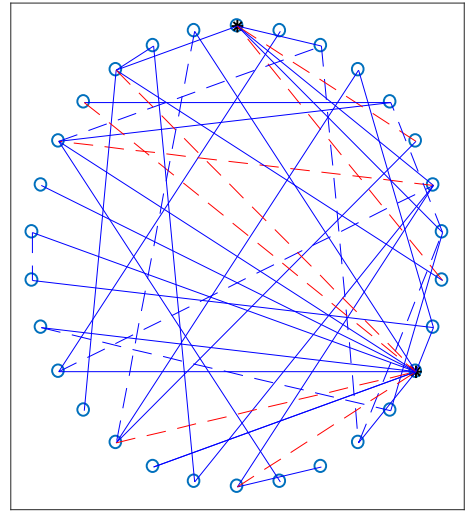

**C. High-school**

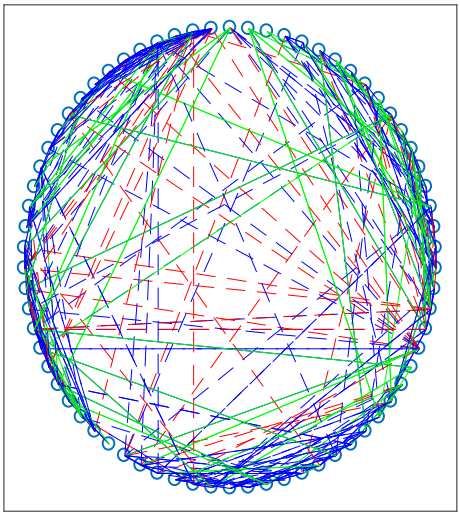

**D. Village**

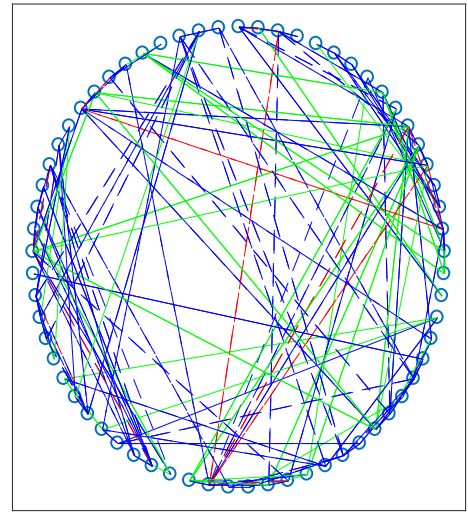

**Notes:** These simulation results are based on the Elastic Net algorithm, with penalization parameters chosen by BIC, under various true networks and time periods T=50, 100 and 150. In the two stylized networks (Erdos-Renyi and political party), we set N=30, and the real world networks, the high school friendship and village network are based on N=65 and 70 non-isolated nodes respectively. Party leaders in the political party network are marked in black in Panel B. In all cases, 1,000 Monte Carlo iterations were performed. The true parameters are rho-0=.3, beta-0=.4 and gamma-0=.5. All specifications include time and node fixed effects. Kept edges are depicted in blue: these links are estimated as non-zero in at least 5% of the iterations and are also non-zero in the true network. Added edges are depicted in green: these links are estimated as non-zero in at least 5% of the iterations but the edge is zero in the true network. Removed edges are depicted in red: these links are estimated as zero in at least 5% of the iterations but are non-zero in the true network. The figures further distinguish between strong and weak links: strong links are shown in solid edges (whose strength is greater than or equal to .3), and weak links are shown as dashed edges.

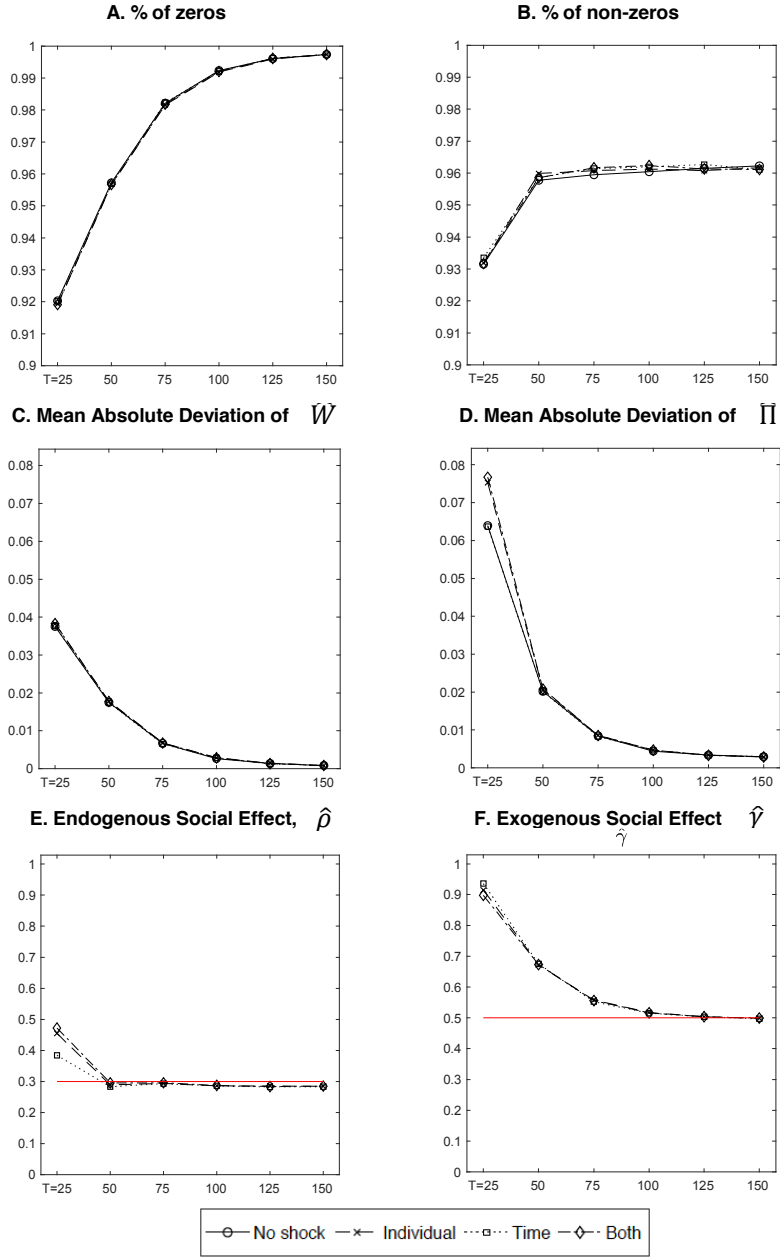

**Figure A3: Simulation Results, Adaptive Elastic Net GMM**
**Alternative Structures of Shocks**
**Erdos-Renyi graph (N=30)**

**Notes:** Simulations with common shocks between the exogenous variable and the error term: time-constant and varying at the individual level ("individual"), constant across individuals and varying over time ("time") and both types of shocks. These simulation results are based on the Adaptive Elastic Net GMM algorithm, with penalization parameters chosen by BIC, under various true networks and time periods T=25, 50, 100, 125 and 150. In all cases, 1000 Monte Carlo iterations were performed. The true parameters are rho-0=.3, beta-0=.4 and gamma-0=.5. In Panel A, the % of zeroes refers to the proportion of true zero elements in the social interaction matrix that are estimated as smaller than .05. In Panel B, the % of non-zeros refers to the proportion of true elements greater than .3 in the social interaction matrix that are estimated as non-zeros. In Panels C and D, the Mean Absolute Deviations are the mean absolute error of the estimated network compared to the true network for the social interaction matrix W and the reduced form matrix respectively. In Panels E and F, the true parameter values are marked in the horizontal red lines. The recovered parameter are the estimated parameters averaged across iterations. All specifications include time and node fixed effects.

**Figure A4: General Equilibrium Impacts of CA Tax Rise Shocks**
State's Reaction to 10% increase in CA taxes

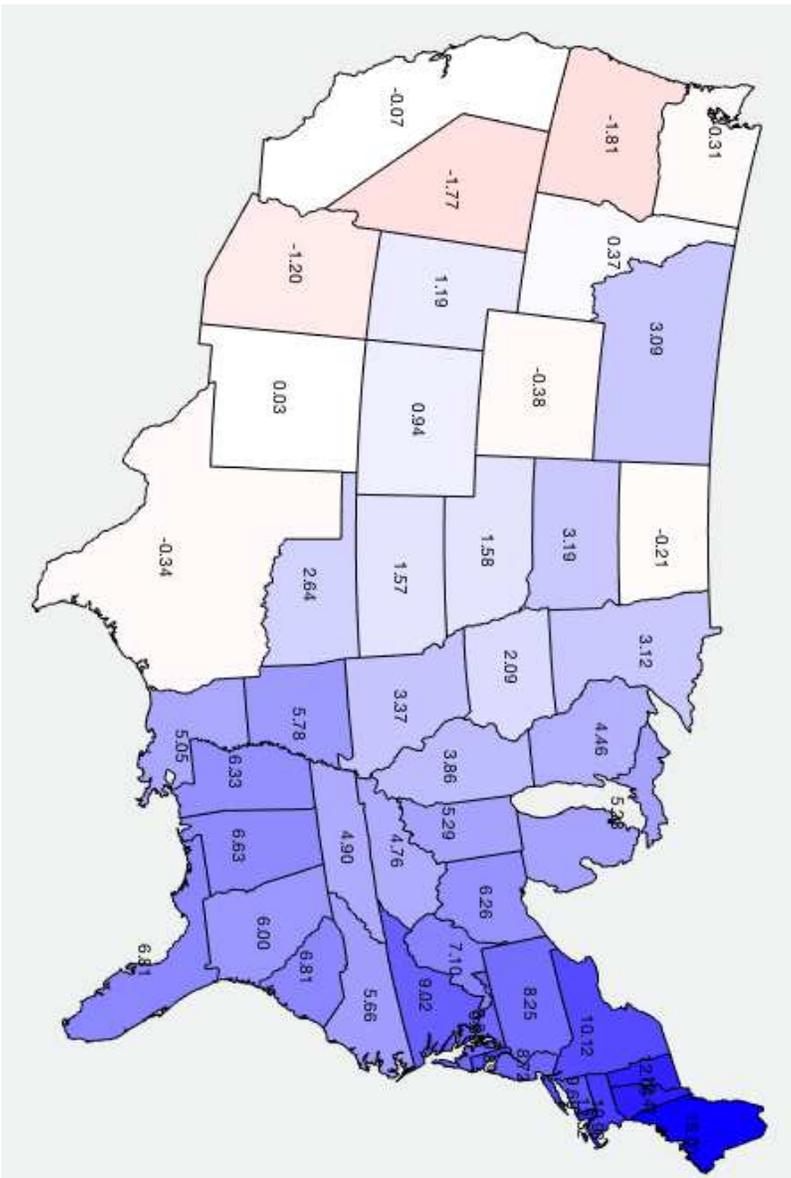

Log(equilibrium taxes under W-econ) - Log(equilibrium taxes under W-geo)

**Positive values indicate higher equilibrium taxes under Economic neighbors than geographic neighbors**
**Negative values indicate low equilibrium taxes under Economic neighbors than geographic neighbors**

**Notes:** This shows the equilibrium impulse responses in taxes set in each state as a result of California increasing its tax change by 10%. This is as derived from our preferred specification, where we penalize geographic neighbors to states, and allow for exogenous social effects. We compare these derived tax changes under the identified economic network structure, relative to that assumed under a geographic neighbors structure. We graph the log change in equilibrium taxes under economic neighbors, minus the log change in equilibrium taxes under geographic neighbors. Positive values (red shaded) states indicate higher equilibrium taxes under economic neighbors than geographic neighbors, and negative values (blue shaded) states indicate lower equilibrium taxes under economic neighbors than geographic neighbors.



\end{document}